\documentclass[11pt]{article}
\usepackage[margin=1in]{geometry}
\usepackage{paralist}
\RequirePackage{amsthm,amsmath,amsfonts,amssymb}
\RequirePackage{algorithm,algorithmic}
\usepackage{thmtools}
\RequirePackage{xcolor}
\RequirePackage{bbm}
\RequirePackage{gensymb}
\RequirePackage{centernot}
\RequirePackage{bm}
\RequirePackage{enumerate}
\RequirePackage{enumitem}
\RequirePackage{graphicx}
\RequirePackage{multirow}
\RequirePackage{booktabs}
\RequirePackage{mathtools}
\RequirePackage[authoryear,round]{natbib}

\theoremstyle{plain}
\newtheorem{theorem}{Theorem}
\newtheorem*{theorem*}{Theorem}
\newtheorem{corollary}{Corollary}

\newtheorem*{claim*}{Claim}
\newtheorem{lemma}{Lemma}
\newtheorem*{lemma*}{Lemma}
\newtheorem{proposition}{Proposition}
\newtheorem*{proposition*}{Proposition}

\theoremstyle{remark}
\newtheorem{definition}{Definition}
\newtheorem{remark}{Remark}
\newtheorem{assumption}{Assumption}
\newtheorem{example}{Example}
\renewcommand\thmcontinues[1]{Continued}

\usepackage{crossreftools}
\makeatletter
\newcommand{\optionaldesc}[2]{%
  \phantomsection
  #1\protected@edef\@currentlabel{#1}\label{#2}%
}
\makeatother

\usepackage{hyperref}
\usepackage{caption} 
\newcommand{\bern}{\textnormal{Bern}}

\newcommand{\eff}{\textnormal{eff}}
\newcommand{\expit}{\textnormal{expit}}

\newcommand{\rct}{\textnormal{rct}}
\newcommand{\obs}{\textnormal{obs}}
\newcommand{\tgt}{\textnormal{tgt}}
\newcommand{\cov}{\textnormal{cov}}

\newcommand{\indic}{\bm{1}}

\newcommand{\ba}{\textnormal{ba}}
\newcommand{\ca}{\mathcal{A}}

\newcommand{\cf}{\mathcal{F}}
\newcommand{\ch}{\mathcal{H}}
\newcommand{\ci}{\mathcal{I}}
\newcommand{\cj}{\mathcal{J}}

\newcommand{\cn}{\mathcal{N}}

\newcommand{\cp}{\mathcal{P}}

\newcommand{\crr}{\mathcal{R}}
\newcommand{\cs}{\mathcal{S}}
\newcommand{\ct}{\mathcal{T}}
\newcommand{\cu}{\mathcal{U}}
\newcommand{\cw}{\mathcal{W}}
\newcommand{\cv}{\textnormal{cv}}
\newcommand{\cx}{\mathcal{X}}

\newcommand{\os}{\textnormal{os}}

\renewcommand{\leq}{\leqslant}

\renewcommand{\geq}{\geqslant}
\renewcommand{\Pr}{\textnormal{Pr}}
\newcommand{\dunif}{\mathbb{U}}

\newcommand{\e}{\mathbb{E}}
\newcommand{\real}{\mathbb{R}}

\newcommand{\cm}{\mathcal{M}}

\newcommand{\var}{\mathrm{Var}}
\newcommand{\tr}{\textnormal{tr}}

\DeclareMathOperator*{\argmin}{arg\,min}

\title{Efficient estimation and data fusion under general semiparametric restrictions on outcome mean functions}
\author{Harrison H. Li, Stanford University}
\date{}

\begin{document}
\maketitle
\begin{abstract}
We provide a novel characterization of semiparametric efficiency in a generic supervised learning setting where the outcome mean function --- defined as the conditional expectation of the outcome of interest given the other observed variables ---  is restricted to lie in some known semiparametric function class.
The primary motivation is causal inference where a researcher running a randomized controlled trial often has access to an auxiliary observational dataset that is confounded or otherwise biased for estimating causal effects.
Prior work has imposed various bespoke assumptions on this bias in an attempt to improve precision via data fusion.
We show how many of these assumptions can be formulated as restrictions on the outcome mean function in the concatenation of the experimental and observational datasets.
Then our theory provides a unified framework to maximally leverage such restrictions for precision gain by constructing efficient estimators in all of these settings as well as in a wide range of others that future investigators might be interested in.
For example,
when the observational dataset is subject to outcome-mediated selection bias,
we show our novel efficient estimator dominates an existing control variate approach
both asymptotically and in 
numerical studies.
\end{abstract}

\section{Introduction}
Consider a generic supervised learning setup with independent and identically distributed (i.i.d.) observations $W_i=(Y_i,R_i)_{i=1}^N$,
where the outcome $Y_i \in \real$ is scalar.
Omitting the observation index $i$ for brevity,
let the \emph{outcome mean function}
\[
m(r) = \e[Y \mid R=r]
\]
be the conditional expectation of the outcome given the covariates.
We seek a novel characterization of the \emph{semiparametric efficiency bound} --- the smallest asymptotic variance attainable by any regular and asymptotically linear (RAL) estimator of any smooth (pathwise differentiable) and identified estimand $\tau \in \real^d$ --- 
in a class of semiparametric models $\cp_{\cm}$ that restrict the outcome mean function $m(\cdot)$ to lie in a known collection $\cm$ of square integrable real-valued functions on the covariate space $\crr$.
This facilitates the construction of efficient one-step estimators in the model $\cp_{\cm}$.

A key motivation for this problem comes from data fusion problems in causal inference.
The economic and ethical challenges of running a large randomized controlled trial (RCT) has spurred significant work on combining data from RCT's with much larger and easily obtained observational datasets to improve the precision of causal estimates.
However, due to factors like unobserved confounding or selection,
causal effects are generally not identifiable from observational data where treatment is not randomized.
Then, naive approaches that simply combine observations from an RCT and an observational study into a single \emph{fused dataset} will lead to biased causal estimates,
even (indeed, most saliently) with a massive observational dataset~\citep{bareinboim2016causal}.

On the other hand,
without any assumptions about the biases in the observational dataset,
one typically cannot use the outcomes in the observational data to improve the asymptotic precision of causal estimates.
Intuitively,
this is because without any linkage between the outcome distributions in the two datasets,
the outcomes in the observational dataset cannot possibly provide any statistical information about the outcome distribution in the RCT,
which typically identifies causal parameters of interest.
For example,
the semiparametric efficiency bound for the average treatment effect $\tau_{\rct}$ in the population in which the RCT was conducted can be attained by estimators such as the augmented inverse propensity weighting (AIPW) estimator~\citep{robins_1994_estimation, tsiatis2006semiparametric, rothe_value_2016} that only use observations from the RCT~\citep{li_improving_2023,liao_transfer_2023}.
Even if we wish to estimate a ``transported" average treatment effect on the population defined by the covariate distribution in the \emph{observational} dataset,
the efficiency bound can be obtained by estimators using just the RCT observations and the covariates in the observational dataset.
The outcomes in the observational dataset cannot provide further asymptotic variance reductions without using non-regular estimators~\citep{dahabreh_generalizing_2019,dahabreh_extending_2020,
lee_improving_2023,li_note_2023}.


Many authors have thus chosen another route: make structural assumptions linking the outcome distributions in the two datasets to enable precision gains.
Through various examples in Section~\ref{sec:examples},
we show that such assumptions often correspond precisely to imposing a model $\cp_{\cm}$
on the fused dataset,
for an appropriate choice of $\cm$.
In Section~\ref{sec:semiparametric_efficiency}
we present our main theoretical result,
Theorem~\ref{thm:semiparametric_tangent_space},
which characterizes semiparametric efficiency in a general model of the form $\cp_{\cm}$.
We show how to use this result to construct efficient one-step estimators in Section~\ref{sec:estimation_inference}.
Such estimators depend on first-stage estimation of nuisance parameters including
the outcome mean function $m(\cdot)$.
As well-established in the literature on double/debiased machine learning,
under appropriate regularity conditions and cross-fitting,
the estimators can attain the efficiency bound even if such nuisance functions are estimated at slower than parametric rates.

In Section~\ref{sec:applications},
we provide specific efficiency bounds and estimators in two of the particular data fusion settings from the literature
involving a linear confounding bias (see Example~\ref{ex:parametric_confounding_bias} below)
and outcome-mediated selection bias (see Example~\ref{ex:outcome_selection_bias}).
To our knowledge,
efficient estimators have not previously been derived in either of these settings.
We provide numerical simulations and a real data example based on the Tennessee Student Teacher Achievement Ratio (STAR) study in Sections~\ref{sec:simulations_data_fusion} and~\ref{sec:tennessee_star}, respectively,
to illustrate the finite-sample performance of our novel estimators,
which provably dominate the existing inefficient proposals asymptotically.
We conclude with a brief discussion of related work in Section~\ref{sec:discussion_data_fusion}.

\subsection{Examples}
\label{sec:examples}

We now provide several motivating examples from the literature of the semiparametric model $\cp_{\cm}$ introduced previously.
Recall that each such model restricts the joint distribution of $W$ so that the outcome mean function $m(\cdot)$ belongs to a particular known function collection $\cm$.

\begin{example}[Restricted moment model]
\label{ex:restricted_moment_model}
If $\cm_1$ is a smooth finite-dimensional parametric class,
i.e. $\cm_1=\{r \mapsto \mu(r;\beta), \beta \in \real^q\}$,
the model $\cp_{\cm_1}$ is known as the ``restricted moment model,"
and is a textbook motivating example of semiparametric theory (e.g., Ch. 4 of~\citet{tsiatis2006semiparametric}).
\end{example}

The remaining examples are data fusion settings in causal inference,
where a model of the form $\cp_{\cm}$ applies to the fused observational and experimental dataset.
The fused dataset has $R=(S,Z,X)$ for
$S$ a binary indicator of which dataset the observation is taken from ($S=1$ denotes experimental, $S=0$ denotes observational),
$Z$ a binary indicator of treatment status,
and $X \in \cx$ a vector of observed covariates.
Examples with more than two datasets and/or treatment statuses follow easily but are notationally cumbersome.
For $(s,z) \in \{0,1\}^2$ we write $m_{sz}(x)$ as a shorthand for $m(s,z,x)$.
We also let $\ch_V$ denote the set of square integrable functions of any variable or set of variables $V$,
which contains $\ch_V^0$,
the set of such functions with mean zero under $P^*$,
the true distribution of $W$.
For example,
explicitly we have $\ch_W = \{f: \cw \mapsto \real \mid \e[f^2(W)] < \infty\}$;
unless otherwise noted,
all expectations are to be taken with respect to $P^*$,
i.e. $\e[f(W)] \equiv \int f(w)dP^*(w)$.
We recall that $\ch_W$ is a Hilbert space equipped with the inner product $\langle f,g\rangle = \e[f(W)g(W)]$,
which does not distinguish between functions equal almost-surely under $P^*$.

\begin{example}[Mean-exchangeable controls, \citet{li_improving_2023}]
\label{ex:mean_exchangeable_controls}
Suppose the observational dataset consists only of untreated subjects or ``controls" (i.e., $\Pr(Z=0 \mid S=0)=1$)
and assume that these controls are mean-exchangeable with those in the experimental dataset,
i.e., $m_{10}(x)=m_{00}(x) \forall x \in \cx$.
This corresponds to the model $\cp_{\cm_2}$ where
\begin{equation}
\label{eq:mean_exchangeable_controls}
\cm_2 = \{r \mapsto zm_1(x)+(1-z)m_0(x) \mid m_0,m_1 \in \ch_X\}.
\end{equation}
\end{example}

\begin{example}[Parametric confounding bias and CATE, \cite{yang2024datafusion}]
\label{ex:parametric_confounding_cate}
Assuming treatment is randomized in the experimental dataset,
the difference $\tau(x;m) = m_{11}(x)-m_{10}(x)$ between the conditional average outcomes in the treated and control groups in the experimental dataset given $X=x$
is well known to have a causal interpretation in the Neyman-Rubin potential outcomes framework as the conditional average treatment effect (CATE).
The analogous difference $m_{01}(x)-m_{00}(x)$ in the observational dataset,
on the other hand,
may not have a causal interpretation due to unobserved confounders.~\citet{yang2024datafusion} consider a model where both the CATE $\tau(x)$ and the differences in these differences $\lambda(x;m) = (m_{01}(x)-m_{00}(x)) - \tau(x)$,
called the \emph{confounding function},
lie in smooth finite-dimensional parametric classes.
This corresponds to the model $\cp_{\cm_3}$ for
\begin{equation}
\label{eq:parametric_confounding_cate}
\cm_3 = \{m \in \ch_R \mid \tau(x;m) = \mu(x;\beta), \lambda(x;m) = \phi(x;\theta), \beta \in \real^p, \theta \in \real^q\}.
\end{equation}
\end{example}

\begin{example}[Linear confounding bias, \citet{kallus_removing_2018}]
\label{ex:parametric_confounding_bias}
Continuing the notation of Example~\ref{ex:parametric_confounding_cate},
we place no restrictions on the CATE $\tau(x;m)$ while assuming the confounding function is linear in a $q$-dimensional basis expansion $\psi$.
This is the model $\cp_{\cm_4}$ for
\begin{equation}
\label{eq:parametric_confounding_bias}
\cm_4 = \{m \in \ch_R \mid \lambda(x;m) = \psi(x)^{\top}\theta, \theta \in \real^q\}.
\end{equation}
\end{example}

\begin{example}[Outcome-mediated selection bias, \citet{guo2022multi}]
\label{ex:outcome_selection_bias}
Suppose that $Y \in \{0,1\}$
and that the observational dataset is subject to ``outcome-mediated selection bias."
This means that the observational dataset is subject to selection that depends only on $Y$ (and not on $(Z,X)$).
For example,
``controls" with $Y=0$ are omitted in the observational dataset more often than ``cases" with $Y=1$,
and omission occurs independently of $(Z,X)$ given $Y$.~\citet{guo2022multi} show in their Proposition 4.1 that outcome-mediated selection bias implies
\begin{equation}
\label{eq:odds_ratio}
OR_1(x) = OR_0(x), \qquad OR_s(x) := \frac{m_{s1}(x)/(1-m_{s1}(x))}{m_{s0}(x)/(1-m_{s0}(x))}, \quad s=0,1.
\end{equation}
Thus, outcome-mediated selection bias implies the model $\cp_{\cm_5}$ for
\begin{equation}
\label{eq:outcome_selection_bias}
\cm_5 = \{m \in \ch_R \mid \eqref{eq:odds_ratio} \text{ holds}, \  \epsilon < m(r) < 1-\epsilon \forall r \in \crr\}.
\end{equation}
\end{example}
We will show that our results can recover various semiparametric efficiency bounds and efficient estimators in the models of Examples~\ref{ex:restricted_moment_model},~\ref{ex:mean_exchangeable_controls}, and~\ref{ex:parametric_confounding_cate}
that were previously derived in the cited works.
The model in Example~\ref{ex:parametric_confounding_bias} was studied by~\citet{kallus_removing_2018} but for estimating the CATE rather than finite dimensional estimands.
Turning to Example~\ref{ex:outcome_selection_bias},~\citet{guo2022multi} propose a control variate approach to leverage~\eqref{eq:odds_ratio} for variance reduction.
In general, however, these estimators are inefficient,
and do not come with clear guidance on how to choose control variates to minimize variance.
Thus, we use our results to derive novel efficiency bounds and efficient estimators for various average treatment effects in the models $\cp_{\cm_4}$ and $\cp_{\cm_5}$ in Examples~\ref{ex:parametric_confounding_bias} and~\ref{ex:outcome_selection_bias}.
The estimators are given explicitly in Section~\ref{sec:applications}.

\section{Semiparametric efficiency bounds}
\label{sec:semiparametric_efficiency}

We return to the generic supervised learning setting of the introduction and present our main result characterizing semiparametric efficiency bounds for the models $\cp_{\cm}$.
Formally, the model $\cp_{\cm}$ is a collection of distributions $P$ on $W$ with outcome mean functions in $\cm$ assumed to contain the true data-generating process $P^*$.
The semiparametric efficiency bound for an a pathwise differentiable estimand $\tau=\tau(P)$ at its true value $\tau^* = \tau(P^*)$
is defined as the smallest asymptotic variance attainable by any RAL estimator of $\tau$.
The reader is encouraged to reference texts on semiparametric theory~\citep{bickel1993efficient,vandervaart2000asymptotic,tsiatis2006semiparametric} for further background and technical details;
we present the main concepts necessary for our discussion below.

By definition of asymptotic linearity, any RAL estimator of $\tau \in \real^d$ must satisfy 
the asymptotic expansion
\begin{equation}
\label{eq:AL}
\sqrt{N}(\hat{\tau}-\tau^*) = \frac{1}{\sqrt{N}} \sum_{i=1}^N \varphi(W_i;\tau^*,\eta^*) + o_p(1)
\end{equation}
as $N \rightarrow \infty$
where $\varphi(\cdot;\tau^*,\eta^*) \in (\ch_W^0)^d$ is known as the \emph{influence function} of $\tau$.
The notation $\varphi(\cdot;\tau^*,\eta^*)$ emphasizes that in general,
the influence function $\varphi$ can depend on both the estimand $\tau$ of interest
and nuisance parameters $\eta$ of interest,
which are often infinite dimensional functionals,
like the outcome mean function $m$,
of the distribution $P$ in the model $\cp_{\cm}$. 
Regularity further restricts the space of possible influence functions of RAL estimators of $\tau$ under the model $\cp_{\cm}$ within $\ch_W^0$,
as discussed below.
The \emph{efficient influence function} $\varphi_{\eff}$ is the unique influence function in this space with the smallest variance
(with respect to the Loewner ordering when $d>1$).
Then the semiparametric efficiency bound is $V_{\eff} = \e[\varphi_{\eff}(W)^{\otimes 2}]$,
where $v^{\otimes 2} = vv^{\top}$ for any vector $v$.

Theorem~\ref{thm:influence_functions} below
is a standard result that characterizes the space of all influence functions of RAL estimators of any pathwise differentiable estimand $\tau \in \real^d$ in the model $\cp_{\cm}$
as an affine translation of a space $(\ct_{\cm}^d)^{\perp}$,
the orthogonal complement of the \emph{semiparametric tangent space} $\ct_{\cm}^d$ in $(\ch_W^0)^d$.
The semiparametric tangent space $\ct_{\cm}^d$
is defined as the closure (with respect to $\ch_W^0$) of the set of all $d$-dimensional linear combinations of score functions of all smooth finite-dimensional parametric submodels $P_{\gamma} \subseteq \cp_{\cm}$.
A \emph{parametric submodel} $P_{\gamma}$ of $\cp_{\cm}$ refers to a subset of distributions in $\cp_{\cm}$ indexed smoothly
by a finite-dimensional parameter $\gamma$ in an open neighborhood of 0,
with $P_0=P^*$.
See~\citet{newey1990semiparametric} for the technical details on smoothness requirements for parametric submodels.
The score function of a parametric submodel $P_{\gamma}$ is the derivative of the log likelihood with respect to $\gamma$,
evaluated at the truth $\gamma=0$.

\begin{theorem}[Theorem 4.3,~\citet{tsiatis2006semiparametric}]
\label{thm:influence_functions}
Suppose a RAL estimator of a pathwise differentiable estimand $\tau \in \real^d$ in the model $\cp_{\cm}$ exists with influence function $\varphi_0=\varphi_0(\cdot;\tau^*,\eta^*)$.
Further assume the semiparametric tangent space $\ct_{\cm}^d$ for the model $\cp_{\cm}$ is a closed linear subspace of $(\ch_W^0)^d$.
Then:
\begin{enumerate}
    \item The influence function $\varphi=\varphi(\cdot;\tau^*,
\eta^*)$ of any RAL estimator of $\tau$ must satisfy $\varphi-\varphi_0 \in (\ct_{\cm}^d)^{\perp}$,
where $(\ct_{\cm}^d)^{\perp}$ is the orthogonal complement of $\ct_{\cm}^d$ in $(\ch_W^0)^d$.
\item The efficient influence function is uniquely given by
\begin{equation}
\label{eq:eif_raw}
\varphi_{\eff}=\Pi(\varphi_0;\ct_{\cm}^d)= \varphi_0  - \argmin_{g \in (\ct_{\cm}^d)^{\perp}} \e[\|\varphi_0(W;\tau^*,\eta^*)-g(W)\|^2]
\end{equation}
where $\Pi(f;\cs)$ denotes the projection of $f$ onto the subspace $\cs \subseteq (\ch_W^0)^d$.
\end{enumerate}
\end{theorem}


Theorem~\ref{thm:influence_functions} does not make use of the fact that $\cp_{\cm}$ is a model that specifically restricts the outcome mean function $m$. 
Our contribution, 
in Theorem~\ref{thm:semiparametric_tangent_space} below,
is to specify the semiparametric tangent space for a wide range of sets $\cm$
in terms of another space we call the \emph{outcome mean function tangent space} $\cs_{\cm}$.
Note it suffices to consider $d=1$ hereafter,
since $\ct_{\cm}^d$ is just the $d$-fold Cartesian product of $\ct_{\cm}^1 \equiv \ct_{\cm}$.

\begin{definition}
\label{def:mean_function_submodel}
An \emph{outcome mean function parametric submodel} for $\cm \subseteq \ch_R$ is any collection of mean functions $\{m(\cdot;\gamma) \mid \gamma \in \cu\}$ contained in $\cm$ and indexed by $\gamma$ in an open neighborhood $U$ of $0 \in \real^s$ such that $m(r;0)=m^*(r)$ and $\partial m(r;\gamma)/\partial \gamma$ exists for $P^*$-almost every $r \in \crr$ and all $\gamma \in U$
with $\partial m(\cdot;\gamma)/\partial \gamma|_{\gamma=0} \in (\ch_R)^s$.
The \emph{outcome mean function tangent space} $\cs_{\cm}$ for $\cm$ is then the closure (in the Hilbert space $\ch_R$) of the function collection 
\[
r \mapsto c^{\top}\frac{\partial m(r;\gamma)}{\partial \gamma}\Big|_{\gamma=0} 
\]
indexed by $c \in \real^s$,
positive integers $s \geq 1$,
and $s$-dimensional parametric submodels $\{m(\cdot;\gamma)\}$.
\end{definition}


\begin{theorem}
\label{thm:semiparametric_tangent_space}
Suppose $\{W_i=(R_i,Y_i)\}_{i=1}^N$ are i.i.d. from some distribution $P^* \in \cp_{\cm}$ with $Y_i \in \real$ where $\cm$ is such that the corresponding outcome mean function tangent space $\cs_{\cm}^1 \equiv \cs_{\cm}$ is a linear subspace of $\ch_R$.
Then $\ct_{\cm}^{\perp}$,
the orthogonal complement in $\ch_W^0$ of the semiparametric tangent space $\ct_{\cm}$ for the model $\cp_{\cm}$,
satisfies
\begin{align}
\ct_{\cm}^{\perp} & = \{w \mapsto h(r)(y-m^*(r)) \mid h \in \cs_{\cm}^{\perp}\} \label{eq:T_M_perp_generalized}
\end{align}
where $\cs_{\cm}^{\perp}$ is the orthogonal complement of $\cs_{\cm}$ in $\ch_R$.
\end{theorem}
\begin{proof}
See Appendix~\ref{app:semiparametric_tangent_space_proof}.
\end{proof}

Putting together Theorem~\ref{thm:semiparametric_tangent_space} with~\eqref{eq:eif_raw} suggests the following three-step outline for computing the EIF of an estimand $\tau$ in the model $\cp_{\cm}$:
\begin{enumerate}
\item Find an influence function $\varphi_0$ of a RAL estimator of $\tau$ in $\cp_{\cm}$
\item Compute the space $\cs_{\cm}$
\item Compute the orthogonal complement $\cs_{\cm}^{\perp}$,
which specifies $\ct_{\cm}^{\perp}$ per Theorem~\ref{thm:semiparametric_tangent_space},
enabling computation of the projection~\eqref{eq:eif_raw}.
\end{enumerate}
We now provide some discussion for each of these steps,
using them to compute EIF's in the outcome-mediated selection bias setting (the model $\cp_{\cm_5}$).
In Appendix~\ref{app:eif_derivations},
we similarly use this framework to derive EIF's in the model $\cp_{\cm_4}$ with linear confounding bias
(the final result is given in Section~\ref{sec:applications}).
We also re-derive known EIF's from the literature in the models $\cp_{\cm_1}$, $\cp_{\cm_2}$, and $\cp_{\cm_3}$  in Appendix~\ref{app:eif_derivations}. 

\subsection{Initial influence function}
\label{sec:varphi_0}
A natural choice for the initial influence function $\varphi_0$ in step 1 above,
so long as the estimand $\tau$ is identified in the nonparametric model $\cp_{\ch_R}$ that does not place any structural restrictions on the data generating process,
is the \emph{canonical gradient} for the estimand $\tau$.
The canoncial gradient is the only possible influence function of a RAL estimator in $\cp_{\ch_R}$
(hence also the EIF in $\cp_{\ch_R}$).
Since $\cp_{\cm} \subseteq \cp_{\ch_R}$ for any $\cm$,
the canonical gradient is also a possible influence function of a RAL estimator in $\cp_{\cm}$.
The restriction that $m \in \cm$,
however, generally restricts the semiparametric tangent space $\ct_{\cm}$,
thereby making the orthogonal complement $\ct_{\cm}^{\perp}$ nontrivial which means that more influence functions are available in the model $\cp_{\cm}$,
among which the minimal variance is attained by~\eqref{eq:eif_raw}.
The canonical gradient is readily available in the literature for many common estimands like average treatment effects and densities.
When it is not, it can often be derived using various techniques like ``point mass contamination" and calculus rules~\citep{kennedy2024semiparametric}.

\begin{table}[!tb]
\centering
\caption{Some notation used throughout the paper for functionals in the data fusion examples.
We use an asterisk superscript to indicate the value of the functional under the true data-generating process $P^*$.}
\label{table:functions}
\begin{tabular}{ccc}
\hline
Notation & Definition & Description \\
\midrule
$p(x)$ & $\Pr(S=1 \mid X=x)$ & RCT selection probability function \\
$e(x)$ & $\Pr(Z=1 \mid S=1, X=x)$ & RCT propensity score \\
$q(x)$ & $\Pr(Z=1 \mid S=0, X=x)$ & Observational propensity score \\
$\rho$ & $\Pr(S=1)$ & RCT sample size proportion \\
$V(s,z,x)$ or $V_{sz}(x)$ & $\var(Y \mid S=s,Z=z,X=x)$ & Outcome variance function \\
\bottomrule
\end{tabular}
\end{table}

In Proposition~\ref{prop:efficiency_bounds_0} we provide the canonical gradients for three estimands $\tau_{\rct}$, $\tau_{\obs}$, and $\tau_{\tgt}$
\begin{equation}
\label{eq:tau_id}
\tau_{\rct} = \e[\tau(x;m) \mid S=1]; \quad \tau_{\obs} = \e[\tau(x;m) \mid S=0]; \quad \tau_{\tgt} = \e[\tau(x;m)]
\end{equation}
in the data fusion setting where $R=(S,Z,X)$.
Recall that $\tau(x;m)=m_{11}(x)-m_{10}(x)$ is the conditional average treatment effect when treatment is randomized (conditional on $X$) in the experiment.
In that case,
$\tau_{\rct}$ identifies the average treatment effect for the population of subjects from which the experimental dataset is sampled.
Under an additional mean-exchangeability assumption on the potential outcomes,~\citet{dahabreh_generalizing_2019,dahabreh_extending_2020} show that $\tau_{\obs}$ and $\tau_{\tgt}$ identify the average treatment effect in the population defined by the observational dataset and the fused dataset,
respectively,
as long as the RCT selection probability $p(x)=\Pr(S=1 \mid X=x)$ is bounded away from 0.
We do not discuss specific identification assumptions further,
as our focus is on improving precision via restrictions on the outcome mean function $m$ for already identified estimands.

\begin{proposition}
\label{prop:efficiency_bounds_0}
The efficient influence functions $\varphi_{0,\rct}$, $\varphi_{0,\obs}$, and $\varphi_{0,\tgt}$ for the estimands $\tau_{\rct}$, $\tau_{\obs}$, and $\tau_{\tgt}$ defined in~\eqref{eq:tau_id} at the true observed data distribution $P^*$ in the nonparametric model $\cp_{\ch_R}$ are given by
\begin{align}
\varphi_{0,\rct}(w;\tau_{\rct}^*, \eta^*) & = \frac{s(\Delta(z,x,y;\eta^*)+m_{11}^*(x)-m_{10}^*(x)-\tau_{\rct}^*)}{\rho^*}\label{eq:tau_exp_bound_0} \\
\varphi_{0,\obs}(w;\tau_{\obs}^*, \eta^*) & = \frac{s(1-p^*(x))\Delta(z,x,y;\eta^*)}{p^*(x)(1-\rho^*)}+\frac{(1-s)(m_{11}^*(x)-m_{10}^*(x)-\tau_{\obs}^*)}{1-\rho^*} \label{eq:tau_obs_bound_0}\\
\varphi_{0,\tgt}(w;\tau_{\tgt}^*, \eta^*) & = \frac{s}{p^*(x)}\Delta(z,x,y;\eta^*)+m_{11}^*(x)-m_{10}^*(x)-\tau_{\tgt}^* \label{eq:tau_tgt_bound_0}
\end{align}
where the components of $\eta^* = \eta^*(\cdot)= (m^*(1,\cdot,\cdot),e^*(\cdot),p^*(\cdot))$ are defined in Table~\ref{table:functions} and $\Delta(z,x,y;\eta^*)$ is the inverse propensity weighted difference
\begin{equation}
\label{eq:Delta}
\Delta(z,x,y;\eta^*) = \frac{z(y-m_{11}^*(x))}{e^*(x)}-\frac{(1-z)(y-m_{10}^*(x))}{1-e^*(x)}.
\end{equation}
\end{proposition}
\begin{proof}
Equation~\eqref{eq:tau_exp_bound_0} is from Section 2.2 of~\citet{li_improving_2023}.
Equations~\eqref{eq:tau_obs_bound_0} and~\eqref{eq:tau_tgt_bound_0} are from Theorem 2.3 of~\citet{li_note_2023}.
\end{proof}

As noted in the introduction,
the AIPW estimator of~\citet{robins_1994_estimation},
among many other popular estimators for $\tau_{\rct}$,
attains the influence function $\varphi_{0,\rct}$ without using any observational data.
Reweighted variants of AIPW,
referred to as augmented inverse propensity of sampling weighting (AIPSW) estimators by~\citet{colnet2024causal},
attain the influence functions $\varphi_{0,\obs}$ and $\varphi_{0,\tgt}$ for $\tau_{\obs}$ and $\tau_{\tgt}$, respectively,
using the experimental data along with the covariates (but not the outcomes) in the observational data.

\subsection{The outcome mean function tangent space}
\label{sec:S_M}

If the set $\cm$ is a closed linear space,
the outcome mean function tangent space $\cs_{\cm}$ is simply $\cm$.

\begin{proposition}
\label{prop:S_M_equals_M}
If $\cm$ is a closed linear subspace of $\ch_R$, then $\cs_{\cm}=\cm$.
\end{proposition}
\begin{proof}
First suppose $f \in \cs_{\cm}$.
Then there exists a sequence of nonrandom vectors $\{c_i\}_{i=1}^N$ and
outcome mean function parametric submodels $\{M_{\gamma}^{(i)}\}_{i=1}^N$ such that the functions 
\[
f_i=c_i^{\top}\frac{\partial M_{\gamma}^{(i)}}{\partial \gamma} \Bigg|_{\gamma=0}
\]
converge in mean square to $f$.
But for each $i$, letting $s_i$ be the dimension of $c_i$ we have
\[
c^{\top} \frac{\partial M_{\gamma}^{(i)}(\cdot)}{\partial \gamma} \Bigg|_{\gamma=0} =  \lim_{h \rightarrow 0} \sum_{k=1}^{s_i} c_k \frac{M_{he_k}(\cdot)-m^*(\cdot)}{h}
\]
is a limit of linear combinations of functions in $\cm$,
hence still in $\cm$ due to $\cm$ being closed and linear.
Above, $e_k$ denotes the $k$-th standard basis vector in $\real^{s_i}$.
We conclude $\cs_{\cm} \subseteq \cm$.

To show the reverse inclusion,
note that for any $f \in \cm$ we can write $f$ as the derivative of the one-dimensional parametric submodel $m(\cdot;\gamma)=m^*(\cdot)+\gamma f(\cdot)$,
which is indeed a submodel as it satisfies $m(\cdot;\gamma) \in \cm$ for all $\gamma$ due to $\cm$ being a linear space.
\end{proof}

It is straightforward to verify that the collections $\cm_2$ and $\cm_4$ of~\eqref{eq:mean_exchangeable_controls} and~\ref{eq:parametric_confounding_bias}, respectively,
are in fact closed linear spaces,
so that $\cs_{\cm_i}=\cm_i$ for $i=2,4$.
A linear restricted moment model $\cm=\{r \mapsto \psi(r)^{\top}\beta \mid \beta \in \real^q\}$,
a special case of $\cm_1$ in Example~\ref{ex:restricted_moment_model},
is also a closed linear space.

In other cases,
computing the space $\cs_{\cm}$ may require some more work.
However, Theorem~\ref{thm:semiparametric_tangent_space} ensures that computing $\cs_{\cm}$
(and its orthogonal complement) can replace the work of directly deriving the space $\ct_{\cm}$.
Such a direct derivation generally requires reasoning about integral restrictions on log likelihoods in parametric submodels.
The proof of Theorem~\ref{thm:semiparametric_tangent_space} essentially proceeds by working out such reasoning for general $\cm$ so that the user only needs to reason about the space $\cs_{\cm}$,
which requires only considering derivatives of outcome mean function submodels.

To illustrate this,
we work through a novel derivation of the outcome mean function tangent space $\cs_{\cm_5}$,
where $\cm_5$ is the outcome mean function collection~\eqref{eq:outcome_selection_bias} in the outcome-mediated selection bias setting of Example~\ref{ex:outcome_selection_bias}.

\begin{example}[continues=ex:outcome_selection_bias]
Let $\ell(x) = \log(x)-\log(1-x)$ be the logistic link function,
so that~\eqref{eq:odds_ratio} can be written as
\[
\ell(m_{11}(x)) = \ell(m_{10}(x)) + \ell(m_{01}(x)) - \ell(m_{00}(x)).
\]
Essentially, the restriction~\eqref{eq:odds_ratio} allows $m_{10}$, $m_{01}$, and $m_{00}$ to freely vary but then $m_{11}$ is a deterministic function of $(m_{10},m_{01},m_{00})$.
Thus,
a generic outcome mean function parametric submodel for $\cm_5$ takes the form
\begin{align}
m_{sz}(x;\gamma) & = szf_{11}(x;\gamma) + s(1-z)f_{10}(x;\gamma_{10}) + (1-s)zf_{01}(x;\gamma_{01}) + (1-s)(1-z)f_{00}(x;\gamma_{00}) \label{eq:outcome_selection_bias_submodel} \\
f_{11}(x;\gamma) & = \ell^{-1}[\ell(f_{10}(x;\gamma_{10}))+\ell(f_{01}(x;\gamma_{01}))-\ell(f_{00}(x;\gamma_{00}))] \label{eq:outcome_selection_bias_submodel_f}
\end{align}
for $\gamma=(\gamma_{10}^{\top},\gamma_{01}^{\top},\gamma_{00}^{\top})^{\top} \in \real^s$.
For any $c=(c_{10}^{\top},c_{01}^{\top},c_{00}^{\top})^{\top} \in \real^r$ partitioned as $\gamma$,
we compute
\begin{align*}
c^{\top} \frac{\partial m(r;\gamma)}{\partial \gamma} \Big|_{\gamma=0} & = [s(1-z)+sz\kappa^*(x)\ell'(m_{10}^*(x))]c_{10}^{\top}f_{10}'(x;0) \\
& \quad + [(1-s)z+sz\kappa^*(x)\ell'(m_{01}^*(x))]c_{01}^{\top}f_{01}'(x;0) \\
& \quad + [(1-s)(1-z)-sz\kappa^*(x)\ell'(m_{00}^*(x))]c_{00}^{\top}f_{00}'(x;0)
\end{align*}
where $f_j'(x;a) \equiv \partial f_j(x;\alpha)/\partial \alpha|_{\alpha=a}$ for $j \in \{10, 01, 00\}$
and 
$\kappa^*(x) = (\ell^{-1})'[\ell(m_{11}^*(x))] = (\ell'(m_{11}^*(x)))^{-1}$.
We note the preceding display is of the form
\[
s(1-z)g_1(x)+(1-s)zg_2(x)+(1-s)(1-z)g_3(x) + sz\kappa^*(x)[\ell'(m_{10}^*(x))g_1(x)+\ell'(m_{01}^*(x))g_2(x)-\ell'(m_{00}^*(x))g_3(x)]
\]
for some $g_1,g_2,g_3\in \ch_X$
and hypothesize that \emph{any} function of this form lies in $\cs_{\cm_5}$.

To show this, we fix arbitrary $g_1,g_2,g_3 \in \ch_X$ and note there exists an outcome mean function submodel $m(r;\gamma)$ of the form in~\eqref{eq:outcome_selection_bias_submodel} and~\eqref{eq:outcome_selection_bias_submodel_f} with 
\begin{align*}
f_{10}(x;\gamma_{10}) & = m_{10}^*(x)+\gamma_{10}g_1(x) \\
f_{01}(x;\gamma_{01}) & = m_{01}^*(x)+\gamma_{01}g_2(x) \\
f_{00}(x;\gamma_{00}) & = m_{00}^*(x)+\gamma_{00}g_3(x)
\end{align*}
where $\gamma=(\gamma_{10},\gamma_{01},\gamma_{00}) \in \real^3$.
Then 
\begin{align*}
& (1,1,1)^{\top} \frac{\partial m_{sz}(x;\gamma)}{\partial \gamma} \Big|_{\gamma=0} = s(1-z)g_1(x)+(1-s)zg_2(x)+(1-s)(1-z)g_3(x)+sz\kappa^*(x)\iota^*(x)
\end{align*}
which shows that indeed,
\begin{align}
\cs_{\cm_5} = & \{r \mapsto s(1-z)g_1(x)+(1-s)zg_2(x)+(1-s)(1-z)g_3(x) \nonumber \\
& \quad + sz\kappa^*(x)[\ell'(m_{10}^*(x))g_1(x)+\ell'(m_{01}^*(x))g_2(x)-\ell'(m_{00}^*(x))g_3(x)] \mid g_1,g_2,g_3 \in \ch_X\}. \label{eq:S_M_outcome_selection_bias}
\end{align}
\end{example}

\subsection{Orthogonal complements and projections}
\label{sec:projections}
The final step in deriving the EIF is to compute the orthogonal complement $\cs_{\cm}^{\perp}$ of the outcome mean function tangent space $\cs_{\cm}$
along with the projection~\eqref{eq:eif_raw}.
A useful way to compute $\cs_{\cm}^{\perp}$ is to characterize it as the set of all functions in $\ch_R$ whose projection onto $\cs_{\cm}$ is zero.
We compute such projections,
as well as the final projection~\eqref{eq:eif_raw},
by using the law of iterated expectations to write the projection objective as a function of $X$ solely,
and then minimize the objective using calculus.
We perform both these steps for the collection $\cm_5$ below.

\begin{example}[continues=ex:outcome_selection_bias]
By definition,
the orthogonal complement $\cs_{\cm_5}^{\perp}$ of the set $\cs_{\cm_5}$ in~\eqref{eq:S_M_outcome_selection_bias}
consists of all functions $h \in \ch_R$ for which the projection $\pi(h;\cs_{\cm_5})$ is zero.
By~\eqref{eq:S_M_outcome_selection_bias},
an arbitrary function in $\cs_{\cm_5}$ takes the form
\begin{align*}
\iota(x;g_1,g_2,g_3) & = s(1-z)g_1(x)+(1-s)zg_2(x)+(1-s)(1-z)g_3(x) + sz \upsilon^*(x;g_1,g_2,g_3) \\
\upsilon^*(x;g_1,g_2,g_3) & = \kappa^*(x)[\ell'(m_{10}^*(x))g_1(x)+\ell'(m_{01}^*(x))g_2(x)-\ell'(m_{00}^*(x))g_3(x)]
\end{align*}
for some $(g_1,g_2,g_3) \in \ch_X^3$.
Then for any $h \in \ch_R$,
we have $\pi(h;\cs_{\cm_5})(x) = \iota(x;h_1,h_2,h_3)$
where $(h_1,h_2,h_3)$ minimize the projection objective 
\begin{align*}
\mathcal{O}(g_1,g_2,g_3) & = \e[(h(R)-\iota(X;g_1,g_2,g_3))^2] \\
& = \e\Big[p^*(X)e^*(X)(h_{11}(X)-\upsilon^*(X;g_1,g_2,g_3))^2\Big] + \e[p^*(X)(1-e^*(X))(h_{10}(X)-g_1(X))^2] \\
& \quad + \e[(1-p^*(X))q^*(X)(h_{01}(X)-g_2(X))^2] + \e[(1-p^*(X))(1-q^*(X))(h_{00}(X)-g_3(X))^2]
\end{align*}
over $(g_1,g_2,g_3) \in \ch_X^3$,
where the second equality follows from writing
\[
h(R) = SZh_{11}(X)+S(1-Z)h_{10}(X)+(1-S)Zh_{01}(X) + (1-S)(1-Z)h_{00}(X)
\]
and conditioning on $X$,
and we have used the notation from Table~\ref{table:functions}.
With the objective now solely an expectation of a function of $X$,
we can compute the optimizers $(h_1, h_2,h_3)$
simply by minimizing the argument of the expectation pointwise in $(g_1,g_2,g_3)$.
This yields the first-order conditions
\begin{align*}
0 & = -p^*(x)e^*(x)(h_{11}(x)-\upsilon^*(x;h_1,h_2,h_3))\kappa^*(x)\ell'(m_{10}^*(x)) - p^*(x)(1-e^*(x))(h_{10}(x)-g_1(x)) \\
0 & = -p^*(x)e^*(x)(h_{11}(x)-\upsilon^*(x;h_1,h_2,h_3))\kappa^*(x)\ell'(m_{01}^*(x)) -(1-p^*(x))q^*(x)(h_{01}(x)-g_2(x)) \\
0 & = p^*(x)e^*(x)(h_{11}(x)-\upsilon^*(x;h_1,h_2,h_3))\kappa^*(x)\ell'(m_{00}^*(x))-(1-p^*(x))(1-q^*(x))(h_{00}(x)-g_3(x))
\end{align*}
for all $x \in \cx$.
We conclude $\cs_{\cm_5}^{\perp}$ is the set of all $h \in \ch_R$
for which $h_1=h_2=h_3=0$ solves this system of equations.
This is precisely the functions $h$ satisfying
\begin{align}
(1-e^*(x))h_{10}(x) & = -e^*(x)h_{11}(x)\frac{\ell'(m_{10}^*(x))}{\ell'(m_{11}^*(x))} \label{eq:T_M_1_perp_1} \\
(1-p^*(x))q^*(x)h_{01}(x) & = -p^*(x)e^*(x)h_{11}(x)\frac{\ell'(m_{01}^*(x))}{\ell'(m_{11}^*(x))} \label{eq:T_M_1_perp_2} \\
(1-p^*(x))(1-q^*(x))h_{00}(x) & = p^*(x)e^*(x)h_{11}(x)\frac{\ell'(m_{00}^*(x))}{\ell'(m_{11}^*(x))} \label{eq:T_M_1_perp_3} 
\end{align}
for all $x \in \cx$.
Noting that~\eqref{eq:T_M_1_perp_1},~\eqref{eq:T_M_1_perp_2}, and~\eqref{eq:T_M_1_perp_3}
determine $h_{10}$, $h_{01}$, and $h_{00}$ in terms of $h_{11}$ and letting $\zeta(x)$ denote $h_{11}(x)/\ell'(m_{11}^*(x))$,
we can write this succinctly as
\[
\cs_{\cm_5}^{\perp} = \{f^*(r)\zeta(x) \mid \zeta \in \ch_X\} 
\]
for $f^*(r) \equiv f(r;\eta^*)$ where
\begin{align}
f(r;\eta) & = sz\ell'(m_{11}(x))-\frac{s(1-z)e(x)\ell'(m_{10}(x))}{1-e(x)} \nonumber \\
& \quad -\frac{(1-s)zp(x)e(x)\ell'(m_{01}(x))}{(1-p(x))q(x)} + \frac{(1-s)(1-z)p(x)e(x)\ell'(m_{00}(x))}{(1-p(x))(1-q(x))} \label{eq:f}
\end{align}
and the nuisance parameters are $\eta(\cdot) = (p(\cdot), e(\cdot), q(\cdot), m(\cdot))$.
We conclude via~\eqref{eq:T_M_perp_generalized} that
\begin{equation}
\label{eq:T_M_perp_outcome_selection_bias}
\ct_{\cm_5}^{\perp} = \{w \mapsto f^*(r)\zeta(x)(y-m^*(r)) \mid \zeta \in \ch_X\}.
\end{equation}

Finally, to derive the EIF we perform the projection~\eqref{eq:eif_raw}.
This amounts to solving the optimization problem
\begin{equation}
\label{eq:outcome_selection_bias_projection}
\min_{\zeta \in \ch_X} \e\left[\left(\varphi_0(W;\tau^*,\eta^*)-f^*(R)\zeta(X)(Y-m^*(R))\right)^2\right]
\end{equation}
Similar to our derivation of the orthogonal complement $\cs_{\cm_5}^{\perp}$,
we perform this minimization by re-writing the objective via conditioning on $X$:
\begin{align*}
& \e\left[\left(\varphi_0(W;\tau^*,\eta^*)-f^*(R)\zeta(X)(Y-m^*(R))\right)^2\right] \\
& \quad = \e\left[\zeta(X)^2\e[f^*(R)^2(Y-m^*(R))^2 \mid X] - 2\zeta(X)\e[\varphi_0(W;\tau^*,\eta^*)f^*(R)(Y-m^*(R)) \mid X]\right] + const.
\end{align*}
Minimizing the argument of the expectation shows the solution to~\eqref{eq:outcome_selection_bias_projection} is $\zeta(\cdot;\eta^*) \equiv \zeta(\cdot;\eta^*,\tau^*,\varphi_0)$ where
\[
\zeta(x;\eta) \equiv \zeta(x;\eta,\tau,\varphi_0) = \frac{\int \varphi_0(w;\tau,\eta)f(r;\eta)(y-m(r))dP_{(S,Z,Y \mid X)}(s,z,y;x)}{\int f(r;\eta)^2(y-m(r))^2 dP_{(S,Z,Y \mid X)}(s,z,y;x)}
\]
and so by~\eqref{eq:eif_raw} the EIF is
\begin{equation}
\label{eq:outcome_selection_bias_eif}
\varphi_{\eff}^{(5)}(w;\tau^*,\eta^*) = \varphi_0(w;\tau^*,\eta^*) - f(r;\eta^*)\zeta(x;\eta^*,\tau^*,\varphi_0)(y-m^*(r)).
\end{equation}
\end{example}

\subsection{Additional insights}
\label{sec:additional_insights}
Commonly,
the propensity score $e^*(x)$ in the experimental dataset is known,
as it is subject to the control of the investigator.~\citet{hahn1998role} famously showed that knowledge of the propensity score does not change the semiparametric efficiency bound for estimating $\tau_{\rct}$ relative to the nonparametric model.
This may not be true in general, however,
for other estimands or for proper semiparametric models like $\cp_{\cm}$.

We can readily extend Theorem~\ref{thm:semiparametric_tangent_space}
to characterize the space of influence functions of RAL estimators in the model $\tilde{\cp}_{\cm} \subseteq \cp_{\cm}$
that imposes the restriction $e(x)=e^*(x)$ for some known function $e^*(\cdot)$
in addition to the outcome mean function restriction that $m \in \cm$.
The main insight is to note that with $e(x) = \Pr(Z=1 \mid S=1,X=x) = \e[Z \mid S=1,X=x]$,
a known experimental propensity score can be viewed as a restriction on the outcome mean function of the observations $(Z,S,X)$,
where now $Z$ is viewed as the ``outcome" and $(S,X)$ are the covariates.
Such a restriction is variationally independent of any outcome mean function restriction on the original dataset (i.e., that $m \in \cm$),
which must pertain to the conditional distribution of $Y$ given $(Z,S,X)$.
This allows us to recursively apply Theorem~\ref{thm:semiparametric_tangent_space}.
Similar ideas could be used to characterize semiparametric efficiency when either instead of or in addition to a known experimental propensity score,
the RCT selection probability $p^*(x)$ is known,
as in a nested trial setup where RCT participants are directly sampled from a larger group by the investigator.
A known selection probability is a restriction on the conditional mean of $S$ given $X$,
which is variationally independent of any restrictions on the conditional distribution of $Z$ given $(S,X)$ (e.g., known experimental propensity score),
or on the conditional distribution of $Y$ given $R$ (e.g., that $m \in \cm$).

\begin{corollary}
\label{cor:known_propensity_score}
Suppose the conditions of Theorem~\ref{thm:semiparametric_tangent_space} are satisfied for some outcome mean function collection $\cm$ in the causal inference setting where $R=(S,Z,X) \in \{0,1\}^2 \times \cx$.
Then $\tilde{\ct}_{\cm}^{\perp}$,
the orthogonal complement in $\ch_W^0$ of the semiparametric tangent space $\tilde{\ct}_{\cm}$ for the model $\tilde{\cp}_{\cm}$ that imposes the restrictions that $m \in \cm$ and $\Pr(Z=1 \mid S=1,X=x) = e^*(x)$ for known $e^* \in \ch_X$,
satisfies
\begin{equation}
\tilde{\ct}_{\cm}^{\perp} = \ct_{\cm}^{\perp} \oplus \left\{w \mapsto sh(x)(z-e^*(x)) \Big| h \in \ch_X \right\} . \label{eq:T_M_perp_tilde}
\end{equation}
\end{corollary}
\begin{proof}
See Appendix~\ref{app:known_propensity_score}.
\end{proof}

Theorem~\ref{thm:semiparametric_tangent_space} can also provide useful insight into when certain modeling assumptions may not be helpful,
in the sense that imposing such restrictions cannot improve (decrease) the semiparametric efficiency bound.
\begin{corollary}
\label{cor:homoskedasticity}
Suppose the experimental propensity score is constant, i.e. $e^*(x) = e^* \in (0,1)$ for all $x \in \cx$,
and the variance functions in the RCT are homoskedastic, i.e. 
\[
V^*(1,z,x)=\var(Y \mid S=1,Z=z,X=x)=\sigma_z^2, \quad z \in \{0,1\}, x \in \cx
\]
for some positive constants $\sigma_0^2$ and $\sigma_1^2$.
Additionally,
suppose $\cm \subseteq \ch_R$ satisfies the conditions of Theorem~\ref{thm:semiparametric_tangent_space}
and the corresponding outcome mean tangent space $\cs_{\cm}$ contains all functions $m$ for which $m_{11}$ and $m_{10}$ are arbitrary constants and $m_{01}$ and $m_{00}$ are identically zero.
Then
$\varphi_{0,\rct}$,
the EIF for $\tau_{\rct}$ in the nonparametric model given in~\eqref{eq:tau_exp_bound_0},
is also the EIF for $\tau_{\rct}$ in both $\cp_{\cm}$ and $\tilde{\cp}_{\cm}$.
\end{corollary}
\begin{proof}
See Appendix~\ref{app:homoskedasticity}.
\end{proof}

To understand the implications of Corollary~\ref{cor:homoskedasticity},
consider the following restrictive set of modeling assumptions: the outcome mean functions $m_{01}$ and $m_{00}$ in the observational dataset are known exactly,
while the outcome mean functions $m_{11}$ and $m_{10}$ in the RCT are linear in some known low-dimensional basis functions $\psi(x)$ containing an intercept.
It is straightforward to show that the set $\cm$ of permissible outcome mean functions under these restrictions satisfies the conditions of Corollary~\ref{cor:homoskedasticity}.
This shows that 
with homoskedastic outcomes and a constant propensity score $e^*$,
even very stringent assumptions on the outcome mean functions do not enable any asymptotic improvement over the AIPW estimator for estimating $\tau_{\rct}$.
The same automatically applies for any set of strictly less stringent assumptions,
such as the model $\cp_{\cm_4}$
(again assuming $\psi(x)$ contains an intercept).
Recall that the AIPW estimator neither uses any information from the observational dataset
nor places any parametric or semiparametric assumptions on the outcome mean functions.
Thus, Corollary~\ref{cor:homoskedasticity} shows it is not useful in large samples to make such assumptions to try and leverage the observational data if we have a constant RCT propensity score and expect homoskedastic outcomes.
Conversely, with heteroskedastic outcomes
it is known that
even without an additional observational dataset,
parametric assumptions on the outcome mean functions can improve the efficiency bound for $\tau_{\rct}$~\citep{tan_comment_2007},
including when the propensity score is constant.

\section{Constructing one-step estimators}
\label{sec:estimation_inference}

Suppose $\varphi=\varphi(\cdot;\tau^*,\eta^*)$ is an influence function for an RAL estimator of $\tau$ in the model $\cp_{\cm}$.
Then techniques like one-step adjustment of an initial estimator $\hat{\tau}$,
solving the estimating equation $N^{-1} \sum_{i=1}^N \varphi(W_i;\tau,\hat{\eta})=0$ for $\tau$
(where $\hat{\eta}$ is an initial estimate of the nuisance parameters $\eta^*$),
and targeted minimum loss estimation~\citep{hines2022demystifying} are all well-studied options to construct an estimator with influence function $\varphi$.

Our focus is on providing precise conditions under which one-step adjustment,
which most directly utilizes the structure of Section~\ref{sec:semiparametric_efficiency},
yields an estimator that in fact achieves the efficiency bound in our models $\cp_{\cm}$
(i.e., has influence function equal to the EIF).
We emphasize that while most attention has been placed on one-step adjustment of non-RAL estimators in nonparametric models to gain $N^{1/2}$ consistency,
here we propose one-step adjustment of RAL estimators in our proper semiparametric models $\cp_{\cm}$,
where there is typically than one possible influence function.

Let $\varphi_0=\varphi_0(\cdot;\tau^*,\eta^*)$ be the influence function of an initial RAL estimator $\hat{\tau}$ in the semiparametric model $\cp_{\cm}$,
as in Section~\ref{sec:varphi_0}.
Expanding the definition of the nuisance parameters $\eta^*$ if necessary,
let $\varphi=\varphi(\cdot;\tau^*,\eta^*)$ be the desired influence function of an RAL estimator of $\tau$
(in our context, generally the EIF).
By Theorem~\ref{thm:influence_functions},
we can write $\varphi=\varphi_0-g$ where $g=g(\cdot;\tau^*,\eta^*) \in \ct_{\cm}^{\perp}$.
Since $\tau^*$ and $\eta^*$ are not known,
we need to plug in initial estimates $\hat{\tau}$ and $\hat{\eta}$.
Ensuring this initial plug-in procedure does not affect the first-order asymptotics of the one-step estimator $\hat{\tau}_{\os} = \hat{\tau} - N^{-1} \sum_{i=1}^N  g(W_i;\hat{\tau},\hat{\eta})$ requires that
\begin{equation}
\label{eq:g_hat_minus_g}
\frac{1}{N} \sum_{i=1}^N g(W_i;\hat{\tau},\hat{\eta})-g(W_i;\tau^*,\eta^*) = o_p(N^{-1/2}).
\end{equation}
One challenge with satisfying~\eqref{eq:g_hat_minus_g} is that the nuisance parameters $\eta^*$ are typically infinite dimensional,
and thus need to be estimated nonparametrically to avoid additional modeling assumptions.
A popular approach to weaken the conditions needed for~\eqref{eq:g_hat_minus_g} to hold with nonparametric nuisance estimates is cross-fitting.
Cross-fitting entails partitioning 
the observation indices $1,\ldots,N$ into $K$ roughly equally-sized folds $\ci_1,\ldots,\ci_K$
for some fixed $K \geq 2$.
Then for each fold $k=1,\ldots,K$,
we find estimates $\hat{\tau}^{(-k)}$ and $\hat{\eta}^{(-k)}$ of $\tau^*$ and $\eta^*$,
respectively,
using only the observations whose indices lie \emph{outside} $\ci_k$.
As we will see,
orthogonality properties of the influence function $\varphi$ enable the following cross-fit variant of~\eqref{eq:g_hat_minus_g}
to hold even when the components of $\hat{\eta}^{(-k)}$ converge at nonparametric rates:
\begin{equation}
\label{eq:g_hat_minus_g_crossfit}
\frac{1}{N} \sum_{k=1}^K \sum_{i \in \ci_k} g\left(W_i;\hat{\tau}^{(-k)},\hat{\eta}^{(-k)}\right)-g(W_i;\tau^*,\eta^*) = o_p(N^{-1/2}).
\end{equation}
Such a result,
with sufficient conditions stated precisely in Lemma~\ref{lemma:dml_master} below,
closely follows the literature on double machine learning~\citep{van_der_laan_cross-validated_2011, chernozhukov2018double}.

We assume cross-fitting into $K \geq 2$ folds for the remainder of this section,
and adopt the notation introduced in the previous paragraph in our formal results.
First,
we formalize the construction of our cross-fit one-step estimator $\hat{\tau}_{\os}$.

\begin{theorem}
\label{thm:cross_fit_os}
Let $\hat{\tau}_0$ be any RAL estimator of pathwise differentiable $\tau \in \real$ in some semiparametric model with influence function $\varphi_0=\varphi_0(\cdot;\tau^*,\eta^*)$.
Suppose~\eqref{eq:g_hat_minus_g_crossfit} holds for some $g=g(\cdot;\tau^*,\eta^*) \in \ch^0$
and cross-fit estimates $\hat{\tau}^{(-k)}$ and $\hat{\eta}^{(-k)}$ of $\tau^*$ and $\eta^*$,
respectively,
for $k=1,\ldots,K$.
Then
\begin{equation}
\label{eq:tau_hat_os}
\hat{\tau}_{\os} = \hat{\tau}_0 - \frac{1}{N} \sum_{k=1}^K \sum_{i \in \ci_k} g\left(W_i;\hat{\tau}^{(-k)},\hat{\eta}^{(-k)}\right)
\end{equation}
satisfies the central limit theorem $\sqrt{N}(\hat{\tau}_{\os}-\tau^*) \stackrel{d}{\rightarrow} \mathcal{N}(0,\e[\varphi(W)^2])$ whenever the model holds,
for $\varphi(\cdot)=\varphi(\cdot;\tau^*,\eta^*) = \varphi_0(\cdot;\tau^*,\eta^*)-g(\cdot;\tau^*,\eta^*)$.
\end{theorem}
\begin{proof}
By the definition of $\hat{\tau}_{\os}$ in~\eqref{eq:tau_hat_os} we have
\begin{align*}
\sqrt{N}(\hat{\tau}_{\os}-\tau^*) & = \sqrt{N}(\hat{\tau}_0-\tau^*) - \frac{1}{\sqrt{N}} \sum_{k=1}^K \sum_{i \in \ci_k} g\left(W_i;\hat{\tau}^{(-k)},\hat{\eta}^{(-k)}\right) \\
& = \frac{1}{\sqrt{N}} \sum_{i=1}^N \varphi_0(W_i;\tau^*,\eta^*)-g(W_i;\tau^*,\eta^*) + o_p(1)
\end{align*}
where by~\eqref{eq:g_hat_minus_g_crossfit} and the definition of $\varphi_0$,
the second equality holds whenever the model holds.
The result follows by the ordinary central limit theorem.
\end{proof}

Theorem~\ref{thm:cross_fit_os} does not use any properties about $g$ --- not even that $\varphi=\varphi_0-g$ is a valid influence function of an RAL estimator of $\tau$.
It also does not consider solely models of the form $\cp_{\cm}$.
However, the geometry of the space of influence functions of RAL estimators in the model $\cp_{\cm}$ given by Theorem~\ref{thm:influence_functions} ensures that the condition~\eqref{eq:g_hat_minus_g_crossfit} can be feasibly satisfied when $\varphi$ is in fact a valid influence function of a RAL estimator (i.e., that $g \in \ct_{\cm}^{\perp}$).
Of course, this is necessarily the case when $\varphi=\varphi_{\eff}$, i.e., we are targeting the EIF.

\begin{lemma}
\label{lemma:dml_master}
Suppose the function $g=g(\cdot;\tau^*,\eta^*)$ is any element of the space $\ct_{\cm}^{\perp}$ for
some collection of outcome mean functions $\cm \subseteq \ch_R$ satisfying the conditions of Theorem~\ref{thm:semiparametric_tangent_space},
so that we can write
\[
g(w;\tau^*,\eta^*) = (y-m^*(r))h_1(r;\tau^*,\eta^*)
\]
where $\eta^*$ includes the mean function $m^*$ and
$h_1^*(\cdot) =  h_1(\cdot;\tau^*,\eta^*) \in \cs_{\cm}^{\perp}$.
Then equation~\eqref{eq:g_hat_minus_g_crossfit} holds
under the following conditions:
\begin{itemize}
\item (Boundedness) There exists $C < \infty$ for which $V^*(r) \leq C$ for $P^*$-almost every $r$ and both $\|h_1^*\|_{\infty,P^*} \leq C$ and $\left\|\hat{h}_1^{(-k)}\right\|_{\infty,P^*} \leq C$ where $\|f\|_{\infty,P^*} = \inf_M \Pr(|f(R)| \leq M) = 1$
\item (Approximate tangency) With probability tending to 1,
for each $k=1,\ldots,K$ there exists $\hat{R}^{(-k)}=\hat{R}^{(-k)}(\cdot) \in \ch_R$ depending only on the observations outside fold $k$ such that the function $r \mapsto \hat{m}^{(-k)}(r) - m^*(r)-\hat{R}^{(-k)}(r)$ lies in the outcome mean function tangent space $\cs_{\cm}$,
$\|\hat{R}^{(-k)}\|_{2,P^*} = o_p(1)$,
and the following rate holds:
\begin{equation}
\label{eq:R_hat}
\frac{1}{|\ci_k|} \sum_{i \in \ci_k} \left|\hat{R}^{(-k)}(R_i)\right| = o_p(N^{-1/2})
\end{equation}
\item (Rate conditions) The following rate conditions hold for all $k=1,\ldots,K$:
\begin{itemize}
    \item (First order consistency) $\|\hat{m}^{(-k)}-m^*\|_{2,P^*} + \|\hat{h}_1^{(-k)}-h_1^*\|_{2,P^*} = o_p(1)$ 
    \item (Product rate condition) $\|\hat{m}^{(-k)}-m^*\|_{2,P^*} \times \|\hat{h}_1^{(-k)}-h_1^*\|_{2,P^*} = o_p(N^{-1/2})$
\end{itemize}
where $\|f\|_{2,P^*} = (\int f^2(w) dP^*(w))^{1/2}$.
\end{itemize}
\end{lemma}

In the proof of Lemma~\ref{lemma:dml_master} in Appendix~\ref{app:proof_dml_master},
we generalize Lemma~\ref{lemma:dml_master} to cover models $\tilde{\cp}_{\cm}$ where the experimental propensity score is known (cf. Corollary~\ref{cor:known_propensity_score}).
The main conditions in Lemma~\ref{lemma:dml_master} that need to be verified for a given choice of $\cm$ are approximate tangency and the rate conditions.
One can check that approximate tangency holds with $\hat{R}^{(-k)}=0$ if $m^* \in \cm$ and $\hat{m}^{(-k)} \in \cm$ for all folds $k$ with probability tending to 1,
and $\cm$ is a closed linear space.
Otherwise, approximate tangency can be established (with generally nonzero remainders $\hat{R}^{(-k)}$) via Taylor expansion in appropriately smooth models
such as $\cp_{\cm_5}$,
provided that the outcome mean function estimates $\hat{m}^{(-k)}$ lie in $\cm$.
The rate conditions required in Lemma~\ref{lemma:dml_master} are standard in the double machine learning literature.
Substantial work has established precise conditions under which they hold for particular nonparametric estimators~\citep{chernozhukov2018double,benkeser2016highly,schuler2024highly}.

\section{Some specific novel efficient estimators}
\label{sec:applications}
We are now ready to put everything together to specify novel one-step efficient estimators of the average treatment effects $\tau_{\rct}$, $\tau_{\obs}$, and $\tau_{\tgt}$ in the models $\cp_{\cm_4}$ and $\cp_{\cm_5}$
that assume linear confounding bias and outcome-mediated selection bias,
respectively.

In Appendix~\ref{app:linear_confounding_bias},
we show that the EIF $\varphi_{\eff}^{(4)}$ in the linear confounding bias model $\cp_{\cm_4}$ in terms of any initial influence function $\varphi_0=\varphi_0(\cdot;\tau^*,\eta^*)$ of an RAL estimator of $\tau$ is given by
\begin{equation}
\label{eq:eif_4}
\varphi_{\eff}^{(4)}(w;\tau^*,\eta^*) = \varphi_0(w;\tau^*,\eta^*) - h(r;\eta^*)\nu(x;\eta^*,\tau^*,\varphi_0)(y-m^*(r))
\end{equation}
where
\begin{align}
h(r;\eta) & = \frac{sz}{p(x)e(x)}-\frac{s(1-z)}{p(x)(1-e(x))} -\frac{(1-s)z}{(1-p(x))q(x)} + \frac{(1-s)(1-z)}{(1-p(x))(1-q(x))} \label{eq:f_1} \\
I(x;\eta,\tau,\varphi_0) & = \int \varphi_0(w;\tau,\eta)(y-m(r))h(r;\eta) dP_{(S,Z,Y) \mid X}(s,z,y;x) \label{eq:I_x} \\
\nu(x;\eta,\tau,\varphi_0) & =
p(x)e(x)\frac{I(x;\eta,\tau,\varphi_0) + \lambda(\eta,\tau,\varphi_0)^{\top}\psi(x)}{\Sigma(x;\eta)} \label{eq:zeta_1_star} \\
\Sigma(x;\eta) & = V_{11}(x)+ \frac{V_{10}(x)e(x)}{1-e(x)} + \frac{V_{01}(x)p(x)e(x)}{(1-p(x))q(x)} + \frac{V_{00}(x)p(x)e(x)}{(1-p(x))(1-q(x))}\label{eq:Sigma_star} \\
\lambda(\eta,\tau,\varphi_0) & = 
-\left(\int \frac{p(x)e(x)}{\Sigma(x;\eta)}\psi(x)\psi(x)^{\top}dP_X(x)\right)^{-1} \int \frac{I(x;\eta,\tau,\varphi_0)p(x)e(x)}{\Sigma(x;\eta)}\psi(x) dP_X(x)\label{eq:lambda}
\end{align}
for $\eta = \eta(P)$ including the nuisance functions $(p(\cdot),e(\cdot),m(\cdot),q(\cdot),V(\cdot))$
along with any other functionals that are needed to compute 
\[
I(x;\eta^*,\tau^*,\nu_0) = \e[\varphi_0(W;\tau^*,\eta^*)(Y-m^*(R))h(R;\eta^*) \mid X=x].
\]
Applying the construction~\eqref{eq:tau_hat_os} to the differences 
\[
g(\cdot;\tau,\eta)=\varphi_{\eff}^{(4)}(\cdot;\tau,\eta)-\varphi_0(\cdot;\tau,\eta)=h(r;\eta)\nu(x;\tau,\eta)(y-m(r))
\]
gives the efficient cross-fit one-step estimators
\begin{align}
\label{eq:tau_hat_eff_4}
\hat{\tau}_{\eff}^{(4)} = \hat{\tau}_0 - \frac{1}{N} \sum_{k=1}^K \sum_{i \in \ci_k} h\Big(R_i;\hat{\eta}^{(-k)}\Big)\nu\Big(X_i;\hat{\eta}^{(-k)},\hat{\tau}_0,\varphi_0\Big)\Big(Y_i-\hat{m}^{(-k)}(R_i)\Big).
\end{align}
where $\hat{\tau}_0$ is any initial RAL estimator with influence function $\varphi_0$.
Taking $\varphi_0=\varphi_{0,\rct},\varphi_{0,\obs},\varphi_{0,\tgt}$ gives efficient estimators
for the estimands $\tau_{\rct},\tau_{\obs},\tau_{\tgt}$, respectively.
In that case,
$\hat{\tau}_0$ can be a cross-fit AIP(S)W estimator obtained by solving the estimating equations 
\begin{align}
\label{eq:aipsw_estimating_equations}
N^{-1} \sum_{k=1}^K \sum_{i \in \ci_k} \varphi_0(W_i;\tau;\hat{\eta}^{(-k)}) = 0
\end{align}
for $\varphi_0=\varphi_{0,\rct},\varphi_{0,\obs},\varphi_{0,\tgt}$.
We remind the reader that these AIP(S)W estimators are efficient in the nonparametric model.
For convenience in computing~\eqref{eq:tau_hat_eff_4}
for the estimands $\tau_{\rct},\tau_{\obs},\tau_{\tgt}$,
note
\begin{align*}
I(x;\eta,\tau, \varphi_{0,\rct}) & =  \frac{1}{\rho}\left(\frac{V_{11}(x)}{e(x)} + \frac{V_{10}(x)}{1-e(x)}\right) \\
I(x;\eta,\tau,\varphi_{0,\obs}) & = \frac{\rho(1-p(x))}{p(x)(1-\rho)} I(x;\eta,\tau,\varphi_{0,\rct}),
\quad I(x;\eta,\tau,\varphi_{0,\tgt}) = \frac{\rho}{p(x)}I(x;\eta,\tau,\varphi_{0,\rct}).
\end{align*}

Similarly, by~\eqref{eq:outcome_selection_bias_eif} we can compute efficient cross-fit one-step estimators
\begin{align}
\label{eq:tau_hat_eff_5}
\hat{\tau}_{\eff}^{(5)} = \hat{\tau}_0 - \frac{1}{N} \sum_{k=1}^K \sum_{i \in \ci_k} f\left(R_i;\hat{\eta}^{(-k)}\right)\zeta\left(X_i;\hat{\eta}^{(-k)},\hat{\tau}_0,\varphi_0\right)\left(Y_i-\hat{m}^{(-k)}(R_i)\right)
\end{align}
in the outcome-mediated selection bias model $\cp_{\cm_5}$,
where $\eta(\cdot)=(p(\cdot),e(\cdot),m(\cdot),q(\cdot))$.
For convenience in computing estimators for $\tau_{\rct}$, $\tau_{\obs}$, and $\tau_{\tgt}$, we write
\begin{align*}
\zeta(x;\eta,\tau,\varphi_{0,\rct}) & = \frac{\ell'(m_{11}(x))V_{11}(x)+\ell'(m_{10}(x))\frac{e(x)}{1-e(x)}V_{10}(x)}{\rho e(x) D(x;\eta)} \\
\zeta(x;\eta,\tau,\varphi_{0,\obs}) & = \frac{\rho(1-p(x))}{(1-\rho)p(x)}\zeta(x;\eta,\tau,\varphi_{0,\rct}), \quad 
\zeta(x;\eta,\tau,\varphi_{0,\tgt}) = \frac{\rho}{p(x)}\zeta(x;\eta,\tau,\varphi_{0,\rct})
\end{align*}
where $V_{sz}(x)=m_{sz}(x)(1-m_{sz}(x))$ for $(s,z) \in \{0,1\}^2$ since $Y \in \{0,1\}$ and 
\begin{align*}
D(x;\eta) & = V_{11}(x)(\ell'(m_{11}(x)))^2 + \frac{e(x)}{1-e(x)}V_{10}(x)(\ell'(m_{10}(x)))^2 \\
& \quad + \frac{p(x)e(x)}{(1-p(x))q(x)}(\ell'(m_{01}(x)))^2V_{01}(x) + (\ell'(m_{00}(x))^2V_{00}(x)\frac{p(x)e(x)}{(1-p(x))(1-q(x))}.
\end{align*}

Of course,
the efficiency of $\hat{\tau}_{\eff}^{(4)}$ and $\hat{\tau}_{\eff}^{(5)}$ requires that~\eqref{eq:g_hat_minus_g_crossfit} holds.
In Appendix~\ref{app:efficiency_conditions} we provide specific primitive conditions that are sufficient for the conditions of Lemma~\ref{lemma:dml_master} (and thus~\eqref{eq:g_hat_minus_g_crossfit}) to hold.
The main non-regularity requirement in these conditions is that the cross-fit outcome mean function estimates $\hat{m}^{(-k)}$ lie in $\cm$.

We also argue in Appendix~\ref{app:propensity_score_does_not_help} that for the estimands $\tau_{\rct}$, $\tau_{\obs}$, and $\tau_{\tgt}$,
any efficient estimators in the model $\cp_{\cm}$
(for any $\cm$ satisfying the conditions of Theorem~\ref{thm:semiparametric_tangent_space})
is also efficient in the model $\tilde{\cp}_{\cm}$.
In other words,
knowledge of the experimental propensity score does not enable asymptotic variance reductions for $\tau_{\rct}$, $\tau_{\obs}$, and $\tau_{\tgt}$ under \emph{any} outcome mean function restrictions.

\section{Simulations}
\label{sec:simulations_data_fusion}

We study the finite-sample performance of our efficient estimators $\hat{\tau}_{\eff}^{(5)}$ from~\eqref{eq:tau_hat_eff_5} for the estimands $\tau_{\rct}$, $\tau_{\obs}$, and $\tau_{\tgt}$ under outcome-mediated selection bias in two numerical simulation scenarios.
All code and data to reproduce these results and those of the next section can be found at \url{https://github.com/hli90722/combining_experimental_observational}.
We compare performance with both baseline AIP(S)W estimators and the control variate estimators proposed by~\citet{guo2022multi}.
Asymptotically, we know our efficient estimators must dominate the control variate estimators,
which in turn dominate the AIP(S)W estimators as shown in~\citet{guo2022multi}.

The first simulation scenario 
is identical to that considered in Section 6 of~\citet{guo2022multi}.
We call it the ``discrete" scenario,
as it has two covariates $X_1$ and $X_2$ which are i.i.d. Bern(0.5),
with treatment and binary outcomes generated according to logistic models.
The second simulation setting is a ``continuous" scenario where $X_1$ and $X_2$ are independent uniformly distributed random variables on $[-1,1]$,
and there are interaction terms in the logistic models generating treatment and outcomes.
See Appendix~\ref{app:sim_details} for the exact data generating processes.
In both scenarios,
some of the observations are subject to strong outcome-mediated selection
where cases ($Y=1$) are included with probability 0.9 and controls ($Y=0$) are included with probability 0.1.
These observations form a synthetic biased observational dataset of size $m=3000$.
The remaining observations are not subject to selection
and have a size $n_{\rct} \in \{300, 600, 1000, 3000\}$.
This selection process is identical to that considered in~\citet{guo2022multi}.

In both scenarios,
the RCT propensity score $e$ is considered known;
this does not change the efficiency bound,
as stated at the end of the previous section.
In the discrete scenario,
the nuisance functions $m$, $p$, and $q$ are estimated by the appropriate sample averages with 1 and 2 added to the numerator and denominator, respectively.
In the continuous scenario,
we use multivariate adaptive regression splines  (MARS;~\citet{friedman1991multivariate}) in $X$ with logit link and final terms chosen via generalized cross validation as implemented by the \texttt{earth} package in the R language~\citep{earthpackage} to fit the regressions defining the nuisance functions $m_{10}$, $m_{01}$, $m_{00}$, $p$, and $q$.
All MARS regressions are cross-fit with $K=5$ folds.
For stability, we truncate each MARS prediction to the interval $[1/\sqrt{n},1-1/\sqrt{n}]$ for $n$ the number of observations used to fit the corresponding model.
As a reminder,
to ensure the efficiency of our one-step estimators,
we must enforce~\eqref{eq:odds_ratio}.
Hence we set
\[
\hat{m}_{11}^{(-k)}(x) = \ell^{-1}\left(\ell\left(\hat{m}_{10}^{(-k)}(x)\right)+\ell\left(\hat{m}_{01}^{(-k)}(x)\right)-\ell\left(\hat{m}_{00}^{(-k)}(x)\right)\right), \quad k=1,\ldots,K.
\]

We compute our baseline AIPSW estimators $\hat{\tau}_{\ba}$ by solving~\eqref{eq:aipsw_estimating_equations}.
Our efficient one-step estimators $\hat{\tau}_{\eff}^{(5)}$ are computed using~\eqref{eq:tau_hat_eff_5}
with $\hat{\tau}_{\ba}$ as the initial estimator $\hat{\tau}_0$.
We also use $\hat{\tau}_{\ba}$ as the baseline estimator to construct the control variate estimators $\hat{\tau}_{\cv}$ of~\citet{guo2022multi},
and follow that paper's guidance for choosing control variate(s) and computing the adjustment term (see Appendix~\ref{app:cv_detail} for more details).
Their guidance is not based on asymptotic variance considerations,
and we note that using an efficient estimator eliminates the need to choose specific control variates.
Finally, to evaluate the finite sample impacts of estimating the nuisance functions on estimator performance,
we also consider oracle variants of all estimators that plug in the true nuisance functions $\eta^*$.

\begin{table}[!tb]
\centering
\caption{The estimated relative efficiencies of the feasible (resp. oracle) efficient one step estimators $\hat{\tau}_{\eff}^{(5)}$ and the control variate estimators $\hat{\tau}_{\cv}$ compared with the baseline AIP(S)W estimators $\hat{\tau}_{\ba}$ in the discrete simulation scenario as a function of the simulated RCT size $n_{\rct}$.
The parentheses give 95\% bootstrap confidence intervals for the relative efficiency,
reflecting the uncertainty from 1,000 Monte Carlo simulations.
}
\label{table:odds_ratio_discrete_sim}

\begin{tabular}{|c|c||c|c||c|c|}
\hline
Estimand & $n_{\rct}$ & Feasible $\hat{\tau}_{\cv}$ & Feasible $\hat{\tau}_{\eff}^{(5)}$ & Oracle $\hat{\tau}_{\cv}$ & Oracle $\hat{\tau}_{\eff}^{(5)}$ \\
\hline
\multirow{4}{0.15\linewidth}{\centering $\tau_{\rct}$} & 300 & 2.37 (2.18, 2.54) & \textbf{3.23} (2.86, 3.56) & 2.19 (2.04, 2.33) & \textbf{3.50} (3.11, 3.85) \\
& 600 & 1.94 (1.77, 2.10) & \textbf{2.29} (2.05, 2.51)  & 1.94 (1.76, 2.09) & \textbf{2.39} (2.15, 2.61) \\
& 1000 & 1.62 (1.50, 1.74) & \textbf{1.79} (1.64, 1.93) & 1.64 (1.51, 1.75) & \textbf{1.82} (1.67, 1.96)\\
& 3000 & 1.05 (0.96, 1.13) & \textbf{1.26} (1.18, 1.33) & 1.05 (0.96, 1.12) & \textbf{1.26} (1.19, 1.34) \\
\hline
\multirow{4}{0.15\linewidth}{\centering $\tau_{\obs}$} & 300 & 2.73 (2.49, 2.95) & \textbf{3.89} (3.40, 4.30) & 2.48 (2.29, 2.66) & \textbf{4.34} (3.83, 4.80) \\
& 600 & 2.16 (1.96, 2.33) & \textbf{2.57} (2.29, 2.83)  & 2.15 (1.95, 2.33) & \textbf{2.67} (2.39, 2.93) \\
& 1000 & 1.78 (1.64, 1.91) & \textbf{1.96} (1.79, 2.11) & 1.77 (1.64, 1.89) & \textbf{2.00} (1.83, 2.15)\\
& 3000 & 1.09 (1.02, 1.16) & \textbf{1.25} (1.17, 1.33) & 1.08 (1.00, 1.14) & \textbf{1.25} (1.18, 1.32) \\
\hline
\multirow{4}{0.15\linewidth}{\centering $\tau_{\tgt}$} & 300 & 2.74 (2.50, 2.95) & \textbf{3.91} (3.42, 4.32) & 2.48 (2.29, 2.65) & \textbf{4.37} (3.85, 4.82) \\
& 600 & 2.16 (1.96, 2.34) & \textbf{2.58} (2.30, 2.84)  & 2.15 (1.95, 2.33) & \textbf{2.69} (2.40, 2.95) \\
& 1000 & 1.77 (1.63, 1.90) & \textbf{1.95} (1.78, 2.11) & 1.76 (1.63, 1.89) & \textbf{1.99} (1.82, 2.14)\\
& 3000 & 1.08 (1.00, 1.15) & \textbf{1.27} (1.19, 1.35) & 1.07 (0.98, 1.14) & \textbf{1.27} (1.20, 1.35) \\
\hline
\end{tabular}
\end{table}

\begin{table}[!tb]
\centering
\caption{Same as Table~\ref{table:odds_ratio_discrete_sim} but for the continuous simulation scenario.
}
\label{table:odds_ratio_continuous_sim}

\begin{tabular}{|c|c||c|c||c|c|}
\hline
Estimand & $n_{\rct}$ & Feasible $\hat{\tau}_{\cv}$ & Feasible $\hat{\tau}_{\eff}^{(5)}$ & Oracle $\hat{\tau}_{\cv}$ & Oracle $\hat{\tau}_{\eff}^{(5)}$ \\
\hline
\multirow{4}{0.15\linewidth}{\centering $\tau_{\rct}$} & 300 & 1.49 (1.42, 1.55) & \textbf{2.47} (2.19, 2.74) & 1.50 (1.42, 1.57) & \textbf{3.28} (2.95, 3.58) \\
& 600 & 1.44 (1.38, 1.50) & \textbf{1.99} (1.77, 2.18)  & 1.46 (1.40, 1.53) & \textbf{2.34} (2.12, 2.54) \\
& 1000 & 1.45 (1.38, 1.52) & \textbf{1.82} (1.63, 1.98) & 1.44 (1.37, 1.50) & \textbf{1.91} (1.75, 2.07)\\
& 3000 & 1.12 (1.10, 1.15) & \textbf{1.27} (1.18, 1.35) & 1.12 (1.10, 1.15) & \textbf{1.28} (1.20, 1.35) \\
\hline
\multirow{4}{0.15\linewidth}{\centering $\tau_{\obs}$} & 300 & 1.50 (1.42, 1.57) & \textbf{2.61} (2.28, 2.90) & 1.58 (1.50, 1.66) & \textbf{4.81} (4.25, 5.29) \\
& 600 & 1.46 (1.38, 1.53) & \textbf{2.10} (1.85, 2.32)  & 1.57 (1.48, 1.64) & \textbf{2.73} (2.46, 2.99) \\
& 1000 & 1.45 (1.39, 1.52) & \textbf{1.96} (1.74, 2.15) & 1.50 (1.43, 1.57)) & \textbf{2.15} (1.95, 2.33)\\
& 3000 & 1.10 (1.07, 1.12) & \textbf{1.24} (1.13, 1.33) & 1.10 (1.07, 1.12) & \textbf{1.28} (1.20, 1.35) \\
\hline
\multirow{4}{0.15\linewidth}{\centering $\tau_{\tgt}$} & 300 & 1.51 (1.43, 1.58) & \textbf{2.66} (2.32, 2.96) & 1.58 (1.50, 1.66) & \textbf{4.81} (4.26, 5.30) \\
& 600 & 1.48 (1.41, 1.55) & \textbf{2.12} (1.88, 2.35)  & 1.57 (1.49, 1.65) & \textbf{2.76} (2.48, 3.02) \\
& 1000 & 1.49 (1.42, 1.56) & \textbf{1.99} (1.77, 2.18) & 1.51 (1.44, 1.58) & \textbf{2.17} (1.97, 2.36) \\
& 3000 & 1.12 (1.10, 1.15) & \textbf{1.26} (1.16, 1.36) & 1.12 (1.09, 1.14) & \textbf{1.30} (1.22, 1.38) \\
\hline
\end{tabular}
\end{table}

From Tables~\ref{table:odds_ratio_discrete_sim} and~\ref{table:odds_ratio_continuous_sim},
we see the oracle one step estimator $\hat{\tau}_{\eff}^{(5)}$ has higher relative efficiency (i.e. lower MSE) than the corresponding oracle control variate estimator $\hat{\tau}_{\cv}$
for all three estimands,
all RCT sizes,
and both simulation scenarios (discrete and continuous).
Relative efficiency for $\hat{\tau}_{\cv}$ and $\hat{\tau}_{\eff}^{(5)}$ is defined as the MSE of the baseline $\hat{\tau}_{\ba}$ divided by the MSE of the estimator.
The computation is done separately for oracle variants of the estimators and feasible ones (i.e., those that estimate the nuisance functions as described above).
We find that estimating the nuisance functions leads to a similar degree of performance deterioration between the oracle and feasible versions for all estimators in the discrete simulations.
In the continuous scenario where the nuisance estimates are harder to compute,
the oracle and feasible relative efficiencies of $\hat{\tau}_{\cv}$ are quite similar.
This suggests that the control variate estimator $\hat{\tau}_{\cv}$ suffers a similar amount of performance degradation from nuisance estimation as the baseline estimator $\hat{\tau}_{\ba}$.
By contrast,
the one-step efficient estimator $\hat{\tau}_{\eff}^{(5)}$ exhibits a notable relative efficiency decline between the oracle and feasible variants,
particularly with smaller sample sizes.
This suggests that the one-step efficient estimator $\hat{\tau}_{\eff}^{(5)}$ is more sensitive to poor nuisance estimation in finite samples than the baseline and control variate estimators.
Nevertheless, the feasible efficient estimator $\hat{\tau}_{\eff}$ has much lower MSE (10-40\% reduction) than the feasible control variate estimator $\hat{\tau}_{\cv}$ across all estimands, RCT sizes, and scenarios.
This suggests the asymptotic variance gap between the control variate and efficient estimators is quite large.

\section{Data example}
\label{sec:tennessee_star}

We now consider a real data example
based on estimating $\tau_{\obs}$ in data from the Tennessee Student Teacher Achievement Ratio (STAR) study.
This is inspired by~\citet{kallus_removing_2018},
who use the same dataset and the linear confounding bias assumption~\eqref{eq:parametric_confounding_bias} to study treatment effect hetereogeneity.
The Tennessee STAR study was a large RCT designed to investigate whether small class sizes improve test outcomes.
For simplicity,
we only consider three binary covariates: gender, race (white or non-white), and whether the student received free lunch in grade 1.
We also include birth date as a continuous covariate.

Following~\citet{kallus_removing_2018},
the outcome is the sum of listening, reading, and math scores 
(divided by 100) at the end of first grade.
We randomly subsample 1405 students from rural or inner-city schools to be our RCT dataset,
and generate an observational dataset (with different covariate distribution) by combining the observations from the remaining 1406 students from rural or inner-city schools with 1407 students from urban or suburban schools.
We randomly set aside 30\% (1125) of the observations in the observational dataset as a ``validation set'' to represent the target population.
In the remaining 1688 observations,
we synthetically introduce confounding by reducing all outcomes above the median by the difference between the average outcome above the median and the average outcome below the median.
The nuisance functions $m_{10}$, $m_{01}$, $m_{00}$, $p$, and $q$ are computed using MARS and cross-fitting ($K=5$) in the same way as in the simulations of Section~\ref{sec:simulations_data_fusion},
except that for the outcome mean function estimates,
we now use the identity link and do not truncate.
The three binary predictors are encoded as 0 or 1.
We enforce~\eqref{eq:parametric_confounding_bias} by setting
\[
\hat{m}_{11}^{(-k)}(x) = \hat{m}_{10}^{(-k)}(x)+\hat{m}_{01}^{(-k)}(x) - \hat{m}_{00}^{(-k)}(x) + \psi(x)^{\top}\hat{\theta}, \quad k=1,\ldots,K
\]
where $\psi(x)$ consists of an intercept term and each of the covariates (additively)
and
\[
\hat{\theta} = \argmin_{\theta} \frac{1}{N} \sum_{k=1}^K \sum_{i \in \ci_k} S_i\left(Y_i\left[\frac{Z_i}{e^*}-\frac{1-Z_i}{1-e^*}\right]-\hat{m}_{01}^{(-k)}(X_i)-\hat{m}_{00}^{(-k)}(X_i)-\psi(X_i)^{\top}\theta\right)^2
\]
following~\citet{kallus_removing_2018},
where $e^*$ is the known constant RCT propensity score.
The astute reader might notice that
$\hat{\theta}$ is not cross-fit so this makes $\hat{m}_{11}^{(-k)}$ not cross-fit either.
This does not violate the theoretical guarantees of efficiency of our estimators, however,
since $\theta \mapsto \psi(x)^{\top}\theta$ is a Donsker class.
The variance functions $V_{11}$, $V_{10}$, $V_{01}$, $V_{00}$ are computed using a generalized linear model with logarithmic link predicting the squared residuals $(Y_i-\hat{m}^{(-k(i))}(R_i))^2$ from the covariates.
This is similar to~\citet{yang2024datafusion}.
Finally,
the estimators $\hat{\tau}_{\eff}^{(4)}$ and the baseline AIPSW estimator $\hat{\tau}_{\ba}$ are computed on a set of ``training observations'' consisting of a randomly sampled fraction $\xi \in \{0.2,0.4,0.6,0.8,1\}$ of our RCT dataset along with the entire confounded observational dataset of 1688 subjects.
The one-step estimator $\hat{\tau}_{\eff}^{(4)}$ is derived assuming~\eqref{eq:parametric_confounding_bias} holds with $\hat{\tau}_{\ba}$ as the initial estimator $\hat{\tau}_0$,
with influence function $\varphi_0$.

Taking an AIPSW estimate of $\tau_{\obs}$ computed via~\eqref{eq:aipsw_estimating_equations} using the validation set to be the true estimand $\tau_{\obs}$ for evaluation purposes,
we observe substantial precision gains from using the efficient one-step estimator $\hat{\tau}_{\eff}^{(4)}$ instead of the baseline AIPSW estimator $\hat{\tau}_{\ba}$ for all RCT sample sizes $\xi$.
The relative variance improvement from using the efficient one-step estimator appears to generally increase with $\xi$.
When the full RCT dataset is used ($\xi=1$), 
we observe a 75\% variance reduction in 1000 bootstrap replicates,
though there is substantial uncertainty in that estimate (Table~\ref{table:tennessee_star}).
Since this is a real data example,
we do not know~\eqref{eq:parametric_confounding_bias} to hold,
yet assuming~\eqref{eq:parametric_confounding_bias} (as $\hat{\tau}_{\eff}^{(4)}$ does)
appears to introduce a negligible amount of bias into $\hat{\tau}_{\eff}^{(4)}$ (Fig.~\ref{fig:tennessee_star}).
In other words,
the MSE reductions in Table~\ref{table:tennessee_star} from using $\hat{\tau}_{\eff}$ can be attributed almost entirely to variance reduction,
and are minimally offset by bias.

\begin{figure}[!tb]
\begin{center}
\includegraphics[width=\linewidth]{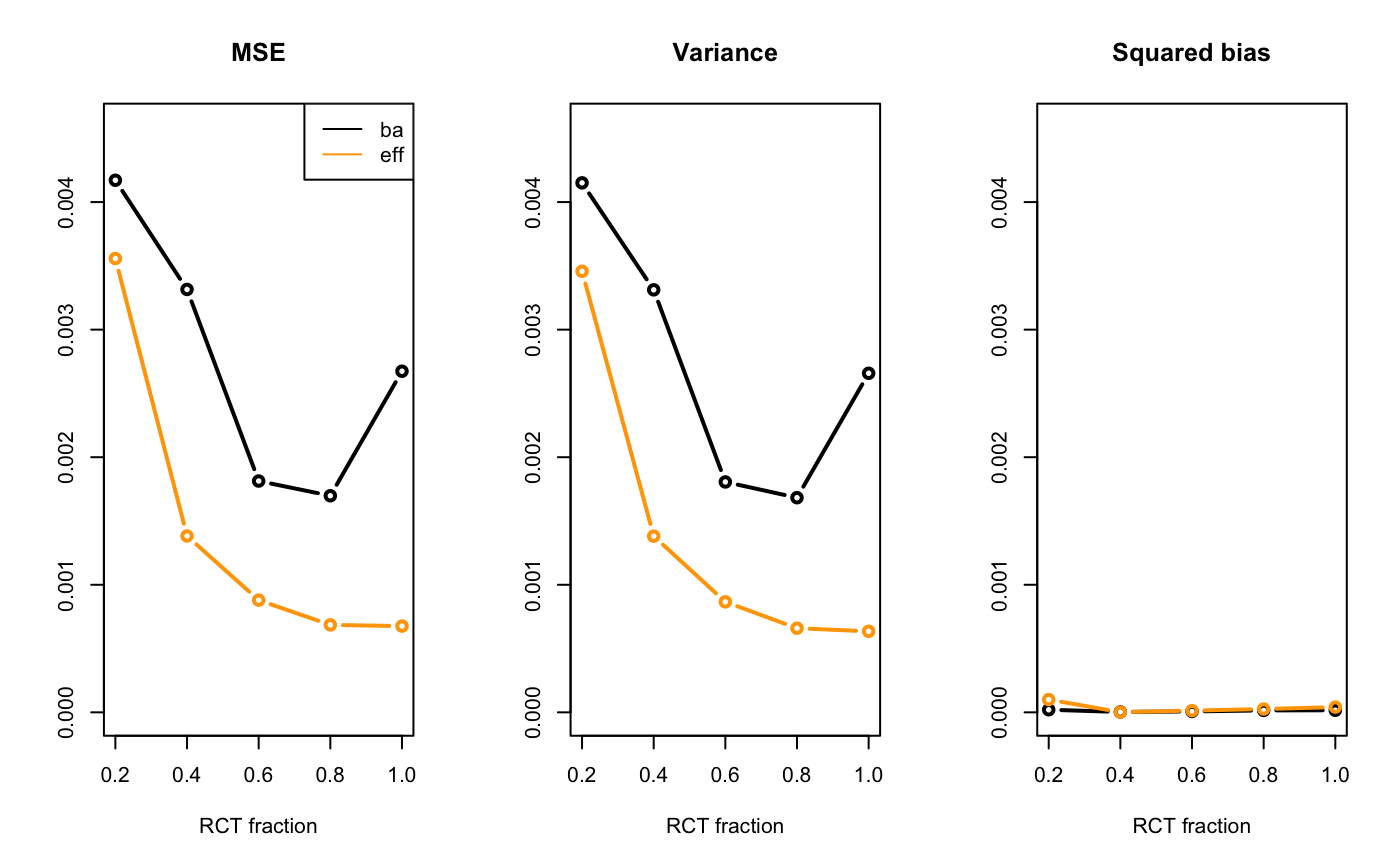}
\end{center}
\caption{The mean squared error (MSE), variance, and squared bias of our efficient one-step estimator $\hat{\tau}_{\eff}^{(4)}$ for $\tau_{\obs}$ and the baseline AIPSW estimator $\hat{\tau}_{\ba}$ in the Tennessee STAR data example described in Section~\ref{sec:tennessee_star},
as a function of the fraction $\xi$ of the entire RCT used.}
\label{fig:tennessee_star}
\end{figure}

\begin{table}[!tb]
\centering
\caption{The relative efficiency of $\hat{\tau}_{\eff}$ compared to $\hat{\tau}_{\ba}$ for estimating $\tau_{\obs}$ in the Tennessee STAR data example of Section~\ref{sec:tennessee_star},
as a function of the fraction $\xi$ of the RCT data used.
Confidence intervals for the relative efficiency are computed using 1000 bootstrap replications.
}
\label{table:tennessee_star}

\begin{tabular}{|c|c|}
\hline
RCT fraction $\xi$ & Relative efficiency (95\% CI) \\
\hline
0.2 & 1.17 (1.01, 1.31) \\
0.4 & 2.40 (1.65, 2.97) \\
0.6 & 2.06 (1.69, 2.36) \\
0.8 & 2.48 (1.72, 3.03) \\
1.0 & 3.95 (1.79, 5.49) \\
\hline
\end{tabular}
\end{table}

\section{Discussion}
\label{sec:discussion_data_fusion}

We have provided a general semiparametric theory for supervised learning
under generic restrictions on the outcome mean function.
Many assumptions used for data fusion in the recent causal inference literature are covered by this theory,
which we hope provides a valuable unifying perspective.
Sometimes these assumptions have not been fully leveraged for maximal variance reduction,
as illustrated by the inefficiency of control variate approaches in the outcome-mediated selection bias setting.
This paper has sought to present a streamlined way to derive fully efficient estimators under a large class of restrictions on outcome mean functions,
while also providing additional mathematical insights into the nature of efficiency bounds under such restrictions.

Our setting does not handle all types of data fusion restrictions that have been of interest.
For instance,
some authors have imposed structure on quantile treatment effects~\citep{athey_2023_semiparametric} or shape restrictions on the density function of a continuous outcome~\citep{li2023efficient},
which cannot be formulated as restrictions on the outcome mean function.
Other recent work~\citep{li2023efficient, li2025data} has also aimed to provide general data fusion frameworks guided by semiparametric efficiency.
These frameworks, however, do not cover many settings of interest that can be cast as restrictions on outcome mean functions.
For instance, consider our data fusion setting with $R=(S,Z,X)$ and suppose $Y \in \{0,1\}$,
so that the outcome mean function in fact specifies the entire conditional distribution of $Y$ given $(S,Z,X)$.
Then for the observational dataset to be used at all in estimating any parameter depending on the conditional distribution of $Y \mid Z,X,S=1$ (i.e. the experimental outcome distribution),~\citet{li2023efficient} require an ``alignment" assumption that amounts to $m(s,z,x)$ being independent of $s$.
The ``weak alignment" assumption of~\citet{li2025data} weakens this to allow 
\[
\ell(m(1,z,x)) = \ell(m(0,z,x)) + \beta^{\top} \psi(z,x)
\]
for some known low-dimensional basis expansion $\psi$.
If $x$ is continuous,
however, this is not flexible enough to cover the outcome-mediated selection bias setting of Example~\ref{ex:outcome_selection_bias},
which is equivalent to $\ell(m(1,z,x))-\ell(m(0,z,x)) = f(x)$, $z=0,1$ for arbitrary $f \in \ch_X$.
Conversely,
if $Y$ is continuous, the weak alignment assumption induces restrictions on higher moments of the conditional distribution of $Y$ given $R$ which cannot be expressed solely as restrictions on the outcome mean function.

Asymptotic efficiency also has some limitations as an objective for precise inference in finite samples.
Based on the simulations in Section~\ref{sec:simulations_data_fusion},
the one-step efficient estimators proposed in this work can have worse relative finite sample performance than asymptotics might suggest,
which we attribute to the need to estimate additional nuisance functions relative to the baselines.
Thus, we believe extensions of ideas from the nonparametric literature to improve nuisance estimation  such as
pooling the datasets~\citep{gagnon-bartsch_precise_2023,huang_leveraging_2023,liao_transfer_2023}
and Riesz regression~\citep{klosin2021automatic,chernozhukov2022automatic, lee2025rieszboost} may be useful to extend to models with outcome mean function restrictions.
Substantial recent work has also exhibited finite sample benefits to non-regular estimators combining experimental and observational datasets~\citep{chen_minimax_2021,cheng_adaptive_2021,yang2023elastic,dang_cross-validated_2023,lin_many_2023, rosenman2023combining}.
While such estimators are, in some sense, more robust in that they don't require assumptions on outcome mean functions,~\citet{oberst_understanding_2023} showed that many of these estimators will end up performing worse than the efficiency bound in the nonparametric model when the bias is larger than some constant multiple of $N^{-1/2}$~\citep{oberst_understanding_2023}.
Further work is needed to more precisely understand when such estimators may be beneficial.

\section*{Acknowledgments}
The author is funded by a Stanford Interdisciplinary Graduate Fellowship (SIGF). The author thanks Art Owen and several anonymous reviewers for careful and insightful comments that greatly improved this manuscript. He is also grateful for extremely helpful feedback on this work from Tim Morrison, Stefan Wager, Shu Yang, and Jann Spiess.

\appendix

\section{Proofs}
Here we collect the proofs of the results stated in the main text.

\subsection{Preliminaries}
\label{app:preliminaries}
First,
we state and prove some useful elementary results in asymptotic statistics. 

\begin{lemma}
\label{lemma:conditional_oh_pee}
Let $X_n$ be a sequence of random vectors and $\{\cf_n, n \geq 1\}$ be a sequence of $\sigma$-algebras such that $\e\bigl[\|X_n\| \mid \cf_n\bigr] = o_p(1)$.
Then $X_n = o_p(1)$.
\end{lemma}
\begin{proof}[Proof of Lemma~\ref{lemma:conditional_oh_pee}]
Fixing $M>0$, we have $M\indic(\|X_n\| > M) \leq \|X_n\|$ for all $n$.
Taking conditional expectations given $\cf_n$ on both sides we have
\begin{equation}
\label{eq:cond_markov}
P\bigl(\|X_n\| > M \mid \cf_n\bigr) \leq M^{-1}\e\bigl[\|X_n\| \mid \cf_n\bigr].
\end{equation}
Thus if $\e\bigl[\|X_n\| \mid \cf_n\bigr] = o_p(1)$ we have $\Pr(\|X_n\|>M \mid \cf_n) = o_p(1)$ as well.
But $\Pr(\|X_n\|>M \mid \cf_n)$ is uniformly bounded so its expectation converges to zero, i.e., $\Pr(\|X_n\|>M)=o(1)$. 
Since $M>0$ was arbitrary we conclude that $X_n=o_p(1)$.
\end{proof}

\begin{lemma}
\label{lemma:multiplication}
Let $\{\hat{f}_n\}_{n=1}^{\infty}$, $\{f_n\}_{n=1}^{\infty}$, $\{\hat{g}_n\}_{n=1}^{\infty}$, and $\{g_n\}_{n=1}^{\infty}$ be sequences of $P$-square integrable functions with $\|\hat{f}_n-f_n\|_{2,P} = o_p(a_n)$ and $\|\hat{g}_n-g_n\|_{2,P} = o_p(b_n)$ for some sequences $a_n \downarrow 0$, $b_n \downarrow 0$.
Further suppose that there exists $C < \infty$ independent of $n$ for which $\|\hat{f}_n\|_{\infty} + \|g_n\|_{\infty} \leq C$ for all $n \geq 1$.
Then $\|\hat{f}_n\hat{g}_n-f_ng_n\|_{2,P} = o_p(a_n) + o_p(b_n)$.
\begin{proof}
By Minkowksi's inequality we have
\begin{align*}
\|\hat{f}_n\hat{g}_n-f_ng_n\|_{2,P} & \leq  \|(\hat{f}_n-f_n)g_n\|_{2,P} + \|\hat{f}_n(\hat{g}_n-g_n)\|_{2,P} \\
& \leq  C\|\hat{f}_n-f_n\|_{2,P} + C\|\hat{g}_n-g_n\|_{2,P} \\
& = o_p(a_n) + o_p(b_n).
\end{align*}
\end{proof}
\end{lemma}

\begin{lemma}
\label{lemma:conditional_oh_pee_fn_f}
Let $X_1,\ldots,X_n$ be i.i.d. random variables taking values in a space $\cx$ with distribution $P$ and let $\|f\|_{2,P} = (\int f^2(x)dP(x))^{1/2}$.
Suppose $\hat{f}_n(\cdot)$ is a sequence of random measurable functions on $\cx$ independent of $X_1,\ldots,X_n$ for which
$\|\hat{f}_n-f\|_{2,P} = o_p(r_n)$ as $n \rightarrow \infty$ for some sequence $\{r_n\}_{n=1}^{\infty}$.
Then 
\[
\sqrt{\frac{1}{n} \sum_{i=1}^n (\hat{f}_n(X_i)-f(X_i))^2} = o_p(r_n).
\]
\end{lemma}
\begin{proof}
Define the quantity
\[
Y_n = \frac{1}{n} \sum_{i=1}^n (\hat{f}_n(X_i)-f(X_i))^2
\]
and note that
\[
\e(Y_n \mid \hat{f}_n) = \|\hat{f}_n-f\|_{2,P}^2 = o_p(r_n^2)
\]
Then $Y_n=o_p(r_n^2)$ by Lemma~\ref{lemma:conditional_oh_pee}, and the result follows.
\end{proof}
\begin{lemma}
\label{lemma:mean_0_times_small}
Suppose $(A_1,B_1),(A_2,B_2),\ldots$ are random variables with $\e[A_n] < \infty$ for all $n \geq 1$. Let $\{\cf_n, n \geq 1\}$ be a sequence of $\sigma$-algebras such that for each $n$,
$(A_1,\ldots,A_n)$ is measurable with respect to $\cf_n$ and $B_1,\ldots,B_n$ are conditionally independent given $\cf_n$. Further assume that $\e[B_i \mid \ca_n] = 0$ and that there exists $C<\infty$ such that $\e[B_i^2 \mid \ca_n] \leq C$ for all $n \geq 1$, $1 \leq i \leq n$.
Finally suppose that $n^{-1} \sum_{i=1}^n A_i^2 = o_p(r_n^2)$ as $n \rightarrow \infty$ for some positive sequence $r_n$.
Then 
\[
\frac{1}{n} \sum_{i=1}^n A_iB_i = o_p(r_n/\sqrt{n}).
\]
\end{lemma}
\begin{proof}
We have
\[
\e\left[\frac{1}{n} \sum_{i=1}^n A_iB_i \mid \ca_n\right] = \frac{1}{n} \sum_{i=1}^n A_i \e[B_i \mid \ca_n] = 0
\]
and hence
\begin{align*}
\e\left[\left(\frac{1}{n} \sum_{i=1}^n A_iB_i\right)^2 \Bigg| \ca_n \right] & = \frac{1}{n^2} \var\left(\sum_{i=1}^n A_iB_i \Big| \ca_n \right) = \frac{1}{n^2} \sum_{i=1}^n A_i^2\e[B_i^2 \mid \ca_n] \leq \frac{C}{n^2} \sum_{i=1}^n A_i^2 = o_p(n^{-1}r_n^2)
\end{align*}
It follows by Lemma~\ref{lemma:conditional_oh_pee} that
\[
\frac{1}{r_n\sqrt{n}} \sum_{i=1}^n A_iB_i = o_p(1)
\]
as desired.
\end{proof}

\subsection{Proof of Theorem~\ref{thm:semiparametric_tangent_space}}
\label{app:semiparametric_tangent_space_proof}
Our argument extends many ideas from Section 4.5 of~\citet{tsiatis2006semiparametric} for the restricted moment model.
As in that text,
we impose the following mild technical conditions on the true data-generating distribution $P^*$ and the model $\cp_{\cm}$:
\begin{assumption}
\label{assump:MGF}
The moment generating function $\zeta(t) = \e[\exp(tY)]$ exists in a neighborhood of 0.
\end{assumption}
\begin{assumption}
\label{assump:Var}
There exist $A < \infty$ and $\delta > 0$ for which $\e[((Y-m^*(R)))^2\indic(|(Y-m^*(R))| \leq A) \mid R] \geq \delta$ and $V^*(R) \leq A$ with probability 1.
\end{assumption}

We begin by characterizing the space $\ct_{\cm}$.
An arbitrary parametric submodel $P_{\gamma}$ for the model $\cp_{\cm}$ is a distribution with density
\begin{equation}
\label{eq:generic_submodel}
p(w;\gamma) = f(y \mid r;\gamma_1)g(r;\gamma_2)
\end{equation}
with respect to some dominating measure $\mu \times \lambda$ on $\real \times (\{0,1\}^2 \times \cx)$;
here $f$ denotes the conditional density of $Y$ given $R$ and $g$ is the density of $R$.
Since the density~\eqref{eq:generic_submodel} must correspond to a distribution in $\cp_{\cm}$,
the function $r \mapsto m(r;\gamma_1)$ lies in $\cm$ for all $\gamma_1$ in a neighborhood of $0 \in \real^q$,
for some $q \geq 1$.
Since the component $f(y \mid r;\gamma_1)$ is variationally independent of $g(r;\gamma_2)$,
it follows that we can write
\[
\ct_{\cm} = \ct_1 \oplus \ct_2
\]
where $\ct_1$ is the closure of the span of the score functions of the densities $f(y \mid r;\gamma_1)$ over all submodels
and $\ct_2$ is the closure of the span of the score functions of the densities $g(r;\gamma_2)$ over all submodels
(see the discussion preceding Theorem 4.5 of~\citet{tsiatis2006semiparametric} or the proof of the main results of~\citet{li2023efficient} for a similar argument).
Since the model $\cp_{\cm}$ places no restrictions on the distribution of $R$,
it follows that $\ct_2 = \ch_R^0$ by Theorem 4.5 of~\citet{tsiatis2006semiparametric}.

We now claim that 
\begin{equation}
\label{eq:T_1}
\ct_1 = \{h \in \ch^0 \mid r \mapsto \e[Yh(W) \mid R=r] \in \cs_M, \e[h(W) \mid R]=0\} .
\end{equation}
The space on the right-hand side of~\eqref{eq:T_1},
which we call $\ct_1^{\dag}$,
is evidently closed since $\cs_{\cm}$ is a closed subspace of $\ch_R$.
Hence to show that $\ct_1 \subseteq \ct_1^{\dag}$,
it suffices to show that any element in the span of the score function of an arbitrary parametric submodel $f(y \mid r;\gamma_1)$ lies in $\ct_1^{\dag}$.
Utilizing the notation
\[
\frac{\partial h(a)}{\partial \gamma} = \left(\frac{\partial}{\partial \gamma_1}h(\gamma),\ldots, \frac{\partial}{\partial \gamma_r}h(\gamma)\right)^{\top} \Big|_{\gamma=a} \in \real^q.
\]
for an arbitrary smooth function $h:\real^q \rightarrow \real$,
such an element is given by
\[
h(w) = c^{\top} \frac{\partial \log f(y \mid r;0)}{\partial \gamma_1} = c^{\top} \frac{\frac{\partial}{\partial \gamma_1} f(y \mid r;0)}{f^*(y \mid r)}.
\]
By the smoothness conditions for parametric submodels (see~\citet{newey1990semiparametric}),
we know that
\[
\frac{\partial m(r;0)}{\partial \gamma_1}  = \int y \frac{\partial}{\partial \gamma_1} f(y \mid r;0) d\mu(y)
\]
for all $r$ and is square integrable.
Hence $r \mapsto m(r;\gamma_1)$ is an outcome mean function parametric submodel for $\cm$ indexed by $\gamma_1$.
With
\[
\e[Yh(W) \mid R=r] = c^{\top} \int y\frac{\partial}{\partial \gamma_1} f(y \mid r;0) d\mu(y) = c_1^{\top}\frac{\partial}{\partial \gamma_1} m(r;0) 
\]
we conclude that $r \mapsto  \e[Yh(W) \mid R=r] \in \cs_{\cm}$ by the definition of $\cs_{\cm}$.
Furthermore 
\[
\e[h_1(W) \mid R=r] = c_1^{\top} \int  \frac{\partial}{\partial \gamma_1} f(y \mid r;0) d\mu(y) = c_1^{\top} \frac{\partial}{\partial \gamma_1} \int f(y \mid r;\gamma_1) d\mu(y) = 0,
\]
so that $h \in \ct_1^{\dag}$.

Now,
we show the reverse inclusion,
i.e. $\ct_1^{\dag} \subseteq \ct_1$.
To that end we consider arbitrary $h_1 \in \ct_1^{\dag}$  
so that $\e[h_1(W) \mid R]=0$ and $ r \mapsto \int (y-m^*(r))h_1(w) f^*(y \mid r)d\mu(y) \in \cs_{\cm}$.
Considering the latter condition,
by definition of $\cs_{\cm}$
there exists a sequence of $r_n$-dimensional outcome mean function parametric submodels $m^{(n)}(r;\gamma_1^{(n)})$ of $\cm$ and constant vectors $c_n \in \real^{r_n}$ for which 
\[
\tilde{h}_1^{(n)}(r) := c_n^{\top} \frac{\partial m^{(n)}(r;0)}{\partial \gamma_1^{(n)} }
\]
converges to 
\[
\tilde{h}_1(r) := \int (y-m^*(r))h_1(w) f^*(y \mid r)d\mu(y)
\]
in mean square as $n \rightarrow \infty$.
In fact, we can assume that
each $\tilde{h}_1^{(n)}$
is a \emph{bounded} element of $\ch_R$\footnote{
Given the outcome mean function parametric submodels $m^{(n)}(r;\gamma_1^{(n)})$ from the main text,
define
\begin{equation}
\label{eq:omega_star}
\omega^*(r) = \int y^3 f^*(y \mid r) d\mu(y) - m^*(r)\int y^2 f^*(y \mid r) d\mu(y)
\end{equation}
Then consider the truncated submodel defined by
\[
\tilde{m}^{(n)}\left(r;\gamma_{1}^{(n)}\right) := 
\begin{cases}
m^*(r) & \text{ if $\sup_{0<\|a\| \leq b_n} \frac{\left|m^{(n)}\left(r;a\right)-m^*(r)\right|}{\|a\|} > a_n$ or $|\omega^*(r)| > a_n$} \\
m^{(n)}\left(r;\gamma_1^{(n)}\right) & \text{ otherwise}
\end{cases}
\]
where $a_n \uparrow \infty$ and $b_n \downarrow 0$ are chosen such that
\[
\lim_{n \rightarrow \infty} \e\left[\left(\tilde{h}_1^{(n)}(R)\right)^2\indic\left(\sup_{0<\|a\| \leq b_n} \frac{\left|m^{(n)}\left(R;a\right)-m^*(R)\right|}{\|a\|} > a_n \text{ or } |\omega^*(R)| > a_n\right) \right] = 0
\]
Such a sequence exists by dominated convergence since each $\tilde{h}_1^{(n)}(R)$ is square integrable
and the existence of $\partial/\partial \gamma_1^{(n)} m^{(n)}\left(r;\gamma_1^{(n)}\right)$ for almost all $r$ ensures that the indicator
\[
\indic\left(\sup_{0<\|a\| \leq \tilde{b}} \frac{\left|m^{(n)}\left(R;a\right)-m^*(R)\right|}{\|a\|} > \tilde{a} \text{ or } |\omega^*(R)| > \tilde{a}\right)
\]
converges to zero with probability 1 as $\tilde{a} \uparrow \infty$ then $\tilde{b} \downarrow \infty$.
Then $c_n^{\top} \frac{\partial \tilde{m}^{(n)}(r;0)}{\partial \gamma_1^{(n)}}$ is bounded above in absolute value by $a_n\|c_n\|$ for all $r$ yet still converges to $\tilde{h}_1$ in mean square.
Note: the truncation by $\omega^*(r)$ is not necessary here,
but is useful for the next footnote.}.
Now define
\[
\bar{h}_1(w) := h_1(w)-\frac{(y-m^*(r))\tilde{h}_1(r)}{V^*(r)}.
\]
and note that
\[
\int \bar{h}_1(w)f^*(y \mid r)d\mu(y) = \int y\bar{h}_1(w)f^*(y \mid r)d\mu(y) = 0.
\]
Consider 
\[
\bar{h}_1^{(n)}(w) = H^{(n)}(w)-\int H^{(n)}(r,y')f^*(y' \mid r)d\mu(y')
\]
where
\[
H^{(n)}(w) = h_1^{(n)}(w) - \frac{(y-m^*(r))\indic(|(y-m^*(r))| \leq A_n) \int (y'-m^*(r))h_1^{(n)}(r,y')f^*(y' \mid r) d\mu(y')}{\int (y'-m^*(r))^2\indic(|y'-m^*(r)| \leq A_n) f^*(y' \mid r)d\mu(y')} 
\]
for $h_1^{(n)}(w) =  h_1(w)\indic(|h_1(w)| \leq A_n)$,
$A_n = \max(A, n)$.
Then one can verify $\bar{h}_1^{(n)}$ is a sequence of bounded functions converging in mean square to $\bar{h}_1$ satisfying
\begin{equation}
\label{eq:h_1_n_bar}
\int \bar{h}_1^{(n)}(w)f^*(y \mid r)d\mu(y) = \int y\bar{h}_1^{(n)}(w)f^*(y \mid r)d\mu(y) = 0
\end{equation}
for all $n \geq 1$.
Now for each $n \geq 1$,
we can define the exponentially tilted density family
\[
f^{(n)}\left(y \mid r;\gamma_0,\gamma_1^{(n)}\right) = \frac{f^*(y \mid r)\left(1+\gamma_0\bar{h}_1^{(n)}(w)\right)\exp\left(c^{(n)}\left(r;\gamma_0,\gamma_1^{(n)}\right)y\right)}{\int f^*(y' \mid r)\left(1+\gamma_0\bar{h}_1^{(n)}(r,y')\right)\exp\left(c^{(n)}\left(r;\gamma_0,\gamma_1^{(n)}\right)y'\right) d\mu(y') }
\]
where $c^{(n)}\left(r;\gamma_0,\gamma_1^{(n)}\right)$ is chosen such that
\begin{equation}
\label{eq:tilted_mean}
\int yf^{(n)}\left(y \mid r;\gamma_0,\gamma_1^{(n)}\right) d\mu(y) = m^{(n)}\left(r;\gamma_1^{(n)}\right)
\end{equation}
for all $r$ and all $\left(\gamma_0,\gamma_1^{(n)}\right)$ in a neighborhood of zero\footnote{To see that this is possible, define
\[
\iota(c;r,\gamma_0) = \frac{\int y f^*(y \mid r)\left(1+\gamma_0\bar{h}_1^{(n)}(w)\right)\exp(cy) d\mu(y)}{\int f^*(y \mid r)\left(1+\gamma_0\bar{h}_1^{(n)}(w)\right)\exp(cy) d\mu(y) }.
\]
By the previous footnote we can assume there is some $a_n < \infty$ and $b_n > 0$ for which $m^{(n)}\left(r;\gamma_1^{(n)}\right) = m^*(r)$ for all $r$ such that $|\omega^*(r)| > a_n$ or  
$\sup_{0<\|a\| \leq b_n} \frac{\left|m^{(n)}\left(R;a\right)-m^*(R)\right|}{\|a\|} > a_n$.
Since $\iota(0;r,\gamma_0)=m^*(r)$ for any $(r,\gamma_0)$ by~\eqref{eq:h_1_n_bar},
we can take $c^{(n)}\left(r;\gamma_0,\gamma_1^{(n)}\right) = 0$ whenever this condition on $r$ holds.
When such conditions are not satisfied,
we have $|\omega^*(r)| \leq a_n$ and also $\sup_{0<\|a\| \leq b_n} \frac{\left|m^{(n)}\left(R;a\right)-m^*(R)\right|}{\|a\|} \leq a_n$.
Then note that by Assumption~\ref{assump:Var} there exists some open interval $\mathcal{C}$ containing 0 on which 
$\frac{\partial \iota}{\partial c}(c;r,\gamma_0)$ and $\frac{\partial ^2 \iota}{\partial c^2}(c;r,\gamma_0)$ exist.
With $\omega^*(r)$ defined as in~\eqref{eq:omega_star},
we compute
\begin{align*}
\frac{\partial \iota(0;r,\gamma_0)}{\partial c} & = \int y^2 f^*(y \mid r)d\mu(y) - \left(\int yf^*(y \mid r) d\mu(y)\right)^2 = V^*(r) \geq \delta, \quad  \forall r,\gamma_0 \\
\frac{\partial^2 \iota(0;r,\gamma_0)}{\partial c^2} & = \omega^*(r)
\end{align*}
The condition $|\omega^*(r)| \leq a_n$ thus ensures
there exists an interval $(-\tilde{c}_n,\tilde{c}_n) \subseteq \mathcal{C}$ on which $\frac{\partial \iota(c;r,\gamma_0)}{\partial c} \geq \delta/2$ for all $(r,\gamma_0)$.
This implies that for all $(r,\gamma_0)$
and $|\epsilon| < \frac{\tilde{c}_n\delta}{2}$,
there exists $c \in (-\tilde{c}_n,\tilde{c}_n)$ such that $\iota(c;r,\gamma_0) = m^*(r)+\epsilon$.
But the condition $\sup_{0<\|a\| \leq b_n} \frac{\left|m^{(n)}\left(R;a\right)-m^*(R)\right|}{\|a\|} \leq a_n$
ensures that whenever $\|\gamma_1^{(n)}\| \leq \min\left(b_n, \frac{\tilde{c}_n\delta}{2a_n}\right) $
we in fact have $\left| m^{(n)}(r;\gamma_1^{(n)}) - m^*(r) \right|  \leq \frac{\tilde{c}_n\delta}{2}$,
showing the existence of $c^{(n)}\left(r;\gamma_0,\gamma_1^{(n)}\right)$ to satisfy~\eqref{eq:tilted_mean} for all $r$ and $\left(\gamma_0,\gamma_1^{(n)}\right)$ in a neighborhood of 0.}.
Note that $f^{(n)}$ is indeed a valid density due to boundedness of $\bar{h}_1^{(n)}$ ensuring $1+\gamma_0\bar{h}_1^{(n)}(w) > 0$ for all $w$ and all $\gamma_0$ sufficiently close to 0,
while equation~\eqref{eq:tilted_mean} ensures the densities $f^{(n)}$ 
correspond to distributions with mean functions lying in $\cp_{\cm}$.
Now using the quotient rule and noting we can take $c^{(n)}(r;0,0)=0$ for all $r$ by~\eqref{eq:h_1_n_bar},
we compute
\begin{align*}
\frac{\partial f^{(n)}(y \mid r,0,0)}{\partial \gamma_0} & = f^*(y \mid r)\bar{h}_1^{(n)}(w) \\
\frac{\partial f^{(n)}(y \mid r;0,0)}{\partial \gamma_1^{(n)}} & = (y-m^*(r))f^*(y \mid r)\frac{\partial c^{(n)}(r;0,0) }{\partial \gamma_1^{(n)}}
\end{align*}
which are square integrable by Assumption~\ref{assump:MGF}.
So $f^{(n)}(y \mid r; \gamma^{(n)})$ is a parametric submodel for the conditional distribution of $Y$ given $R=r$ in the model $\cp_{\cm_1}$ with score
\begin{align*}
\frac{\partial \log f^{(n)}(w;0)}{\partial \gamma^{(n)}} & = \left(\bar{h}_1^{(n)}(w), (y-m^*(r))\frac{\partial c^{(n)}(r;0,0) }{\partial \left(\gamma_1^{(n)}\right)^{\top}}\right)^{\top}
\end{align*}
for $\gamma^{(n)} = (\gamma_0,\gamma_1^{(n)})$.
By taking the derivative of both sides of~\eqref{eq:tilted_mean} with respect to $\gamma_1^{(n)}$ at $\gamma_0=0$, $\gamma_1^{(n)}=0$,
we see that
\[
\frac{\partial c^{(n)}(r;0;0)}{\partial \gamma_1^{(n)}}V^*(r) = \frac{\partial m^{(n)}(r;0)}{\partial \gamma_1^{(n)}} 
\]
so we can rewrite 
\[
 (y-m^*(r))\frac{\partial c^{(n)}(r;0,0) }{\partial \gamma_1^{(n)}} = 
\frac{(y-m^*(r))\frac{\partial m^{(n)}(r;0)}{\partial \gamma_1^{(n)}}}{V^*(r)}
\]
Hence
\begin{align*}
(1, c_n^{\top})^{\top} \frac{\partial \log f^{(n)}(w;0)}{\partial \gamma^{(n)}} & = \bar{h}_1^{(n)}(w) + (y-m^*(r)) \frac{c_n^{\top}\frac{\partial m^{(n)}(r;0)}{\partial \gamma_1^{(n)}}}{V^*(r)} 
\end{align*}
which we've ensured converge in mean square to
\[
\bar{h}_1(w) + \frac{(y-m^*(r))\tilde{h}_1(r)}{V^*(r)} 
= h_1(w) 
\]
as $n \rightarrow \infty$.
This allows us to conclude $h_1 \in \ct_1$, as desired.

Having shown $\ct_{\cm} = \ct_1 \oplus \ct_2$ 
with $\ct_1$ as in~\eqref{eq:T_1} and $\ct_2=\ch_{R}^0$,
we are now ready to prove~\eqref{eq:T_M_perp_generalized}.
First suppose $g$ is an element of the right-hand side of~\eqref{eq:T_M_perp_generalized},
so that $g(w)=h(r)(y-m^*(r))$ for some $h \in \cs_{\cm}^{\perp}$.
Then for any $g_1 \in \ct_1$ we have
\begin{align*}
\e[g(W)g_1(W)] & = \e[h(R)\e[(Y-m^*(R))g_1(W) \mid R]] \text{ by conditioning on $R$} \\
& = -\e[h(R)m^*(R)\e[g_1(W) \mid R]] \text{ since $\e[Yg_1(W) \mid R] \in \cs_{\cm}$} \\
& = 0 \text{ since $\e[g_1(W) \mid R]=0$}.
\end{align*}
Additionally,
for any $g_2 \in \ct_2$ we have
\[
\e[g(W)g_2(R)] = \e[h(R)g_2(R)\e[Y-m^*(R) \mid R]] = 0
\]
Hence $g \in \ct_1^{\perp} \cap \ct_2^{\perp} = (\ct_1 \oplus \ct_2)^{\perp}$,
completing the inclusion ``$\supseteq$" of~\eqref{eq:T_M_perp_generalized}.

To show the reverse inclusion,
we now fix $g \in (\ct_1 \oplus \ct_2)^{\perp}$.
Since $g \in \ct_2^{\perp}$ but $x \mapsto \e[g(W) \mid R=r] \in \ct_2$,
by conditioning on $R$ we have
\[
0 = \e[g(W)\e[g(W) \mid R]] = \e[(\e[g(W) \mid R])^2]
\]
and so $\e[g(W) \mid R]=0$.
Next, we define
\[
\tilde{g}(w)=h(r)(y-m^*(r)), \quad h(r)=\frac{\int (y'-m^*(r))g(r,y')f^*(y \mid r) d\mu(y)}{V^*(r)}
\]
so that $h(R)V^*(R) = \e[(Y-m^*(R))g(W) \mid R]$.
Notice that $g-\tilde{g} \in \ct_1$,
since 
\[
\e[\tilde{g}(W) \mid R] = h(R)\e[Y-m^*(R) \mid R] = 0 = \e[g(W) \mid R]
\]
while
\begin{align*}
\e[Y(g(W)-\tilde{g}(W)) \mid R] & = \e[Yg(W) \mid R] - \e\left[Yh(R)(Y-m^*(R)) \mid R\right] \\
& = h(R)V^*(R) - \e[(Y-m^*(R))^2h(R) \mid R] \\
& = 0 \in \cs_{\cm}
\end{align*}
where the second equality uses the fact that $\e[h(R)(Y-m^*(R)) \mid R]=\e[h(R)g(W) \mid R]=0$ for all $h \in \ch_R$.
Since $g \in \ct_1^{\perp}$ we must have
\begin{equation}
\label{eq:g_perp}
\e[g(W)(g(W)-\tilde{g}(W))] = 0.
\end{equation}
It follows that
\begin{align*}
\e[(g(W)-\tilde{g}(W))^2] & = \e[\tilde{g}^2(W)]-\e[g(W)\tilde{g}(W)] \\
& = \e\left[\frac{(\e[(Y-m^*(R))g(W) \mid R])^2(Y-m^*(R))^2}{(V^*(R))^2}\right]  \\
& - \e\left[\frac{\e[g(W)(Y-m^*(R)) \mid R]g(W)(Y-m^*(R))}{V^*(R)}\right] \\
& = 0
\end{align*}
Hence $g(w)=\tilde{g}(w)$. 
Finally we note that $h \in \cs_{\cm}^{\perp}$ since for any $h_1 \in \cs_{\cm}$ we have
\begin{align*}
\e[h_1(R)h(R)] & = \e\left[\frac{h_1(R)\e[(Y-m^*(R))g(W) \mid R]}{V^*(R)}\right] \\
& = \e\left[\frac{(Y-m^*(R))g(W)h_1(R)}{V^*(R)}\right] \\
& = 0
\end{align*}
where the last equality follows since $g \in \ct_1^{\perp}$ but
\[
\hat{g}(w) = \frac{(y-m^*(r))h_1(r)}{V^*(r)}
\]
satisfies $\e[\hat{g}(W) \mid R] = 0$ and $\e[Y\hat{g}(W) \mid R] = h_1(R)$
so that $\hat{g} \in \ct_1$ by definition.

\subsection{Proof of Corollary~\ref{cor:known_propensity_score}}
\label{app:known_propensity_score}

We write $\ch_W^0$ as the direct sum of orthogonal subspaces $\ct_1 \oplus \ct_2$ where
\[
\ct_1 = \{f \in \ch_W^0 \mid \e[f(W) \mid R] = 0 \}, \quad \ct_2 = \ch_R^0.
\]
Since the model $\cp_{\cm}$ places no restrictions on the distribution of $R$,
it follows that its semiparametric tangent space $\ct_{\cm} = \bar{\ct}_1 \oplus \ct_2$ for some $\bar{\ct}_1 \subseteq \ct_1$
(this is also easy to see directly from the fact that $\ct_{\cm}^{\perp} \subseteq \ct_1$ by~\eqref{eq:T_M_perp_generalized}).
Conversely, the experimental propensity score being known places no restriction on the conditional distribution of $Y$ given $R$.
Thus,
we can write $\tilde{\ct}_{\cm} = \bar{\ct}_1 \oplus \bar{\ct}_2$
where $\bar{\ct}_2$ is the semiparametric tangent space for the model $\cp_{\cn}$ on the distribution of $R=(Z,S,X)$,
viewing $Z$ as the outcome and $(S,X)$ the covariates
and letting $\cn = \{(s,x) \mapsto se^*(x)+(1-s)h(x) \mid h \in \ch_X\}$ so that the restriction that the mapping $(s,x) \mapsto \e[Z \mid S=s,X=x] \in \cn$
amounts precisely to the experimental propensity score being known and equal to $e^*$.
By~\eqref{eq:T_M_perp_generalized} we have
\begin{equation}
\label{eq:T_2_perp_tilde}
\bar{\ct}_2^{\perp} = \{(z-e^*(x))h(x) \mid h \in \cs_{\cn}^{\perp}\}
\end{equation}
where $\cs_{\cn}^{\perp}$ is the orthogonal complement of $\cs_{\cn}$ in $\ch_{SX}$
and $\bar{\ct}_2^{\perp}$ denotes the orthogonal complement of $\bar{\ct}_2$ in $\ct_2=\ch_R^0$,
(not the orthogonal complement in $\ch_W^0$).
We conclude 
\[
\tilde{\ct}_{\cm}^{\perp} = \ct_{\cm}^{\perp} \oplus \tilde{\ct}_2^{\perp}.
\]
Thus, to prove~\eqref{eq:T_M_perp_tilde} 
it suffices to show that 
\[
\cs_{\cn}^{\perp} = \{(s,x) \mapsto sh(x) \mid h \in \ch_X\}.
\]

To that end, 
we first show that $\cs_{\cn} = \{(s,x) \mapsto (1-s)h(x) \mid h \in \ch_X\}$.
To do so,
note that the RHS is a closed subset of $\ch_{SX}$,
and that an arbitrary outcome mean function parametric submodel of $\cn$ takes the form
\[
m(s,x;\gamma) = se^*(x) + (1-s)h(x;\gamma)
\] 
where $h(x;0) = r^*(x)$.
For any $c$ with the same length as $\gamma$ we have
\[
c^{\top} \frac{\partial m(s,x;0)}{\partial \gamma} = (1-s)c^{\top} \frac{\partial h(x;0)}{\partial \gamma}
\]
which clearly lies in the set $\{(s,x) \mapsto (1-s)h(x) \mid h \in \ch_X\}$.
Hence $\cs_{\cn}$ is contained inside this set.
Conversely, given an arbitrary element $g$ given by $g(s,x) = (1-s)h(x)$ in this set,
we can consider the univariate submodels
\[
m_n(s,x;\gamma) = se^*(x) + (1-s)(q^*(x)+\gamma h_n(x))
\]
for a sequence of bounded functions $h_n \in \ch_X$ converging in mean square to $h$.
Note the validity of each of these submodels for $\gamma$ in a neighborhood of zero uses the fact that $h_n$ is bounded and $\delta < q^*(x) < 1-\delta$ for all $x$ and some $\delta > 0$.
With $\partial m_n(s,x;\gamma)/\partial \gamma\Big|_{\gamma=0}$ converging to $g(s,x)$ in mean square,
we conclude $g \in \cs_{\cn}$,
and indeed $\cs_{\cn}=\{(s,x) \mapsto (1-s)h(x) \mid h \in \ch_X\}$.
Given this,
it is then straightforward to verify that
\[
\cs_{\cn}^{\perp} = \{(s,x) \mapsto sh(x) \mid h \in \ch_X\}
\]
upon recalling that $s(1-s)=0$.

\subsection{Proof of Corollary~\ref{cor:homoskedasticity}}
\label{app:homoskedasticity}

It suffices to show that $\varphi_{0,\rct}$ is the EIF for $\tau_{\rct}$ in the model $\tilde{\cp}_{\cm}$,
since $\tilde{\cp}_{\cm} \subseteq \cp_{\cm} \subseteq \cp_{\ch_R}$ and hence the semiparametric bound for $\cp_{\cm}$ must be sandwiched between the bounds for $\tilde{\cp}_{\cm}$ and the nonparametric model $\cp_{\ch_R}$.
To do so,
by Theorem~\ref{thm:influence_functions} it suffices to show that $\varphi_{0,\rct} \in (\tilde{\ct}_{\cm}^{\perp})^{\perp} = \tilde{\ct}_{\cm}$.

Fix an arbitrary element $g \in \tilde{\ct}_{\cm}^{\perp}$.
By Theorem~\ref{thm:semiparametric_tangent_space},
we can write $g(w) = h(r)(y-m^*(r)) + s\tilde{h}(x)(z-e^*(x))$ for some
$h \in \cs_{\cm}^{\perp}$ and $\tilde{h} \in \ch_X$.
For each $f \in \ch_R$, 
with $\Delta(z,x,y;\eta)$ as in Proposition~\ref{prop:efficiency_bounds_0},
we define
\begin{align*}
\phi(r;f) & = h(r)\e[((\rho^*)^{-1}S\Delta(Z,X,Y;\eta^*)+f(R))(Y-m^*(R)) \mid R=r] \\
& = h(r) \left(\frac{sz}{\rho^*e^*}\e[(Y-m_{11}^*(X))^2 \mid R] - \frac{s(1-z)}{\rho^*(1-e^*)}\e[(Y-m_{10}^*(X))^2 \mid R]\right) \\
& = h(r)(\rho^*)^{-1}(c_1sz-c_0s(1-z))
\end{align*}
where $c_1=\sigma_1^2/e^*$ and $c_0=\sigma_0^2/(1-e^*)$,
and the second equality uses the fact that $f(r)\e[Y-m^*(R) \mid R]=0$.
Note $\phi(r;f)$ is independent of $f \in \ch_R$
and by the assumption of the theorem, $r \mapsto (\rho^*)^{-1}(c_1sz-c_0s(1-z)) \in \cs_{\cm}$.
With $h \in \cs_{\cm}^{\perp}$ we conclude that 
\begin{equation}
\label{eq:phi}
\e[\phi(R;f)]=0, \quad \forall f \in \ch_R.
\end{equation}

Next, we note that for any $\bar{g} \in \ch_X$ we have
\begin{align}
\e\left[S\tilde{h}(X)(Z-e^*)(\Delta(Z,X,Y;\eta^*)+\bar{g}(X))\right] & = \e\left[S\tilde{h}(X)(Z-e^*)\e(\Delta(Z,X,Y;\eta^*) \mid S=1,Z,X)\right] \nonumber \\
& + \e[S\tilde{h}(X)\bar{g}(X)\e(Z-e^* \mid S=1,X)] \nonumber \\
& = 0 \label{eq:h_tilde_iota}
\end{align}

Finally, we put together~\eqref{eq:phi} and~\eqref{eq:h_tilde_iota} to show that $\varphi_{0,\rct}$ must be orthogonal to our arbitrary chosen function $g \in \tilde{\ct}_{\cm}^{\perp}$.
By iterated expectations,
defining $f(r)=s(m_{11}^*(x)-m_{10}^*(x)-\tau_{\rct}^*)/\rho^*$ we compute
\begin{align*}
\e[\varphi_{0,rct}(W;\tau_{rct}^*,\eta^*)g(W)] & = \e[\varphi_{0,\rct}(W;\tau_{\rct}^*;\eta^*)(h(R)(Y-m^*(R))+S\tilde{h}(X)(Z-e^*))] \\
& = \e[h(R)\e(\varphi_{0,\rct}(W;\tau_{\rct}^*;\eta^*)(Y-m^*(R)) \mid R] \\
& \quad + (\rho^*)^{-1} \e[S\tilde{h}(X)(Z-e^*)(\Delta(Z,X,Y;\eta^*)+m_{11}^*(X)-m_{10}^*(X)-\tau_{\rct}^*)]  \nonumber \\
& = \e[\phi(S,Z,X;f)] + 0 \text{ by~\eqref{eq:h_tilde_iota} with $\bar{g}(X)=m_{11}^*(X)-m_{10}^*(X)-\tau^*_{\rct}$} \\
& = 0 \text{ by~\eqref{eq:phi}}.
\end{align*}


\subsection{Proof of Lemma~\ref{lemma:dml_master}}
\label{app:proof_dml_master}

As indicated in the main text,
here we prove a stronger version of Lemma~\ref{lemma:dml_master} that is applicable to models of the form $\tilde{\cp}_{\cm}$ where the experimental propensity score is known,
as considered in Corollary~\ref{cor:known_propensity_score}.

\begin{lemma*}
Suppose the function $g=g(\cdot;\tau^*,\eta^*)$ is any element of the space $\tilde{\ct}_{\cm}^{\perp}$ for
some collection of outcome mean functions $\cm \subseteq \ch_R$ satisfying the conditions of Corollary~\ref{cor:known_propensity_score},
so that we can write
\[
g(w;\tau^*,\eta^*) = (y-m^*(r))h_1(r;\tau^*,\eta^*) + s(z-e^*(x))h_2(x;\tau^*,\eta^*)
\]
where $\eta^*$ includes the mean function $m^*$,
$h_1^*(\cdot) =  h_1(\cdot;\tau^*,\eta^*) \in \cs_{\cm}^{\perp}$,
and $h_2^*(\cdot) = h_2(\cdot;\tau^*,\eta^*) \in \ch_X$.
Then equation~\eqref{eq:g_hat_minus_g_crossfit} holds
under the following conditions:
\begin{itemize}
\item (Boundedness) There exists $C < \infty$ for which $V^*(r) \leq C$ for $P^*$-almost every $r$ and both $\|h_i^*\|_{\infty,P^*} \leq C$ and $\left\|\hat{h}_i^{(-k)}\right\|_{\infty,P^*} \leq C$ for $i=1,2$
\item (Approximate tangency) With probability tending to 1,
for each $k=1,\ldots,K$ there exists $\hat{R}^{(-k)}=\hat{R}^{(-k)}(\cdot) \in \ch_R$ depending only on the observations outside fold $k$ such that the function $r \mapsto \hat{m}^{(-k)}(r) - m^*(r)-\hat{R}^{(-k)}(r)$ lies in the outcome mean function tangent space $\cs_{\cm}$,
$\|\hat{R}^{(-k)}\|_{2,P^*} = o_p(1)$,
and the following rate holds:
\begin{equation}
\label{eq:R_hat_generalized}
\frac{1}{|\ci_k|} \sum_{i \in \ci_k} \left|\hat{R}^{(-k)}(R_i)\right| = o_p(N^{-1/2})
\end{equation}
\item (Rate conditions) The following rate conditions hold for all $k=1,\ldots,K$:
\begin{itemize}
    \item (First order consistency) $\|\hat{m}^{(-k)}-m^*\|_{2,P^*} + \|\hat{h}_1^{(-k)}-h_1^*\|_{2,P^*} + \|\hat{h}_2^{(-k)}-h_2^*\|_{2,P^*} = o_p(1)$ 
    \item (Product rate condition) $\|\hat{m}^{(-k)}-m^*\|_{2,P^*} \times \|\hat{h}_1^{(-k)}-h_1^*\|_{2,P^*} = o_p(N^{-1/2})$.
\end{itemize}
\end{itemize}
\end{lemma*}

Lemma~\ref{lemma:dml_master} follows immediately from the modified Lemma above since $\tilde{\ct}_{\cm}^{\perp} \supseteq \ct_{\cm}^{\perp}$,
and for any $g \in \ct_{\cm}^{\perp}$ we can apply the modified Lemma with $h_2=h_2^{(-k)}=0$.

Writing
\begin{align*}
\frac{1}{N} \sum_{k=1}^K \sum_{i \in \ci_k} g\left(W_i;\hat{\tau}^{(-k)},\hat{\eta}^{(-k)}\right)-g(W_i;\tau^*,\eta^*) = \sum_{k=1}^K \frac{|\ci_k|}{N} \frac{1}{|\ci_k|} \sum_{i \in \ci_k} g\left(W_i;\hat{\tau}^{(-k)},\hat{\eta}^{(-k)}\right)-g(W_i;\tau^*,\eta^*)
\end{align*}
we see it suffices to fix $k \in \{1,\ldots,K\}$ and show
\[
 \frac{1}{|\ci_k|} \sum_{i \in \ci_k} g\left(W_i;\hat{\tau}^{(-k)},\hat{\eta}^{(-k)}\right)-g(W_i;\tau^*,\eta^*) = o_p(N^{-1/2}).
\]
We have
\[
 \frac{1}{|\ci_k|} \sum_{i \in \ci_k} g\left(W_i;\hat{\tau}^{(-k)},\hat{\eta}^{(-k)}\right)-g(W_i;\tau^*,\eta^*) =  a_k+b_k
\]
where
\begin{align*}
a_k & = \frac{1}{|\ci_k|} \sum_{i \in \ci_k} \left(Y_i-\hat{m}^{(-k)}(R_i)\right)\hat{h}_1^{(-k)}(R_i) - (Y_i-m^*(R_i))h_1^*(R_i) \\
b_k & = \frac{1}{|\ci_k|} \sum_{i \in \ci_k} S_i(Z_i-e^*(X_i))\left(\hat{h}_2^{(-k)}(X_i)-h_2^*(X_i)\right).
\end{align*}
We further expand
\[
a_k = a_{1k} + a_{2k} + a_{3k}
\]
where
\begin{align*}
a_{1k} & = \frac{1}{|\ci_k|} \sum_{i \in \ci_k} \left(m^*(R_i)-\hat{m}^{(-k)}(R_i)\right)h_1^*(R_i) \\
a_{2k} & = \frac{1}{|\ci_k|} \sum_{i \in \ci_k} (Y_i-m^*(R_i))\left(\hat{h}_1^{(-k)}(R_i)-h_1^*(R_i)\right) \\
a_{3k} & = \frac{1}{|\ci_k|} \sum_{i \in \ci_k} \left(m^*(R_i)-\hat{m}^{(-k)}(R_i)\right)\left(\hat{h}_1^{(-k)}(R_i)-h_1^*(R_i)\right)
\end{align*}

First, we show $a_{1k} = o_p(N^{-1/2})$.
Note that
\[
\Bigg| \frac{1}{|\ci_k|} \sum_{i \in \ci_k} \hat{R}^{(-k)}(R_i)h_1^*(R_i)\Bigg|  \leq C \frac{1}{|\ci_k|} \sum_{i \in \ci_k} \left|\hat{R}^{(-k)}(R_i) \right| = o_p(N^{-1/2})
\]
by~\eqref{eq:R_hat_generalized} and the boundedness condition $\|h_1^*\|_{\infty,P^*} \leq C$.
Hence
$a_{1k} = \tilde{a}_{1k} + o_p(N^{-1/2})$ where
\[
\tilde{a}_{1k} = \frac{1}{|\ci_k|} \sum_{i \in \ci_k} \left(\hat{R}^{(-k)}(R_i) + m^*(R_i)-\hat{m}^{(-k)}(R_i)\right)h_1^*(R_i).
\]
Furthermore, whenever $\hat{R}^{(-k)}$ exists,
we have that
\[
\e\left[\left(\hat{R}^{(-k)}(R_i) + m^*(R_i)-\hat{m}^{(-k)}(R_i)\right)h_1^*(R_i) \mid \cf_N^{(-k)}\right] = 0
\]
for each $i \in \ci_k$;
here $\cf_N^{(-k)}$ is the $\sigma$-algebra generated by the observations outside fold $k$,
and the preceding display holds since $\hat{R}^{(-k)}+m^*-m^{(-k)} \in \cs_{\cm}$
(recall $\cs_{\cm}$ is assumed to be a linear space, hence it is closed under multiplication by negative 1)
but $h_1^* \in \cs_{\cm}^{\perp}$.
With the summands in $\tilde{a}_{1k}$ i.i.d. across $i \in \ci_k$ conditional on $\cf_N^{(-k)}$,
we compute
\begin{align*}
\e\left[\tilde{a}_{1k}^2 \mid \cf_N^{(-k)}\right] & = \frac{1}{|\ci_k|} \int \left(\hat{R}^{(-k)}(r) + m^*(r)-\hat{m}^{(-k)}(r)\right)^2(h_1^*(r))^2 dP^*(w) \\
& \leq \frac{C^2}{|\ci_k|} \left\|\hat{R}^{(-k)}+m^*-\hat{m}^{(-k)}\right\|_{2,P^*}^2 \text{ by boundedness of $h_1^*$} \\
& = o_p(N^{-1})
\end{align*}
since $|\ci_k|^{-1}=O(N^{-1})$ while
\[
\|\hat{R}^{(-k)}+m^*-\hat{m}^{(-k)}\|_{2,P^*} \leq \|\hat{R}^{(-k)}\|_{2,P^*} + \|\hat{m}^{(-k)}-m^*\|_{2,P^*} = o_p(1).
\]
by Minkowski's inequality and the rate conditions $\|\hat{R}^{(-k)}\|_{2,P^*} = o_p(1)$ and $\|\hat{m}^{(-k)}-m^*\|_{2,P^*} = o_p(1)$.
We conclude by Lemma~\ref{lemma:conditional_oh_pee} that $\tilde{a}_{1k} = o_p(N^{-1/2})$ and hence $a_{1k} = o_p(N^{-1/2})$ as well.

Next, we show $a_{2k} = o_p(N^{-1/2})$.
Note that $\e\left[Y_i-m^*(R_i) \mid \cf_N^{(-k)},\cf_R^{(k)}\right] = 0$ for all $i \in \ci_k$ where $\cf_R^{(k)}$ is the $\sigma$-algebra generated by $(R_i)_{i \in \ci_k}$.
Additionally we have that the collection $\{Y_i-m^*(R_i)\}_{i \in \ci_k}$ is conditionally independent given the $\sigma$-algebra generated by $\cf_N^{(-k)}$ and $\cf_R^{(k)}$,
and that $\hat{h}_1^{(-k)}(R_i)-h_1^*(R_i)$ is measurable with respect to this $\sigma$-algebra.
Finally we have $\e\left[(Y_i-m^*(R_i))^2 \mid \cf_N^{(-k)},\cf_R^{(k)}\right] = V^*(R_i) \leq C$ by the boundedness condition.
Thus with 
\[
\frac{1}{|\ci_k|} \sum_{i \in \ci_k} \left(\hat{h}_1^{(-k)}(R_i)-h_1^*(R_i)\right)^2 = o_p(1)
\]
by the rate conditions for the first component and Lemma~\ref{lemma:conditional_oh_pee_fn_f},
we conclude by Lemma~\ref{lemma:mean_0_times_small} that $a_{2k} = o_p(N^{-1/2})$.

To show that $a_{3k} = o_p(N^{-1/2})$,
simply note that
\begin{align*}
|a_{3k}| & \leq \left(\frac{1}{|\ci_k|} \sum_{i \in \ci_k} \left(m^*(R_i)-\hat{m}^{(-k)}(R_i)\right)^2\right)^{1/2} \times \left(\frac{1}{|\ci_k|} \sum_{i \in \ci_k}\left(\hat{h}_1^{(-k)}(R_i)-h_1^*(R_i)\right)^2\right)^{1/2} \\
& = o_p(N^{-1/2})
\end{align*}
by Cauchy-Schwarz and the product rate condition $\|\hat{m}^{(-k)}-m^*\|_{2,P^*}\|\hat{h}_1^{(-k)}-h_1^*\|_{2,P^*} = o_p(N^{-1/2})$,
in view of Lemma~\ref{lemma:conditional_oh_pee_fn_f}.

Finally, let $\cf_X^{(k)}$ be the $\sigma$-algebra generated by the observations $\{X_i\}_{i \in \ci_k}$,
so that
\[
\e\left[S_i(Z_i-e^*(X_i)) \mid \cf_N^{(-k)},\cf_X^{(k)}\right] = 0, \quad \forall i \in \ci_k.
\]
The collection $\{S_i(Z_i-e^*(X_i))\}_{i \in \ci_k}$ is conditionally independent given the $\sigma$-algebra generated by $\cf_N^{(-k)}$ and $\cf_X^{(k)}$,
and $\hat{h}_2^{(-k)}(X_i)-h_2^*(X_i)$ is measurable with respect to this $\sigma$-algebra.
Since
\[
\e\left[\left(S_i(Z_i-e^*(X_i))\right)^2 \mid \cf_N^{(-k)}, \cf_X^{(k)} \right] \leq 1, \quad \frac{1}{|\ci_k|} \sum_{i \in \ci_k} \left(\hat{h}_2^{(-k)}(X_i)-h_2^*(X_i)\right)^2 = o_p(1)
\]
by the rate conditions for the second component and Lemma~\ref{lemma:conditional_oh_pee_fn_f}, 
we conclude by Lemma~\ref{lemma:mean_0_times_small} that $b_k = o_p(N^{-1/2})$.

\section{Additional derivations of efficient influence functions}
\label{app:eif_derivations}

Here we work through the three-step outline presented following the statement of Theorem~\ref{thm:semiparametric_tangent_space} in the main text to derive EIF's in the models $\cp_{\cm_1}$, $\cp_{\cm_2}$, $\cp_{\cm_3}$, and $\cp_{\cm_4}$ corresponding to examples~\ref{ex:restricted_moment_model},~\ref{ex:mean_exchangeable_controls},~\ref{ex:parametric_confounding_cate}, and~\ref{ex:parametric_confounding_bias}.

\subsection{Restricted moment model}
\label{app:restricted_moment_model}
In the restricted moment model $\cp_{\cm_1}$ of Example~\ref{ex:restricted_moment_model},
under suitable regularity conditions
we have by standard $Z$-estimation theory that the solution to
\[
\frac{1}{N} \sum_{i=1}^N a(R_i)(Y_i-\mu(R_i;\beta)) = 0
\]
is a RAL estimator for $\beta \in \real^q$ for any $a \in \ch_R^q$.
The corresponding influence function is
\[
\varphi_0(w;\beta^*,\eta^*) = \e\left[a(r)D(r;\beta^*)^{\top}\right]^{-1} a(r)(y-\mu(r;\beta))
\]
where $D(r;b) =  \mu_{\beta}(r;b) \equiv \partial \mu(r;\beta) / \partial \beta\Big|_{\beta=b} \in \real^q$
(cf. Chapter 5 of~\citet{vandervaart2000asymptotic}).
In the case that $\mu(r;\beta)=\psi(r)^{\top}\beta$ is a linear model,
taking $a(r) = \psi(r)$ gives the ordinary least squares estimator.
In any case, we can use $\varphi_0$ as above for our initial influence function of an RAL estimator of $\beta$,
which we aim to project via~\eqref{eq:eif_raw} to derive the EIF for $\beta$ in the model $\cp_{\cm_1}$.

To do so, we first compute the semiparametric tangent space $\cs_{\cm_1}$, similar to the example in Section~\ref{sec:S_M}.
An arbitrary outcome mean function parametric submodel for $\cm_1$ evidently takes the form
\[
m(r;\gamma) = \mu(r;\alpha(\gamma))
\]
for some smooth $\alpha:\real^s \rightarrow \real^q$ defined on a neighborhood of $0 \in \real^s$ with $\alpha(0)=\beta^*$.
Then by the chain rule
\[
m_{\gamma}(r;0) = J_{\alpha}(0)^{\top} D(r;\beta^*)
\]
for $J_{\alpha}(0) \in \real^{q \times s}$ a Jacobian matrix.
For any $c_0 \in \real^s$,
$c_0^{\top}m_{\gamma}(r;0)$ is of the form $r \mapsto c^{\top}D(r;\beta^*)$ for some $c \in \real^q$,
which is evidently closed.
Conversely,
any function of this form clearly lies in the span of the derivative of the simple $q$-dimensional submodel $m(r;\gamma)= \mu(r;\beta^* + \gamma)$.
It follows that
\[
\cs_{\cm_1} = \{r \mapsto c^{\top}D(r;\beta^*) \mid c \in \real^q\}.
\]
We note in passing that for the linear model $\mu(r;\beta) = \psi(r)^{\top}\beta$ we can verify that indeed $\cs_{\cm_1}=\cm_1$,
as shown by Proposition~\ref{prop:S_M_equals_M}.

To derive the orthogonal complement $\cs_{\cm_1}^{\perp}$,
as suggested in Section~\ref{sec:projections} we solve for the set of functions $h \in \ch_R$ for which the projection $g=\Pi(h;\cs_{\cm_1})$ is equal to 0.
We have $g(r)=c_h^{\top}D(r;\beta^*)$ where
\[
c_h = \argmin_{c \in \real^q} \e\left[\left(h(R)-c^{\top}D(R;\beta^*)\right)^2\right]
\]
Taking a derivative gives
\[
c_h = \left(\e[D(R;\beta^*)D(R;\beta^*)^{\top}]\right)^{-1}\e[h(R)D(R;\beta^*)]
\]
so $g=0 \iff c_h=0 \iff \e[h(R)D(R;\beta^*)=0]$.
It follows that
\[
\cs_{\cm_1}^{\perp} = \{h \in \ch_R \mid \e[h(R)D(R;\beta^*)]=0\}.
\]
and (by Theorem~\ref{thm:semiparametric_tangent_space}) that
\[
\ct_{\cm_1}^{\perp} = \{w \mapsto h(r)(y-\mu(r;\beta^*)) \mid \e[h(R)D(R;\beta^*)]=0\}.
\]

Finally, we compute the projection $g=\Pi\left(\varphi_0;\left(\ct_{\cm_1}^{\perp}\right)^q\right)$ to get the EIF via~\eqref{eq:eif_raw}.
By definition,
$g(w)=(y-\mu(r;\beta^*))h_0(r)$ where
\begin{align*}
h_0 & = \argmin_{h \in \ch_R^q: \e[h(R)D(R;\beta^*)^{\top}]=0} \e\left[(Y-\mu(R;\beta^*))^2\|A^{-1}a(R)-h(X)\|^2\right] \\
A & = \e[a(R)D(R;\beta^*)^{\top}].
\end{align*}
To compute this projection,
we handle the constraint $\e[h(R)D(R;\beta^*)^{\top}]=0$ by introducing a Lagrange multiplier $\Lambda \in \real^{q \times q}$.
Then we minimize the Lagrangian by conditioning on $R$ and minimizing pointwise:
\begin{align*}
& \argmin_{h \in \ch_R^q} \e\left[(Y-\mu(R;\beta^*))^2\|A^{-1}a(R)-h(R)\|^2\right] - 2\tr\left(\Lambda^{\top}\e[h(R)D(R;\beta^*)^{\top}]\right) \\
& \quad = \argmin_{h \in \ch_R^q} \e\left[V^*(R)\|h(R)\|^2 - 2\left(V^*(R)a(R)^{\top}A^{-1} + D(R;\beta^*)^{\top}\Lambda^{\top}\right) h(R)\right]
\end{align*}
Minimizing the argument to the expectation pointwise gives
\[
h_0(r) = A^{-1}a(r) + \frac{1}{V^*(r)}\Lambda D(r;\beta^*)
\]
We solve for the Lagrange multiplier $\Lambda$ by plugging the preceding display into the matrix equations $\e[h_0(R)D(R;\beta^*)^{\top}]=0$,
giving
\[
\Lambda = -\left(\e\left[\frac{1}{V^*(R)}D(R;\beta^*)D(R;\beta^*)^{\top}\right]\right)^{-1}.
\]
Then the EIF is
\begin{align*}
\varphi_{\eff}(w;\beta^*,\eta^*) & = \varphi_0(w;\beta^*,\eta^*) - h_0(r)(y-\mu(r;\beta^*)) \\
& = \frac{y-\mu(r;\beta^*)}{V^*(r)} \left(\e\left[\frac{1}{V^*(R)}D(R;\beta^*)D(R;\beta^*)^{\top}\right]\right)^{-1}D(r;\beta^*)
\end{align*}
which matches equation (4.55) of~\citet{tsiatis2006semiparametric}.

\subsection{Mean-exchangeable controls}

We derive EIF's in the model $\cp_{\cm_2}$.
We can use $\varphi_{0,\obs}$ and $\varphi_{0,\tgt}$ as the initial influence functions $\varphi_0$.
By Proposition~\ref{prop:S_M_equals_M},
we know $\cs_{\cm_2}=\cm_2$.
It only remains to compute $\cs_{\cm_2}^{\perp}$ and the projection~\eqref{eq:eif_raw}.

As in the derivation of $\cs_{\cm_5}^{\perp}$ in the main text,
we characterize $\cs_{\cm_2}^{\perp}$ as the set of all functions $h \in \ch_R$ for which $g=\Pi(h;\cs_{\cm_2})=0$.
To that end we write $h(r)=zh_{11}(x)+(1-s)h_{01}(x)+s(1-z)h_{10}(x)$ and
\[
g = zg_1(x)+(1-z)g_0(x)
\]
where
\begin{align*}
(g_0,g_1) & = \argmin_{(f_0,f_1) \in \ch_X^2} \e[(h(R)-Zf_1(X)-(1-Z)f_0(X))^2] \\
& = \argmin_{(f_0,f_1) \in \ch_X^2} \e[Z(h_{11}(X)-f_1(X))^2] \\
& \quad + \e[(1-S)(h_{01}(X)-f_0(X))^2+S(1-Z)(h_{10}(X)-f_0(X))^2] \\
& = \argmin_{(f_0,f_1) \in \ch_X^2} \e[p^*(X)e^*(X)(h_{11}(X)-f_1(X))^2 + (1-p^*(X))(h_{01}(X)-f_0(X))^2] \\
& \quad +\e[p^*(X)(1-e^*(X))(h_{10}(X)-f_0(X))^2].
\end{align*}
The final equality follows by conditioning on $X$,
and expresses the objective as the expectation of a function of $X$.
Without any restrictions on $f_0$ and $f_1$ we can perform the minimization by minimizing the argument of the expectation pointwise in $x$.
Writing the first order conditions,
we find they are satisfied by $f_0=f_1=0$ if and only if
\begin{align*}
p^*(x)e^*(x)h_{11}(x) & = 0 \\
p^*(x)(1-e^*(x))h_{10}(x) + (1-p^*(x))h_{00}(x) & = 0
\end{align*}
for all $x \in \cx$. 
Hence $h \in \cs_{\cm_2}^{\perp}$ if and only if $h_{11}(x)=0$ and $h_{10}(x)=-\frac{1-p^*(x)}{p^*(x)(1-e^*(x))} h_{00}(x)$,
with $h_{00} \in \ch_X$ free to vary.
In other words
\[
\cs_{\cm_2}^{\perp} = \left\{r \mapsto \left[(1-s)-s(1-z)\frac{1-p^*(x)}{p^*(x)(1-e^*(x))}\right]\zeta(x) \Big| \zeta \in \ch_X\right\}
\]
and so by Theorem~\ref{thm:semiparametric_tangent_space}
\[
\ct_{\cm_2}^{\perp} = \left\{w \mapsto \left[(1-s)-s(1-z)\frac{1-p^*(x)}{p^*(x)(1-e^*(x))}\right]\zeta(x)(y-m^*(r)) \Big| \zeta \in \ch_X\right\}.
\]

Now we compute the projection $g=\Pi(\varphi_0;\ct_{\cm_2}^{\perp})$ from~\eqref{eq:eif_raw} in terms of $\varphi_0$.
We write
\[
g(w)=\left[(1-s)-s(1-z)\frac{1-p^*(x)}{p^*(x)(1-e^*(x))}\right]\zeta(x)(y-m^*(r))
\]
where
\[
\zeta = \argmin_{h \in \ch_X} \e\left[\left(\varphi_{0,\obs}(W)-\left[(1-S)-S(1-Z)\frac{1-p^*(X)}{p^*(X)(1-e^*(X))}\right]h(X)(Y-m^*(R))\right)^2\right].
\]
By our usual approach of rewriting the objective as the expectation of a function of $X$ by iterated expectations and then minimizing the argument of this expectation pointwise,
we get
\[
\zeta(x) = \frac{\e\left[\varphi_0(W)\left((1-S)-S(1-Z)\frac{1-p^*(X)}{p^*(X)(1-e^*(X))}\right)(Y-m^*(R)) \mid X=x\right]}{(1-p^*(x))V_{00}^*(x)+\frac{(1-p^*(x))^2}{p^*(x)(1-e^*(x))}V_{10}^*(x)}
\]

For $\varphi_0=\varphi_{0,\obs}$ we observe
\[
\e[\varphi_0(W;\tau_{\obs}^*,\eta^*)(1-S)(Y-m^*(R)) \mid X] = 0
\]
while
\[
\e\left[-\varphi_0(W;\tau_{\obs}^*,\eta^*)S(1-Z)\frac{1-p^*(X)}{p^*(X)(1-e^*(X))}(Y-m^*(R)) \mid X\right] = \frac{(1-p^*(X))^2V_{10}^*(X)}{p^*(X)(1-e^*(X))(1-\rho^*)}.
\]
Hence
\[
\zeta(x) = \frac{(1-p^*(x))V_{10}^*(x)}{(1-\rho^*)[V_{00}^*(x)p^*(x)(1-e^*(x))+V_{10}^*(x)(1-p^*(x))]}
\]
and then the EIF for $\tau_{\obs}$ is
\[
\varphi_{\eff}(w;\tau_{\obs}^*,\eta^*) = \varphi_{0,\obs}(w;\tau_{\obs}^*,\eta^*) - \left[(1-s)-s(1-z)\frac{1-p^*(x)}{p^*(x)(1-e^*(x))}\right]\zeta(x)(y-m^*(z,x))
\]
which after algebraic simplification matches the expression given in Proposition 3 of~\citet{li_improving_2023}.

Repeating the calculation in the preceding paragraph with $\varphi_0=\varphi_{0,\tgt}$ gives the efficiency bound for $\tau_{\tgt}$,
also given in Proposition 3 of~\citet{li_improving_2023}.

\subsection{Parametric confounding bias and CATE}
Following our three-step outline,
we derive EIF's in the model $\cp_{\cm_3}$ in terms of an arbitrary influence function $\varphi_0=\varphi_0(\cdot,\tau,\eta)$ for an RAL estimator of some generic pathwise-differentiable estimand $\tau$.
We later specialize to the estimand $\tau_{\obs}$ using $\varphi_0=\varphi_{0,\obs}$,
which verifies Theorem 4 of~\citet{yang2024datafusion}.

We begin by deriving the outcome mean function tangent space $\cs_{\cm_3}$.
The method to do so is similar to that of Appendix~\ref{app:restricted_moment_model}.
In words,
the model $\cp_{\cm_3}$ allows the control level ($z=0$) of the outcome mean function to be an arbitrary function of $(s,x)$.
Then given the control level, the treatment level differs by $\mu(x;\beta)$ from the control level when $s=1$ ,
and by $\mu(x;\beta) + \phi(x;\theta)$ from the control level when $s=0$.
Formalizing this mathematically,
this means that an arbitrary outcome mean function parametric submodel of $\cm_3$ must take the form
\begin{equation}
\label{eq:M_3_submodel}
m(r;\gamma)=f_0(s,x;\gamma_0) + z[\mu(x;\beta(\gamma_1)) + (1-s)\phi(x;\theta(\gamma_2))]
\end{equation}
for some smooth mappings $\beta:\real^{s_1} \rightarrow \real^p$ and $\theta:\real^{s_2} \rightarrow \real^q$,
where $\gamma=(\gamma_0^{\top},\gamma_1^{\top},\gamma_2^{\top})^{\top}$ and $\gamma_i \in \real^{s_i}$, $i=0,1,2$.
Then using the subscript notation for (partial) derivatives from Example~\ref{ex:outcome_selection_bias} in Section~\ref{sec:semiparametric_efficiency} of the main text,
for arbitrary $c=(c_0^{\top},c_1^{\top},c_2^{\top})^{\top}$ partitioned in the same way as $\gamma$, we have
\[
c^{\top}m_{\gamma}(r;0) = c_0^{\top}f_0'(s,x;0) + zc_1^{\top}J_{\beta}(0)^{\top}D_{\mu}(x;\beta^*) + z(1-s)c_2^{\top}J_{\theta}(0)^{\top} D_{\phi}(x;\theta^*)
\]
where $J_{\beta}(0) \in \real^{p \times s_1}$ and $J_{\theta}(0) \in \real^{q \times s_2}$ are the Jacobian matrices for the mappings $\beta$ and $\theta$ at 0,
and
\[
D_{\mu}(x;\beta^*) = \frac{\partial \mu(x;\beta)}{\partial \beta}\Big|_{\beta=\beta^*}, \quad D_{\phi}(x;\theta^*) = \frac{\partial \phi(x;\theta)}{\partial \theta} \Big|_{\theta=\theta^*}.
\]
This is of the form
\begin{equation}
\label{eq:S_M_3_form}
f(s,x) + z[c_1^{\top}D_{\mu}(x;\beta^*) + (1-s)c_2^{\top}D_{\phi}(x;\theta^*)]
\end{equation}
for some $f \in \ch_{SX}$, $c_1 \in \real^p$, $c_2 \in \real^q$ (note we have overloaded the notation for $c_1$ and $c_2$).
Conversely,
we now show that any function of the form~\eqref{eq:S_M_3_form},
for some $f \in \ch_{SX}$, $c_1 \in \real^p$, $c_2 \in \real^q$,
lies in $\cs_{\cm_3}$.
To that end, consider the submodel of the form~\eqref{eq:M_3_submodel} with $s_0=1$, $s_1=p$, $s_2=q$ and $f(s,x)=m^*(s,0,x) + \gamma_0f(s,x)$, $\beta(\gamma_1)=\beta^*+\gamma_1$, and $\theta(\gamma_2)=\theta^*+\gamma_2$.
We observe that our function given by~\eqref{eq:S_M_3_form} is equivalent to $c^{\top}m_{\gamma}(r;0)$ for $c=(1,c_1^{\top},c_2^{\top})$,
and thus lies in $\cs_{\cm_3}$ by definition.
We conclude that indeed,
\begin{equation}
\label{eq:S_M_3}
\cs_{\cm_3} = \{r \mapsto f(s,x) + z[c_1^{\top}D_{\mu}(x;\beta^*) + (1-s)c_2^{\top}D_{\phi}(x;\theta^*)] \mid f \in \ch_{SX}, c_1 \in \real^p, c_2 \in \real^q\}.
\end{equation}

It remains to compute the orthogonal complement $\cs_{\cm_3}^{\perp}$ and then the projection~\eqref{eq:eif_raw}.
As in the other examples,
we compute $\cs_{\cm_3}^{\perp}$ as the set of all functions $h \in \ch_R$ for which $g=\Pi(h;\cs_{\cm_3})=0$.
Given $h$ we write 
\[
g(r) = f_0(s,x) + z[d_1^{\top}D_{\mu}(x;\beta^*) + (1-s)d_2^{\top}D_{\phi}(x;\theta^*)]
\]
where
\[
(f_0,d_1,d_2) = \argmin_{f \in \ch_{SX}, c_1 \in \real^p, c_2 \in \real^q} \e\left[\left(h(R)-f(S,X)-Z[c_1^{\top}D_{\mu}(X;\beta^*) + (1-S)c_2^{\top}D_{\phi}(X;\theta^*)]\right)^2\right].
\]
As we have seen,
we can solve this minimization by rewriting the objective as the expectation of a function of solely $X$ via iterated expectations:
\begin{align*}
& \e\left[\left(h(R)-f(S,X)-Z[c_1^{\top}D_{\mu}(X;\beta^*) + (1-S)c_2^{\top}D_{\phi}(X;\theta^*)]\right)^2\right] \\
& \quad = \e[p^*(X)e^*(X)(h_{11}(X)-f(1,X)-c_1^{\top}D_{\mu}(X;\beta^*))^2] + \e[p^*(X)(1-e^*(X))(h_{10}(X)-f(1,X))^2] \\
& \quad + \e[(1-p^*(X))q^*(X)(h_{01}(X)-f(0,X)-c_1^{\top}D_{\mu}(X;\beta^*)-c_2^{\top}D_{\phi}(X;\theta^*))^2]  \\
& \quad + \e[(1-p^*(X))(1-q^*(X))(h_{00}(X)-f(0,X))^2]
\end{align*}
For fixed $c_1$ and $c_2$,
the minimizer of this objective in $f \in \ch_{SX}$ takes the form $f(s,x)=sf(1,x)+(1-s)f(0,x)$ where $f(1,x)$ and $f(0,x)$ globally minimize the argument to the expectation, i.e. they satisfy the first-order conditions
\begin{align*}
0 & = p^*(x)e^*(x)(h_{11}(x)-f(1,x)-c_1^{\top}D_{\mu}(x;\beta^*)) + p^*(x)(1-e^*(x))(h_{10}(x)-f(1,x))^2 \\
0 & = (1-p^*(x))q^*(x)(h_{01}(x)-f(0,x)-c_1^{\top}D_{\mu}(x;\beta^*)-c_2^{\top}D_{\phi}(X;\theta^*)) \\
& \quad + (1-p^*(x))(1-q^*(x))(h_{00}(x)-f(0,x))
\end{align*}
Conversely for any fixed $f \in \cs_{SX}$,
taking the derivative of the objective shows that the minimizers $c_1 \in \real^p$ and $c_2 \in \real^q$ must satisfy the first order conditions
\begin{align*}
0 & = \e[p^*(X)e^*(X)(h_{11}(X)-f_1(X)-c_1^{\top}D_{\mu}(X;\beta^*))D_{\mu}(X;\beta^*)] \\
& \quad + \e[(1-p^*(X))q^*(X)(h_{01}(X)-f(0,X)-c_1^{\top}D_{\mu}(X;\beta^*)-c_2^{\top}D_{\phi}(X;\theta^*))D_{\mu}(X;\beta^*)] \\
0 & = \e[(1-p^*(X))q^*(X)(h_{01}(X)-f(0,X)-c_1^{\top}D_{\mu}(X;\beta^*)-c_2^{\top}D_{\phi}(X;\theta^*))D_{\phi}(X;\theta^*)]
\end{align*}
The four first-order conditions are satisfied with $f(0,x)=f(1,x)=0$ and $c_1=c_2=0$ if and only if
\begin{align}
\e[[p^*(X)e^*(X)h_{11}(X)+(1-p^*(X))q^*(X)h_{01}(X)]D_{\mu}(X;\beta^*)] & = 0 \label{eq:S_M_3_perp_1} \\
\e[(1-p^*(X))q^*(X)h_{01}(X)D_{\phi}(X;\theta^*)] & = 0 \label{eq:S_M_3_perp_2} \\
e^*(x)h_{11}(x) + (1-e^*(x))h_{10}(x) & = 0 \label{eq:S_M_3_perp_3} \\
q^*(x)h_{01}(x) + (1-q^*(x))h_{00}(x) & = 0. \label{eq:S_M_3_perp_4}
\end{align}
Thus $\cs_{\cm_3}^{\perp}$ is the set of all functions $h$ satisfying~\eqref{eq:S_M_3_perp_1},~\eqref{eq:S_M_3_perp_2},~\eqref{eq:S_M_3_perp_3}, and~\eqref{eq:S_M_3_perp_4}.
Recognizing that~\eqref{eq:S_M_3_perp_3} and~\eqref{eq:S_M_3_perp_4} determine $h_{00}$ and $h_{10}$ in terms of $h_{11}$ and $h_{01}$, respectively,
which are free to vary other than the moment constraints given by~\eqref{eq:S_M_3_perp_1} and~\eqref{eq:S_M_3_perp_2},
we can compactly write $\cs_{\cm_3}^{\perp}$ as the set of all mappings
\[
r \mapsto f(s,x)\left[z-(1-z)\left(\frac{se^*(x)}{1-e^*(x)} + \frac{(1-s)q^*(x)}{1-q^*(x)}\right)\right]
\]
such that $f \in \ch_{SX}$ with
\begin{align*}
\e[(p^*(X)e^*(X)f(1,X)+(1-p^*(X))q^*(X)f(0,X))D_{\mu}(X;\beta^*)] & = 0 \quad \text{and} \\
\e[(1-p^*(X))q^*(X)f(1,X)D_{\phi}(X;\theta^*)] & = 0.
\end{align*}

Finally, we compute the projection~\eqref{eq:eif_raw}.
Given the initial influence function $\varphi_0$,
we have by the form of $\cs_{\cm_3}^{\perp}$ derived in the previous paragraph and Theorem~\ref{thm:semiparametric_tangent_space} that
\[
\pi(\varphi_0;\ct_{\cm_3}^{\perp}) = w \mapsto \left(s\Gamma_1^*(x)+(1-s)\Gamma_0^*(x)\right)\left(z-(1-z)\left[\frac{se^*(x)}{1-e^*(x)} + \frac{(1-s)q^*(x)}{1-q^*(x)}\right]\right)(y-m^*(r))
\]
where $\Gamma_0^*$ and $\Gamma_1^*$ solve the optimization problem
\begin{align*}
& \min_{\Gamma_0 \in \ch_X, \Gamma_1 \in \ch_X} \e\left[\left(\varphi_0(W;\tau^*,\eta^*)-\left[S\Gamma_1(X)+(1-S)\Gamma_0(X)\right]\nu^*(R)\right)^2\right] \\
\text{s.t.} \quad & \e[[p^*(X)e^*(X)\Gamma_1(X) + (1-p^*(X))q^*(X)\Gamma_0(X)]D_{\mu}(X;\beta^*)] = 0, \quad \text{and} \\
& \e[(1-p^*(X))q^*(X)\Gamma_0(X)D_{\phi}(X;\theta^*)] = 0
\end{align*}
with
\[
\nu^*(r) = \nu(r;\eta^*) \text{ for } \nu(r;\eta) = \left[z-(1-z)\left(\frac{se^*(x)}{1-e^*(x)} + \frac{(1-s)q^*(x)}{1-q^*(x)}\right)\right](y-m^*(r)).
\]
Once again, we rewrite the objective by conditioning on $X$:
\begin{align*}
& \e\left[\left(\varphi_0(W;\tau^*,\eta^*)-\left(S\Gamma_1(X)+(1-S)\Gamma_0(X)\right)\nu^*(R)\right)^2\right] \\
& = \e[\varphi_0(W;\tau^*,\eta^*)^2 \mid X] + \e\left[\Gamma_1^2(X)p^*(X)e^*(X)\left(V_{11}^*(X)+\frac{e^*(X)}{1-e^*(X)}V_{10}^*(X)\right)\right] \\
& \quad + \e\left[\Gamma_0^2(X)(1-p^*(X))q^*(X)\left(V_{01}^*(X)+\frac{q^*(X)}{1-q^*(X)}V_{00}^*(X)\right)\right] -2\e[\cj(X;\Gamma_0,\Gamma_1,\eta,\tau,\varphi_0)]
\end{align*}
where
\[
\cj(x;\Gamma_0,\Gamma_1,\eta,\tau,\varphi_0) = \e\left[\left(S\Gamma_1(X)+(1-S)\Gamma_0(X)\right)\nu^*(R)\varphi_0(W;\tau^*,\eta^*) \mid X=x\right].
\]
To deal with the two moment constraints,
we introduce Lagrange multipliers $\Lambda_1 \in \real^p $ and $\Lambda_2 \in \real^q$, similar to Appendix~\ref{app:restricted_moment_model}.
Dropping the constant $\e[\varphi_0(W;\tau^*,\eta^*)^2 \mid X=x]$ term from the objective,
the Lagrangian is $\e[\mathcal{L}(X)]$ where
\begin{align*}
\mathcal{L}(x) & = \Gamma_1^2(x)p^*(x)e^*(x)\left(V_{11}^*(x)+\frac{e^*(x)}{1-e^*(x)}V_{10}^*(x)\right) \\
& \quad + \Gamma_0^2(x)(1-p^*(x))q^*(x)\left(V_{01}^*(x)+\frac{q^*(x)}{1-q^*(x)}V_{00}^*(x)\right) \\
& \quad -2 \cj(x;\Gamma_0,\Gamma_1,\eta^*,\tau^*,\varphi_0) \\
& \quad - 2\Lambda_1^{\top}[p^*(x)e^*(x)\Gamma_1(x)+(1-p^*(x))q^*(x)\Gamma_0(x)]D_{\mu}(x;\beta^*) \\
& \quad - 2\Lambda_2^{\top}[(1-p^*(x))q^*(x)\Gamma_0(x)D_{\phi}(x;\theta^*)].
\end{align*}
We can minimize the Lagrangian over $\Gamma_0 \in \ch_X$ and $\Gamma_1 \in \ch_X$ by minimizing $\mathcal{L}(x)$ pointwise,
giving
\begin{align}
\Gamma_0^*(x) & = (1-q^*(x))\left(\frac{\Lambda_1^{\top}D_{\mu}(x;\beta^*)+\Lambda_2^{\top}D_{\phi}(x;\theta^*)}{(1-q^*(x))V_{01}^*(x)+q^*(x)V_{00}^*(x)}\right) \nonumber \\
& \quad + \frac{(1-q^*(x))\cj_0(x;\eta,\tau,\varphi_0)}{(1-p^*(x))q^*(x)[(1-q^*(x))V_{01}^*(x)+q^*(x)V_{00}^*(x)]} \label{eq:Gamma_0_star} \\
\Gamma_1^*(x) & = (1-e^*(x))\left(\frac{\Lambda_1^{\top}D_{\mu}(x;\beta^*)}{(1-e^*(x))V_{11}^*(x)+e^*(x)V_{10}^*(x)}\right) \nonumber \\
& \quad + \frac{(1-e^*(x)) \cj_1(x;\eta,\tau,\varphi_0)}{p^*(x)e^*(x)[(1-e^*(x))V_{11}^*(x)+e^*(x)V_{10}^*(x)]} \label{eq:Gamma_1_star}
\end{align}
where the Lagrange multipliers can be solved for by plugging back into the moment constraints and
\begin{align*}
\cj_1(x;\eta,\tau,\varphi_0) & = \e\left[S\left(Z-(1-Z)\frac{e^*(X)}{1-e^*(X)}\right)(Y-m^*(R))\varphi_0(W;\tau^*,\eta^*) \mid X=x\right] \\
\cj_0(x;\eta,\tau,\varphi_0) & = \e\left[(1-S)\left(Z-(1-Z)\frac{q^*(X)}{1-q^*(X)}\right)(Y-m^*(R))\varphi_0(W;\tau^*,\eta^*) \mid X=x\right]
\end{align*}
which come from writing
\[
\cj(x;\eta,\tau,\varphi_0) = \cj_1(x;\eta,\tau,\varphi_0)\Gamma_1(x) + \cj_1(x;\eta,\tau,\varphi_0)\Gamma_0(x).
\]
This shows that
$\Lambda_1$ and $\Lambda_2$ solve the linear matrix system
\begin{align*}
A_{\mu\mu}\Lambda_1 + A_{\mu\phi}\Lambda_2 & = -b_{\mu} \\
A_{\mu\phi}^{\top} + A_{\phi\phi}\Lambda_2 & = -b_{\phi}
\end{align*}
where
\begin{align}
A_{\mu\mu} & = \e\left[\frac{p^*(X)e^*(X)(1-e^*(X))}{(1-e^*(X))V_{11}^*(X)+e^*(X)V_{10}^*(X)}D_{\mu}(X;\beta^*)D_{\mu}(X;\beta^*)^{\top}\right] \nonumber \\
& \quad + \e\left[\frac{(1-p^*(X))q^*(X)(1-q^*(X))}{(1-q^*(X))V_{01}^*(X)+q^*(X)V_{00}^*(X)}D_{\mu}(X;\beta^*)D_{\mu}(X;\beta^*)^{\top}\right]  \in \real^{p \times p}  \label{eq:A_mu_mu} \\
A_{\mu\phi} & = \e\left[\frac{(1-p^*(X))q^*(X)(1-q^*(X))}{(1-q^*(X))V_{01}^*(X)+q^*(X)V_{00}^*(X)}D_{\mu}(X;\beta^*)D_{\phi}(X;\theta^*)^{\top}\right]  \in \real^{p \times q} \label{eq:A_mu_phi} \\
A_{\phi\phi} & = \e\left[\frac{(1-p^*(X))q^*(X)(1-q^*(X))}{(1-q^*(X))V_{01}^*(X)+q^*(X)V_{00}^*(X)}D_{\phi}(X;\theta^*)D_{\phi}(X;\theta^*)^{\top}\right] \in \real^{q \times q} \label{eq:A_phi_phi} \\
b_{\mu} & = \e\left[\frac{(1-e^*(X))\cj_1(X;\eta^*,\tau^*,\varphi_0)}{(1-e^*(X))V_{11}^*(X)+e^*(X)V_{10}^*(X)}D_{\mu}(X;\beta^*)\right] \nonumber \\
& \quad + \e\left[\frac{(1-q^*(X))\cj_0(X;\eta^*,\tau^*,\varphi_0)}{(1-q^*(X))V_{01}^*(X)+q^*(X)V_{00}^*(X)}D_{\mu}(X;\beta^*)\right] \in \real^p 
\label{eq:b_mu} \\
b_{\phi} & = \e\left[\frac{(1-q^*(X))\cj_0(X;\eta^*,\tau^*,\varphi_0)}{(1-q^*(X))V_{01}^*(X)+q^*(X)V_{00}^*(X)}D_{\phi}(X;\theta^*)\right] \in \real^q.
\end{align}

Taking $\varphi_0=\varphi_{0,\obs}$ to give a more explicit EIF for $\tau_{\obs}$,
we can easily compute
\begin{align}
\cj_0(x;\eta,\tau,\varphi_{0,\obs}) & = 0, \label{eq:J_0} \\
\cj_1(x;\eta,\tau,\varphi_{0,\obs}) & = \frac{1-p^*(x)}{1-\rho^*}\left(V_{11}^*(x)+\frac{e^*(x)}{1-e^*(x)}V_{10}^*(x)\right). \label{eq:J_1}
\end{align}
Then some of the other expressions simplify:
\begin{align}
b_{\phi} & = 0 \label{eq:b_phi_tau_obs} \\
b_{\mu} & = \e\left[\frac{1-p^*(X)}{1-\rho^*}D_{\mu}(X;\beta^*)\right]=\e[D_{\mu}(X;\beta^*) \mid S=0] \label{eq:b_mu_tau_obs} \\
\Gamma_0^*(x) & = (1-q^*(x))\left(\frac{\Lambda_1^{\top}D_{\mu}(x;\beta^*)+\Lambda_2^{\top}D_{\phi}(x;\theta^*)}{(1-q^*(x))V_{01}^*(x)+q^*(x)V_{00}^*(x)}\right) \label{eq:Gamma_0_star_tau_obs} \\
\Gamma_1^*(x) & = (1-e^*(x))\left(\frac{\Lambda_1^{\top}D_{\mu}(x;\beta^*)}{(1-e^*(x))V_{11}^*(x)+e^*(x)V_{10}^*(x)}\right) + \frac{(1-p^*(x))}{(1-\rho^*)p^*(x)e^*(x)}. \label{eq:Gamma_1_star_tau_obs}
\end{align}
Solving the matrix equations for the Lagrange multipliers then gives
\begin{align}
\Lambda_1 & = -\left(A_{\mu\mu} - A_{\mu\phi}A_{\phi\phi}^{-1}A_{\mu\phi}^{\top}\right)^{-1}b_{\mu}, \label{eq:Lambda_1} \\
\Lambda_2 = -A_{\phi\phi}^{-1}A_{\mu\mu}^{\top}\Lambda_1 & =A_{\phi\phi}^{-1}A_{\mu\mu}^{\top}\left(A_{\mu\mu} - A_{\mu\phi}A_{\phi\phi}^{-1}A_{\mu\phi}^{\top}\right)^{-1}b_{\mu}. \label{eq:Lambda_2}
\end{align}
Putting everything together,
the EIF for $\tau_{\obs}$ in the model $\cp_{\cm_3}$ is given by
\begin{align}
& \varphi_{\eff,\obs}^{(3)}(w;\tau^*,\eta^*) = \varphi_{0,\obs}(w;\tau^*,\eta^*) -  \nonumber \\
& \quad \left(s\Gamma_1^*(x)+(1-s)\Gamma_0^*(x)\right)\left(z-(1-z)\left[\frac{se^*(x)}{1-e^*(x)} + \frac{(1-s)q^*(x)}{1-q^*(x)}\right]\right)(y-m^*(r)) \nonumber \\
& = \frac{1-s}{1-\rho^*}(m_{11}^*(x)-m_{10}^*(x)-\tau_{\obs}^*) + s\Delta(z,x,y;\eta^*)\left(\frac{1-p^*(x)}{p^*(x)(1-\rho^*)}-e^*(x)\Gamma_1^*(x)\right) \nonumber \\
& \quad -(1-s)\Gamma_0^*(x)\left(z-(1-z)\frac{q^*(x)}{1-q^*(x)}\right)(y-m^*(r)) \label{eq:our_eif}
\end{align}
where $\Gamma_0^*$ and $\Gamma_1^*$ are given by~\eqref{eq:Gamma_0_star_tau_obs} and~\eqref{eq:Gamma_1_star_tau_obs},
respectively,
for $\Lambda_1$ and $\Lambda_2$ as in~\eqref{eq:Lambda_1} and~\eqref{eq:Lambda_2}.

We can now verify this matches the result in Theorem 4 of~\citet{yang2024datafusion}.
While that Theorem is stated as giving the efficient \emph{score} of $\tau_{\obs}$ rather than the efficient \emph{influence function} of $\tau_{\obs}$,
this is incorrect according to S. Yang (personal communication, April 19, 2025).
We correct the theorem below;
essentially we must replace the efficient score $s_{\beta}(w;\eta^*)$ for $\beta$ with the efficient influence function $\varphi_{\beta}(w;\eta^*)$ for $\beta$.
\begin{theorem*}[Theorem 4,~\citet{yang2024datafusion}, corrected]
Suppose Assumptions 1-3 hold.
Then the EIF of $\tau_{\obs}$ at the truth $\tau_{\obs}^*$ is 
\begin{equation}
\label{eq:yang_eif}
\tilde{\varphi}_{\eff,\obs}^{(3)}(w;\tau_{\obs}^*,\eta^*) = \frac{1-s}{1-\rho^*}(m_{11}^*(x)-m_{10}^*(x)-\tau_{\obs}^*) + \e\left[D_{\mu}(X;\beta^*)^{\top} \mid S=0\right]\varphi_{\beta}(w;\eta^*)
\end{equation}
where $\varphi_{\beta}(w;\eta^*)$ is the first $p$ components of
\[
\varphi_{\lambda}(w;\eta^*) = \left(\e\left[s_{\lambda}(W;\eta^*)s_{\lambda}(W;\eta^*)^{\top}\right]\right)^{-1} s_{\lambda}(w;\eta^*)
\]
for
\begin{align*}
s_{\lambda}(w;\eta^*) & = (Z-\e[Z\omega^*(R) \mid S,X]\e[\omega^*(R) \mid S,X]^{-1})\omega^*(R)(Y-m^*(R))
\begin{bmatrix}
D_{\mu}(x;\beta^*) \\
D_{\phi}(x;\theta^*)
\end{bmatrix}, \\
\omega^*(r) & = \frac{1}{V^*(r)}
\end{align*}
the efficient score for $\lambda=(\beta,\theta)$,
as given explicitly by Theorem 1 of~\citet{yang2024datafusion}.
\end{theorem*}
Showing that $\varphi_{\eff,\obs}^{(3)} = \tilde{\varphi}_{\eff,\obs}^{(3)}$ requires some algebra which we outline below.
First we compute
\begin{align*}
\e[\omega^*(R) \mid S,X] & = S\left(\frac{e^*(X)}{V_{11}^*(X)}+\frac{1-e^*(X)}{V_{10}^*(X)}\right) + (1-S)\left(\frac{q^*(X)}{V_{01}^*(X)} + \frac{1-q^*(X)}{V_{00}^*(X)}\right), \quad \text { and } \\
\e[Z\omega^*(R) \mid S,X] & = S\left(\frac{e^*(X)}{V_{11}^*(X)}\right) + (1-S)\left(\frac{q^*(X)}{V_{01}^*(X)}\right).
\end{align*}
Then
\begin{align}
\Omega^*(R) & := (Z-\e[Z\omega^*(R) \mid S,X]\e[\omega^*(R) \mid S,X]^{-1})\omega^*(R) \nonumber \\
& \quad = S\frac{Z(1-e^*(X))-(1-Z)e^*(X)}{V_{10}^*(X)e^*(X)+V_{11}^*(X)(1-e^*(X))}+(1-S) \frac{Z(1-q^*(X))-(1-Z)q^*(X)}{V_{00}^*(X)q^*(X)+V_{01}^*(X)(1-q^*(X))}. \label{eq:Omega_star}
\end{align}
This allows us to compute
\begin{align*}
& \e[s_{\lambda}(W;\eta^*)s_{\lambda}(W;\eta^*)^{\top}] \\
& \quad = \e\left(\e[(\Omega^*(R))^2(Y-m^*(R))^2 \mid X]
\begin{bmatrix}
D_{\mu}(X;\beta^*)D_{\mu}(X;\beta^*)^{\top} & D_{\mu}(X;\beta^*)D_{\phi}(X;\theta^*)^{\top} \\
D_{\phi}(X;\theta^*)D_{\mu}(X;\beta^*)^{\top} & D_{\phi}(X;\theta^*)D_{\phi}(X;\theta^*)^{\top}
\end{bmatrix}
\right)
\end{align*}
with
\[
\e[(\Omega^*(R))^2(Y-m^*(R))^2 \mid X] = \frac{p^*(X)e^*(X)(1-e^*(X))}{V_{10}^*(X)e^*(X)+V_{11}^*(X)(1-e^*(X))} + \frac{(1-p^*(X))q^*(X)(1-q^*(X))}{(V_{00}^*(X)q^*(X)+V_{01}^*(X)(1-q^*(X))}.
\]
Hence
\[
\e[s_{\lambda}(W;\eta^*)s_{\lambda}(W;\eta^*)^{\top}] = 
\begin{bmatrix}
A_{\mu\mu} & A_{\mu\phi} \\
A_{\mu\phi}^{\top} & A_{\phi\phi}
\end{bmatrix}
\]
where $A_{\mu\mu}$, $A_{\mu\phi}$, and $A_{\phi\phi}$ are as in~\eqref{eq:A_mu_mu},~\eqref{eq:A_mu_phi}, and~\eqref{eq:A_phi_phi}.
So
\[
\varphi_{\lambda}(w;\eta^*) = 
\begin{bmatrix}
A_{\mu\mu} & A_{\mu\phi} \\
A_{\mu\phi}^{\top} & A_{\phi\phi}
\end{bmatrix}^{-1}s_{\lambda}(w;\eta^*)
\]
and by block matrix inversion formulas
\[
\varphi_{\beta}(w;\eta^*) = 
\begin{bmatrix}
(A_{\mu\mu}-A_{\mu\phi}A_{\phi\phi}^{-1}A_{\mu\phi}^{\top})^{-1} & -(A_{\mu\mu}-A_{\mu\phi}A_{\phi\phi}^{-1}A_{\mu\phi}^{\top})^{-1}A_{\mu\phi}A_{\phi\phi}^{-1}
\end{bmatrix}
s_{\lambda}(w;\eta^*).
\]
Then by~\eqref{eq:yang_eif},~\eqref{eq:Lambda_1}, and~\eqref{eq:Lambda_2} we can write
\begin{align}
\label{eq:yang_eif_final}
\tilde{\varphi}_{\eff,\obs}^{(3)}(w;\tau^*,\eta^*) = \frac{1-s}{1-\rho^*}(m_{11}^*(x)-m_{10}^*(x)-\tau_{\obs}^*) - 
\begin{bmatrix}
\Lambda_1^{\top} & \Lambda_2^{\top}
\end{bmatrix}
s_{\lambda}(w;\eta^*).
\end{align}
Using~\eqref{eq:Omega_star} we compute
\begin{align*}
\begin{bmatrix}
\Lambda_1^{\top} & \Lambda_2^{\top}
\end{bmatrix}
& s_{\lambda}(w;\eta^*) \nonumber \\
& = \Omega^*(r)[y-m^*(r)][\Lambda_1^{\top}D_{\mu}(x;\beta^*) + \Lambda_2^{\top}D_{\phi}(x;\theta^*)] \\
& = s\Delta(z,x,y;\eta^*)\frac{e^*(x)(1-e^*(x))}{V_{10}^*(x)e^*(x)+V_{11}^*(x)(1-e^*(x))}[\Lambda_1^{\top}D_{\mu}(x;\beta^*) + \Lambda_2^{\top}D_{\phi}(x;\theta^*)] \\
& \quad + (1-s)\frac{z(1-q^*(x))-(1-z)q^*(x)}{V_{00}^*(x)q^*(x)+V_{01}^*(x)(1-q^*(x))}(y-m^*(r))[\Lambda_1^{\top}D_{\mu}(x;\beta^*) + \Lambda_2^{\top}D_{\phi}(x;\theta^*)] \\
& = s\Delta(z,x,y;\eta^*)\left[e^*(x)\Gamma_1^*(x)-\frac{1-p^*(x)}{p^*(x)(1-\rho^*)}\right] \\
& \quad + (1-s)\Gamma_0^*(x)\left(z-(1-z)\frac{q^*(x)}{1-q^*(x)}\right)(y-m^*(r))
\end{align*}
where the last equality follows by~\eqref{eq:Gamma_0_star_tau_obs} and~\eqref{eq:Gamma_1_star_tau_obs}.
Plugging into~\eqref{eq:yang_eif_final} verifies that indeed, $\varphi_{\eff,\obs}^{(3)} = \tilde{\varphi}_{\eff,\obs}^{(3)}$.

\subsection{Linear confounding bias}
\label{app:linear_confounding_bias}

Here we derive the EIF for an arbitrary scalar estimand $\tau$, which can be estimated by an RAL estimator with influence function $\varphi_0$,
in the linear confounding bias model $\cp_{\cm_4}$.

As argued in the main text, $\cm_4$ is a linear space so $\cs_{\cm_4}=\cm_4$ by Proposition~\ref{prop:S_M_equals_M}.
The remainder of the EIF derivation closely follows the outcome-mediated selection bias example in Section~\ref{sec:projections}.
As in that example,
we note $\cs_{\cm_4}^{\perp}$ is precisely set of all functions $f \in \ch_R$ for which the projection $g=\Pi(f;\cs_{\cm_4})$ is zero.
Given $f$,
by the form of $\cs_{\cm_4}=\cm_4$ we can write
\[
g(r) = s(1-z)g_1(x)+(1-s)zg_2(x)+(1-s)(1-z)g_3(x)+sz\kappa(x;\beta,g_1,g_2,g_3)
\]
where $g_1,g_2,g_3$ in $\ch_X$ and $\beta \in \real^q$ minimize the objective
\begin{align*}
\mathcal{O}(g_1,g_2,g_3,\beta) & = \e\left[\left(f(R)-S(1-Z)g_1(X)-(1-S)Zg_2(X)- (1-S)(1-Z)g_3(X)-SZ\kappa(X;\beta,g_1,g_2,g_3)\right)^2\right] \\
& = \e[p^*(X)e^*(X)(f_{11}(X)-\kappa(X;\beta,g_1,g_2,g_3))^2 + p^*(X)(1-e^*(X))(f_{10}(X)-g_1(X))^2] \\
& \quad + \e[(1-p^*(X))r^*(X)(f_{01}(X)-g_2(X))^2 + (1-p^*(X))(1-r^*(X))(f_{00}(X)-g_3(X))^2]
\end{align*}
over $\ch_X^3 \times \real^q$,
for
\[
\kappa(x;\beta,g_1,g_2,g_3) = g_1(x)+g_2(x)-g_3(x)+\psi(x)^{\top}\beta.
\]
The second equality follows from writing
\[
f(R) = SZf_{11}(X) + S(1-Z)f_{10}(X) + (1-S)Zf_{01}(X) + (1-S)(1-Z)f_{00}(X)
\]
and conditioning on $X$.
With the objective now an expectation of a function of $X$,
for fixed $\beta$
we can minimize over $g_1$, $g_2$, and $g_3$ pointwise to obtain the first-order conditions
\begin{align*}
0 & = -p^*(x)e^*(x)(f_{11}(x)-\kappa(x;\beta,g_1,g_2,g_3)) - p^*(x)(1-e^*(x))(f_{10}(x)-g_1(x)) \\
0 & = -p^*(x)e^*(x)(f_{11}(x)-\kappa(x;\beta,g_1,g_2,g_3)) -(1-p^*(x))r^*(x)(f_{01}(x)-g_2(x)) \\
0 & = p^*(x)e^*(x)(f_{11}(x)-\kappa(x;\beta,g_1,g_2,g_3))-(1-p^*(x))(1-r^*(x))(f_{00}(x)-g_3(x))
\end{align*}
for all $x \in \cx$.
Further,
setting the partial derivative of the objective with respect to $\beta$ equal to zero yields the additional first-order condition 
\[
\e[p^*(X)e^*(X)(f_{11}(X)-\kappa(X;\beta,g_1,g_2,g_3))\psi(X)]=0.
\]
Rearranging,
we see these first-order conditions are satisfied by $g=0$ (i.e. $g_1=g_2=g_3=0$, $\beta=0$) if and only if 
\begin{align}
(1-e^*(x))f_{10}(x) & = -e^*(x)f_{11}(x) \label{eq:T_M_perp_1} \\
(1-p^*(x))r^*(x)f_{01}(x) & = -p^*(x)e^*(x)f_{11}(x) \label{eq:T_M_perp_2}  \\
(1-p^*(x))(1-r^*(x))f_{00}(x) & = p^*(x)e^*(x)f_{11}(x) 
\label{eq:T_M_perp_3} \\
0 & = \e\left[p^*(X)e^*(X)f_{11}(X)\psi(X)\right] \label{eq:T_M_perp_4} 
\end{align}
for all $x \in \cx$.
So $\cs_{\cm_4}$ is the set of all functions $f \in \ch_R$ satisfying these equations,
and by Theorem~\ref{thm:semiparametric_tangent_space}
\[
\ct_{\cm_4}^{\perp} = \{w \mapsto f(r)(y-m^*(r)) \mid f \in \ch_R \text{ satisfies~\eqref{eq:T_M_perp_1},~\eqref{eq:T_M_perp_2},~\eqref{eq:T_M_perp_3}, and \eqref{eq:T_M_perp_4}}.
\]
Noting that~\eqref{eq:T_M_perp_1},~\eqref{eq:T_M_perp_2}, and~\eqref{eq:T_M_perp_3}
determine $f_{10}$, $f_{01}$, and $f_{00}$ in terms of $f_{11}$ and letting $\nu(x)=p^*(x)e^*(x)f_{11}(x)$,
we can write this succinctly as
\begin{equation}
\label{eq:T_M_1_perp}
\ct_{\cm_4}^{\perp} = \{w \mapsto h(r;\eta^*)\nu(x)(y-m^*(r)) \mid \nu \in \ch_X, \ \e[\nu(X)\psi(X)]=0\}.
\end{equation}
where $h(r;\eta)$ is as in~\eqref{eq:f_1}.

To compute the projection~\eqref{eq:eif_raw},
as in the other examples we condition on $X$ to write the objective as the expectation of a function of $X$ and minimize this function pointwise.
Unlike in the projection for computing $\varphi_{\eff}^{(5)}$ in the main text,
but like the computation in Appendix~\ref{app:restricted_moment_model} for the EIF in the restricted moment model,
we deal with the constraints using Lagrange multipliers. 
By~\eqref{eq:T_M_perp_generalized},
an arbitrary element $g \in \ct_{\cm_4}^{\perp}$ takes the form 
\[
g(w)=h(r;\eta^*)\nu(x)(y-m^*(r))
\]
for some $\nu \in \ch_X$ with $\e[\nu(X)\psi(X)]=0$.
Then the EIF is $\varphi_{\eff}^{(4)}=\varphi_0-g_{\eff}^{(4)}$ where 
\[
g_{\eff}^{(4)}(w)=h(r;\eta^*)\nu^*(x;\varphi_0)(y-m^*(r))
\]
for $\nu^*(\cdot;\varphi_0)$ solving the optimization problem
\begin{align*}
\mbox{min. } & \e\left[\left(\varphi_0(W;\tau^*,\eta^*)-h(R;\eta^*)\nu(X)(Y-m^*(R))\right)^2\right] \mbox{over $\nu \in \ch_X$} \\
\mbox{s.t. } & \e[\nu(X)\psi(X)]=0
\end{align*}
We eliminate the constraint by introducing Lagrange multipliers $\lambda \in \real^q$,
and then rewrite the Lagrangian by conditioning on $X$:
\begin{align*}
\mbox{min. } & \e[\tilde{g}(X;\nu)] \mbox{ over $\nu \in \ch_X$} \mbox{ where } \\
\tilde{g}(X;\nu) & = \e\left[\left(\varphi_0(W;\tau^*,\eta^*)-h(R;\eta^*)\nu(X)(Y-m^*(R))\right)^2 \mid X\right]  \\
& = \nu(X)^2\e[h^2(R;\eta^*)(Y-m^*(R))^2\mid X] -2\nu(X)I(X;\eta^*,\tau^*,\varphi_0) -2\nu(X)\lambda^{\top}\psi(X) \
\end{align*}
Since $\tilde{g}(x;\nu)$ is quadratic in $\nu$,
the minimizer is
\[
\nu^*(x) = \frac{I(x;\eta^*,\tau^*,\varphi_0) + \lambda^{\top}\psi(x)}{\Sigma(x;\eta^*)}
\]
where $I(x;\eta,\tau,\varphi_0)$ is given by~\eqref{eq:I_x}.
This matches the formula for $\nu(x;\eta^*,\tau^*,\varphi_0)$ in~\eqref{eq:zeta_1_star},
upon noting that
\[
p^*(x)e^*(x)\e[h^2(R;\eta^*)(Y-m^*(R))^2 \mid X=x]=\Sigma(x;\eta^*)
\]
for $\Sigma(x;\eta)$ as in~\eqref{eq:Sigma_star}.
We solve for the Lagrange multiplier $\lambda$
using the constraint 
\[
\e[\nu^*(X;\nu_0)\psi(X)]=0
\]
which gives $\lambda = \lambda(\eta^*,\tau^*,\eta_0)$ as in~\eqref{eq:lambda}.
Thus
\[
g_{\eff}^{(4)}(w) = h(r;\eta)\nu(x;\eta^*,\tau^*,\varphi_0)(y-m^*(r))
\]
which shows~\eqref{eq:eif_4}.

\subsection{Known RCT propensity scores}
\label{app:propensity_score_does_not_help}
We characterize the EIF $\tilde{\varphi}_{\eff}$ for any pathwise differentiable estimand $\tau \in \real$ in the model $\tilde{\cp}_{\cm}$ 
in terms of the EIF $\varphi_{\eff}$ for the analogous model $\cp_{\cm}$ where the RCT propensity score $e^*$ is unknown.
Our result will work for any mean function collection $\cm$ satisfying the conditions of Corollary~\ref{cor:known_propensity_score}.

It is easy to see that the space
\[
\tilde{\cu} = \{w \mapsto sh(x)(z-e^*(x)) \mid h \in \ch_X\}
\]
is orthogonal to $\ct_{\cm}^{\perp}$.
Hence
\[
\tilde{\varphi}_{\eff} = \varphi_0 - \Pi(\varphi_0;\tilde{\ct}_{\cm}^{\perp}) = \varphi_0-\Pi(\varphi_0;\ct_{\cm}^{\perp})- \Pi(\varphi_0;\tilde{\cu}) = \varphi_{\eff} - \Pi(\varphi_0;\tilde{\cu}) 
\]
where as in the main text,
$\varphi_0$ is the influence function of any initial RAL estimator.
Thus it suffices to compute $\Pi(\varphi_0;\tilde{\cu})$.
By definition
\[
\Pi(\varphi_0;\tilde{\cu})(w) = s\nu^*(x)(z-e^*(x))
\]
where
\[
\nu^* = \argmin_{h \in \ch_X} \e\left[\left(\varphi_0(W)-Sh(X)(Z-e^*(X))\right)^2\right].
\]
As usual, we expand the right-hand side and condition on $X$:
\begin{align*}
& \argmin_{h \in \ch_X}  \e\left[\left(\varphi_0(W)-Sh(X)(Z-e^*(X))\right)^2\right] \\
& \quad = \argmin_{h \in \ch_X} \e\left[\e[S(Z-e^*(X))^2 \mid X]h(X)^2 -2\e[S(Z-e^*(X))\varphi_0(W) \mid X]h(X)\right]
\end{align*}
We compute
\[
\e[S(Z-e^*(X))^2 \mid X] = p^*(X)e^*(X)(1-e^*(X)) 
\]
and so by minimizing pointwise we conclude
\[
\nu^*(x) = \frac{\e[S(Z-e^*(X))\varphi_0(W) \mid X=x]}{p^*(x)e^*(x)(1-e^*(x))}.
\]

For $\varphi_0 \in \{\varphi_{0,\rct},\varphi_{0,\obs},\varphi_{0,\tgt}\}$ we have
\[
\e[S(Z-e^*(X))\varphi_0(W) \mid X] = 0
\]
and hence $\nu^*(x)=0$.
Thus,
knowing the RCT propensity score does not improve the semiparametric efficiency bound for estimating $\tau_{\rct}$, $\tau_{\obs}$, or $\tau_{\tgt}$ in any semiparametric model restricting the outcome mean function.
In particular,
the EIF's $\varphi_{\eff}^{(k)}, k=1,2,3,4,5$ in the models $\cp_{\cm_k}$ for these estimands are also the EIF's in the corresponding models $\tilde{\cp}_{\cm_k}$.

\section{Conditions for efficiency}
\label{app:efficiency_conditions}
In this section,
we provide specific primitive conditions under which the cross-fit one-step estimators $\hat{\tau}_{\eff}^{(4)}$ and $\hat{\tau}_{\eff}^{(5)}$ defined in Section~\ref{sec:applications} indeed attain the efficiency bound,
by proving that these conditions imply those of the generalized Lemma in Section~\ref{app:proof_dml_master}.
One condition is the following overlap assumption;
a version of overlap is needed to even identify causal estimands in the nonparametric model.
\begin{assumption}[Overlap]
\label{assump:overlap}
There exists $\delta>0$ such that $e^*(x) \in [\delta,1-\delta]$ whenever $p^*(x)>0$,
$r^*(x) \in [\delta,1-\delta]$,
and $p^*(x) \in \{0\} \cup [\delta,1-\delta]$ for all $x \in \cx$.
\end{assumption}
\begin{remark}
The condition on $p^*(x)$ in Assumption~\ref{assump:overlap} presumes that the support of the covariates in the RCT is contained within the support of the covariates of the observational dataset.
Unlike typical analyses,
we allow this inclusion to be strict since assumptions like linear confounding bias or outcome mediated selection permit extrapolation of the CATE beyond the support of the covariates in the RCT. 
\end{remark}

We also need the nuisance estimates to satisfy a similar overlap condition,
as well as $o_p(N^{-1/4})$ root-mean-square convergence rates,
typical in the double machine learning literature.

\begin{assumption}
\label{assump:regularity}
Cross-fit estimates of the nuisance functions $e^*$, $m_{10}^*$, $m_{01}^*$, $m_{00}^*$, $r^*$, and $p^*$
satisfy Assumption~\ref{assump:overlap}
along with the following for each $k=1,\ldots,K$:
\begin{itemize}
\item (Convergence of infinite-dimensional nuisance estimates)
\begin{align}
\label{eq:nuisance_rates}
& \|\hat{e}^{(-k)}-e^*\|_{2,P^*} +  \|\hat{m}_{sz}^{(-k)}-m_{sz}^*\|_{2,P^*} +  \|\hat{r}^{(-k)}-r^*\|_{2,P^*} +  \|\hat{p}^{(-k)}-p^*\|_{2,P^*} \\
& \quad = o_p(N^{-1/4}), \quad (s,z) \neq (1,1)
\end{align}
\item (Overlap in infinite-dimensional nuisance estimates)
With probability tending to 1 as $N \rightarrow \infty$,
the statement of Assumption~\ref{assump:overlap} holds with $e^*$, $r^*$, and $p^*$ replaced by their estimates $\hat{e}^{(-k)}$, $\hat{r}^{(-k)}$, and $\hat{p}^{(-k)}$,
respectively.
\end{itemize}
\end{assumption}

\subsection{Linear confounding bias}
\label{app:linear_confounding_bias_conditions}
\begin{proposition*}
\label{prop:cross_fit_dml}
Suppose Assumption~\ref{assump:overlap} holds and $m^* \in \cm_4$.
Consider a function
\[
g(w;\eta) = h(r;\eta)\zeta(x;\eta)(y-m(r)) + s(z-e^*(x))h_2(x;\eta)
\]
where $\eta$ includes $(p(\cdot),e(\cdot),m(\cdot),q(\cdot),V(\cdot))$,
$h(r;\eta)$ is as in~\eqref{eq:f_1},
$h_2(\cdot;\eta^*) \in \ch_X$,
and
$\zeta^*(x)=\zeta(x;\eta^*)$ satisfies $\e[\zeta^*(X)\psi(X)]=0$
so that $g^*(\cdot) = g(\cdot;\eta^*) \in \tilde{\ct}_{\cm_4}^{\perp}$ by~\eqref{eq:T_M_1_perp} and Corollary~\ref{cor:known_propensity_score}.
Further suppose that there are cross-fit estimates $\hat{\eta}^{(-k)}$ of $\eta^*$ satisfying Assumption~\ref{assump:regularity} with
\begin{equation}
\label{eq:m_hat_respect}
\hat{m}_{11}^{(-k)}(x) = \hat{m}_{10}^{(-k)}(x)+\hat{m}_{01}^{(-k)}(x)-\hat{m}_{00}^{(-k)}(x) + \psi(x)^{\top}\hat{\theta}^{(-k)}, \quad k=1,\ldots,K
\end{equation}
where $\|\hat{\theta}^{(-k)}-\theta^*\|_{2,P^*}=o_p(N^{-1/4})$.
Finally, assume that the rate conditions
$\|\zeta(\cdot;\hat{\eta}^{(-k)})-\zeta(\cdot;\eta^*)\|_{2,P^*} = o_p(N^{-1/4})$ and
$\|h_2\left(\cdot;\hat{\eta}^{(-k)}\right)-h_2\left(\cdot;\eta^*\right)\|_{2,P^*} = o_p(1)$ hold,
and additionally that
$V^*(s,z,x) \leq C$ and
$\max(|\zeta(x;\eta^*)|, |\zeta(x,\hat{\eta}^{(-k)})|) \leq C$ for all $(s,z) \in \{0,1\}^2$, $x \in \cx$, and folds $k=1,\ldots,K$.
Then the Lemma in Section~\ref{app:proof_dml_master} holds with $\cm=\cm_4$,
$g(\cdot;\tau^*,\eta^*) = g(\cdot;\eta^*)$,
and
$g\left(\cdot;\hat{\tau}^{(-k)},\hat{\eta}^{(-k)}\right)=g\left(\cdot;\hat{\eta}^{(-k)}\right)$.
\end{proposition*}

\begin{proof}
We must show the boundedness, approximate tangency, and rate conditions in the Lemma of Appendix~\ref{app:proof_dml_master} with
\begin{align*}
h_1^*(r) & = h(r;\eta^*)\zeta(x;\eta^*), \quad \hat{h}_1^{(-k)}(r) = h\left(r;\hat{\eta}^{(-k)}\right)\zeta\left(x;\hat{\eta}^{(-k)}\right) \\
h_2^*(x) & = h_2(x;\eta^*), \quad \hat{h}_2^{(-k)}(x) = h_2\left(x;\hat{\eta}^{(-k)}\right).
\end{align*}

\noindent
\textbf{Boundedness:} We have $V^*(s,z,x) \leq C$ for $P^*$-almost all $(s,z,x)$ by assumption.
Furthermore by the overlap condition (Assumption~\ref{assump:overlap}) 
it is clear that $h(r;\eta^*)$
is uniformly upper bounded for $P^*$-almost all $r$, 
a qualifier we omit hereafter in all inequalities and equalities in this proof,
in absolute value by $C\delta^{-2}$.
Similarly, by the overlap in infinite-dimensional nuisance estimates condition in Assumption~\ref{assump:regularity},
we have $h\left(r;\hat{\eta}^{(-k)}\right)$ is uniformly upper bounded in absolute value by $C\delta^{-2}$ for all $k=1,\ldots,K$.
With $\max\left(|\zeta(x;\eta^*)|,\Big|\zeta\left(x;\hat{\eta}^{(-k)}\right)\Big|\right) \leq C$ for all $k$,
we conclude that $\max\left(\|h_1^*\|_{\infty,P^*},\left\|\hat{h}_1^{(-k)}\right\|_{\infty,P^*}\right) \leq C^2\delta^{-2}$. 

\noindent
\textbf{Approximate tangency:} 
Approximate tangency holds trivially with $\hat{R}^{(-k)}=0$ in view of the observation that $\hat{m}^{(-k)} \in \cm_4=\cs_{\cm_4}$ for all $k$ by~\eqref{eq:m_hat_respect},
so that $\hat{m}^{(-k)}-m \in \cs_{\cm_4}$ as well since $\cs_{\cm_4}$ is a linear space.

\noindent
\textbf{Rate conditions:}
Fixing $k \in \{1,\ldots,K\}$,
we have $\|\hat{m}^{(-k)}-m^*\|_{2,P^*}=o_p(N^{-1/4})$ by~\eqref{eq:nuisance_rates}
and~\eqref{eq:m_hat_respect} in light of the assumption $\|\hat{\theta}^{(-k)}-\theta^*\|=o_p(N^{-1/4})$.
With $\|\hat{h}_2^{(-k)}-h_2^*\|_{2,P^*} = o_p(1)$ by assumption,
to show the rate conditions it suffices to show that $\|\hat{h}_1^{(-k)}-h_1^*\|_{2,P^*} = o_p(N^{-1/4})$.
To that end,
we write
\[
\hat{h}_1^{(-k)}(r)-h_1^*(r) = h\left(r;\hat{\eta}^{(-k)}\right)\zeta\left(x;\hat{\eta}^{(-k)}\right) - h(r;\eta^*)\zeta(x;\eta^*)
\]
and note that since $|h(r;\eta^*)| \leq C$ as shown above
and we have $\max\left(|\zeta(x;\eta^*)|,\Big|\zeta\left(x;\hat{\eta}^{(-k)}\right)\Big|\right) \leq C$ and $\|\zeta(x;\hat{\eta}^{(-k)})-\zeta(x;\eta^*)\|_{2,P^*} = o_p(N^{-1/4})$ by assumption,
it suffices by Lemma~\ref{lemma:multiplication} to show that
\begin{equation}
\label{eq:h_rate}
\int \left(h\left(r;\hat{\eta}^{(-k)}\right)-h(r;\eta^*)\right)^2 dP^*(w) = o_p(N^{-1/2}).
\end{equation}
We compute
\begin{align*}
& \left(h\left(r;\hat{\eta}^{(-k)}\right)-h(r;\eta^*)\right)^2 \\
& = sz\left(\frac{p^*(x)e^*(x)-\hat{p}^{(-k)}(x)\hat{e}^{(-k)}(x)}{\hat{p}^{(-k)}(x)\hat{e}^{(-k)}(x)p^*(x)e^*(x)}\right)^2 + s(1-z)\left(\frac{p^*(x)(1-e^*(x))-\hat{p}^{(-k)}(x)(1-\hat{e}^{(-k)}(x))}{\hat{p}^{(-k)}(x)(1-\hat{e}^{(-k)}(x))p^*(x)(1-e^*(x))}\right)^2  \nonumber \\
& \quad + (1-s)z\left(\frac{(1-p^*(x))r^*(x)(1-\hat{p}^{(-k)}(x))(\hat{r}^{(-k)}(x))}{(1-\hat{p}^{(-k)}(x))\hat{r}^{(-k)}(x)(1-p^*(x))r^*(x)}\right)^2 \nonumber \\
& \quad + (1-s)(1-z)\left(\frac{(1-p^*(x))(1-r^*(x))-(1-\hat{p}^{(-k)}(x))(1-\hat{r}^{(-k)}(x))}{(1-\hat{p}^{(-k)}(x))(1-\hat{r}^{(-k)}(x))(1-p^*(x))(1-r^*(x))}\right)^2. \nonumber
\end{align*}
In light of~\eqref{eq:nuisance_rates},
we apply Lemma~\ref{lemma:multiplication} four times to show that
\begin{align*}
\int \left(\hat{p}^{(-k)}(x)\hat{e}^{(-k)}(x)-p^*(x)e^*(x)\right)^2 dP^*(w) & = o_p(N^{-1/2}) \\
\int \left(\hat{p}^{(-k)}(x)(1-\hat{e}^{(-k)}(x))-p^*(x)(1-e^*(x))\right)^2 dP^*(w) & = o_p(N^{-1/2}) \\
\int \left((1-\hat{p}^{(-k)}(x))\hat{r}^{(-k)}(x)-(1-p^*(x))r^*(x)\right)^2 dP^*(w) & = o_p(N^{-1/2}) \\
\int \left((1-\hat{p}^{(-k)}(x))(1-\hat{r}^{(-k)}(x))-(1-p^*(x))(1-r^*(x))\right)^2 dP^*(w) & = o_p(N^{-1/2})
\end{align*}
Then by Assumption~\ref{assump:overlap} and overlap in infinite-dimensional nuisance estimates,
we conclude that indeed~\eqref{eq:h_rate} holds.
\end{proof}

\subsection{Outcome-mediated selection bias}

\begin{proposition*}
Suppose Assumption~\ref{assump:overlap} holds and $m^* \in \cm_5$.
Consider a function
\[
g(w;\eta) = f(r;\eta)\zeta(x;\eta)(y-m(r)) + s(z-e^*(x))h_2(x;\eta)
\]
where $\eta$ includes $(p(\cdot),e(\cdot),m(\cdot),q(\cdot))$,
$f(r;\eta)$ is as in~\eqref{eq:f},
$h_2(\cdot;\eta^*) \in \ch_X$,
and
$\zeta^*(\cdot)=\zeta(\cdot;\eta^*) \in \ch_X$
so that $g^*(\cdot) = g(\cdot;\eta^*) \in \tilde{\ct}_{\cm_5}^{\perp}$ by~\eqref{eq:T_M_perp_outcome_selection_bias} and Corollary~\ref{cor:known_propensity_score}.
Further suppose that there are cross-fit estimates $\hat{\eta}^{(-k)}$ of $\eta^*$ satisfying Assumption~\ref{assump:regularity} with
\begin{equation}
\label{eq:m_hat_respect_selection}
\hat{m}_{11}^{(-k)}(x) = \ell^{-1}\left(\ell\left(\hat{m}_{10}^{(-k)}(x)\right)+\ell\left(\hat{m}_{01}^{(-k)}(x)\right)-\ell\left(\hat{m}_{00}^{(-k)}(x)\right)\right), \quad k=1,\ldots,K.
\end{equation}
Finally, assume that the rate conditions
$\|\zeta(\cdot;\hat{\eta}^{(-k)})-\zeta(\cdot;\eta^*)\|_{2,P^*} = o_p(N^{-1/4})$ and
$\|h_2\left(\cdot;\hat{\eta}^{(-k)}\right)-h_2\left(\cdot;\eta^*\right)\|_{2,P^*} = o_p(1)$ hold,
and additionally that
$\max(|\zeta(x;\eta^*)|, |\zeta(x,\hat{\eta}^{(-k)})|) \leq C$ for all $(s,z) \in \{0,1\}^2$, $x \in \cx$, and folds $k=1,\ldots,K$.
Then the Lemma in Section~\ref{app:proof_dml_master} holds with $\cm=\cm_5$,
$g(\cdot;\tau^*,\eta^*) = g(\cdot;\eta^*)$,
and
$g\left(\cdot;\hat{\tau}^{(-k)},\hat{\eta}^{(-k)}\right)=g\left(\cdot;\hat{\eta}^{(-k)}\right)$.
\end{proposition*}
\begin{proof}
The proof is similar in structure to the argument in Appendix~\ref{app:linear_confounding_bias_conditions}.
We must show the boundedness, approximate tangency, and rate conditions in the Lemma of Appendix~\ref{app:proof_dml_master} with
\begin{align*}
h_1^*(r) & = f(r;\eta^*)\zeta(x;\eta^*), \quad \hat{h}_1^{(-k)}(r) = f\left(r;\hat{\eta}^{(-k)}\right)\zeta\left(x;\hat{\eta}^{(-k)}\right) \\
h_2^*(x) & = h_2(x;\eta^*), \quad \hat{h}_2^{(-k)}(x) = h_2\left(x;\hat{\eta}^{(-k)}\right).
\end{align*}

\noindent
\textbf{Boundedness:} We have $V^*(s,z,x) \leq 1$ trivially since $Y \in [0,1]$.
Furthermore by the overlap condition (Assumption~\ref{assump:overlap}) 
it is clear that $f(r;\eta^*)$
is uniformly upper bounded for $P^*$-almost all $r$, 
a qualifier we omit hereafter in all inequalities and equalities in this proof,
in absolute value by $C\delta^{-2}$.
Similarly, by the overlap in infinite-dimensional nuisance estimates condition in Assumption~\ref{assump:regularity},
we have $f\left(r;\hat{\eta}^{(-k)}\right)$ is uniformly upper bounded in absolute value by $C\delta^{-2}$ for all $k=1,\ldots,K$.
With $\max\left(|\zeta(x;\eta^*)|,\Big|\zeta\left(x;\hat{\eta}^{(-k)}\right)\Big|\right) \leq C$ for all $k$,
we conclude that $\max\left(\|h_1^*\|_{\infty,P^*},\left\|\hat{h}_1^{(-k)}\right\|_{\infty,P^*}\right) \leq C^2\delta^{-2}$. 
\end{proof}

\noindent
\textbf{Approximate tangency:}
Fix $k \in \{1,\ldots,K\}$
and consider the one-dimensional outcome mean function parametric submodel
\begin{align*}
\tilde{m}(s,z,x;t) & = m^*(s,z,x) + t\left(\hat{m}^{(-k)}(s,z,x)-m^*(s,z,x)\right), \quad (s,z) \neq (1,1)  \\
\tilde{m}(1,1,x;t) & = \ell^{-1}(\ell(\tilde{m}(1,0,x;t))+\ell(\tilde{m}(0,1,x;t))-\ell(\tilde{m}(0,0,x;t)))
\end{align*}
indexed by $t$ in an open subset of $\real$ containing $[0,1]$,
where we note that $\tilde{m}(s,z,x;1) = \hat{m}^{(-k)}(s,z,x)$ for all $(s,z)$ (including for $(s,z)=(1,1)$ by~\eqref{eq:m_hat_respect_selection}).
The overlap and nuisance in overlap conditions ensure that this is a well-defined submodel in the sense that it does not violate the restriction $\epsilon < m < 1-\epsilon$ in the definition of $\cm_5$ in~\eqref{eq:outcome_selection_bias}.

By Taylor's theorem,
for each $x$
there exists $\tilde{t}=\tilde{t}(x) \in [0,1]$ for which
\begin{align*}
\hat{m}^{(-k)}(1,1,x) -m^*(1,1,x) & = \tilde{m}(1,1,x;1)-\tilde{m}(1,1,x;0) \\
&= \frac{\partial \tilde{m}(1,1,x;t)}{\partial t} \Big|_{t=0} + \frac{1}{2} \frac{\partial^2 \tilde{m}(1,1,x;t)}{\partial t^2} \Big|_{t=\tilde{t}(x)}
\end{align*}
which implies that for all $(s,z,x)$ we have
\[
\hat{m}^{(-k)}(s,z,x)-m^*(s,z,x)-\hat{R}^{(-k)}(s,z,x) = \frac{\partial \tilde{m}(s,z,x;t)}{\partial t} \Big|_{t=0} 
\]
for
\[
\hat{R}^{(-k)}(s,z,x) = \frac{sz}{2} \frac{\partial^2 \tilde{m}(1,1,x;t)}{\partial t^2} \Big|_{t=\tilde{t}(x)}.
\]
since evidently
\[
\frac{\partial \tilde{m}(s,z,x;t)}{\partial t} = m^{(-k)}(s,z,x) - m^*(s,z,x)
\]
for $(s,z) \neq (1,1)$.
Hence $\hat{m}^{(-k)}-m^*-\hat{R}^{(-k)} \in \cs_{\cm_1}$ and
it remains to show the necessary rate conditions on $\hat{R}^{(-k)}$.

By the chain rule
\begin{align*}
\frac{\partial \tilde{m}(1,1,x;t)}{\partial t} & = \frac{ \sum_{(s,z) \neq (1,1)} \ell'(\tilde{m}(s,z,x;t)) \left(\hat{m}^{(-k)}_{sz}(x)-m^*_{sz}(x)\right)}{\ell'(\tilde{m}(1,1,x;t))}.
\end{align*}
and then by the quotient rule
\begin{align*}
\frac{\partial^2 \tilde{m}(1,1,x;t)}{\partial t^2} & = \frac{\sum_{(s,z) \neq (1,1)} \ell''(\tilde{m}(s,z,x;t))\left(\hat{m}^{(-k)}_{sz}(x)-m^*_{sz}(x)\right)^2  - \ell''(\tilde{m}(1,1,x;t))\left(\frac{\partial \tilde{m}(1,1,x;t)}{\partial t}\right)^2}{\ell'(\tilde{m}(1,1,x;t))}.
\end{align*}
We also note that for each $(s,z), (s',z') \in \{0,1\}^2$, we have
\begin{equation}
\label{eq:m_error_product}
\frac{1}{|\ci_k|} \sum_{i \in \ci_k} S_iZ_i  \left|\hat{m}_{sz}^{(-k)}(X_i)-m_{sz}^*(X_i)\right|\left|\hat{m}_{s'z'}^{(-k)}(X_i)-m_{s'z'}^*(X_i)\right| = o_p(N^{-1/2})
\end{equation}
by the Cauchy-Schwarz inequality,
the rate condition $\left\|\hat{m}^{(-k)}-m^*\right\|_{2,P^*} = o_p(N^{-1/4})$,
and Lemma~\ref{lemma:conditional_oh_pee_fn_f}.

Note $\ell'$ and $\ell''$ are uniformly bounded, i.e. there exist $0 < c < C < \infty$ for which
\begin{align}
\label{eq:bounded_derivatives}
c \leq \ell'(\tilde{m}(s,z,x;t)) \leq C, \quad |\ell''(\tilde{m}(s,z,x;t))| \leq C
\end{align}
for all $(s,z,x)$, $t \in [0,1]$.
Hence
\begin{align*}
& \frac{1}{|\ci_k|} \sum_{i \in \ci_k} \Big|\hat{R}^{(-k)}(R_i) \Big| \\
& \quad \leq \frac{C}{2c|\ci_k|} \sum_{i \in \ci_k} S_iZ_i  \left[\sup_{t \in [0,1]} \left(\frac{\partial \tilde{m}(1,1,X_i;t)}{\partial t}\right)^2 + \sum_{(s,z) \neq (1,1)} \left(\hat{m}_{sz}^{(-k)}(X_i)-m^*_{sz}(X_i)\right)^2\right].
\end{align*}
But
\begin{align*}
& \frac{1}{|\ci_k|} \sum_{i \in \ci_k} S_iZ_i \sup_{t \in [0,1]} \left(\frac{\partial \tilde{m}(1,1,X_i;t)}{\partial t}\right)^2 \\
& \quad \leq \frac{1}{c^2|\ci_k|} \sum_{i \in \ci_k} S_iZ_i \sup_{t \in [0,1]} \left(\sum_{(s,z) \neq (1,1)} \ell'(\tilde{m}(s,z,X_i;t))\left(\hat{m}_{sz}^{(-k)}(X_i)-m_{sz}^*(X_i)\right) \right)^2 \\
& = \frac{1}{c^2|\ci_k|} \sum_{i \in \ci_k} S_iZ_i  \sup_{t \in [0,1]} \sum_{(s,z) \neq (1,1)} \sum_{(s',z') \neq (1,1)} \ell'(\tilde{m}(s,z,X_i;t))\ell'(\tilde{m}(s',z',X_i;t))\left(\hat{m}_{sz}^{(-k)}(X_i)-m_{sz}^*(X_i)\right) \times \\
& \qquad \left(\hat{m}_{s'z'}^{(-k)}(X_i)-m_{s'z'}^*(X_i)\right) \\
& \leq \frac{C^2}{c^2|\ci_k|} \sum_{i \in \ci_k} \sum_{(s,z) \neq (1,1)} \sum_{(s',z') \neq (1,1)}  S_iZ_i|\hat{m}_{sz}^{(-k)}(X_i)-m_{sz}^*(X_i)||\hat{m}_{s'z'}^{(-k)}(X_i)-m_{s'z'}^*(X_i)| \\
& = o_p(N^{-1/2}) \text{ by~\eqref{eq:m_error_product}}
\end{align*}
while
\[
\frac{1}{|\ci_k|} \sum_{i \in \ci_k} \sum_{(s,z) \neq (1,1)}  S_iZ_i \left(\hat{m}_{sz}^{(-k)}(X_i)-m_{sz}^*(X_i)\right)^2 = o_p(N^{-1/2}) \text{ by~\eqref{eq:m_error_product}.}
\]
This completes the proof of the required rate conditions on $\hat{R}^{(-k)}$. \\

\noindent
\textbf{Rate conditions:}
Fixing $k \in \{1,\ldots,K\}$,
we have $\|\hat{m}^{(-k)}-m^*\|_{2,P^*}=o_p(N^{-1/4})$ by~\eqref{eq:nuisance_rates}
and~\eqref{eq:m_hat_respect_selection}.
With $\|\hat{h}_2^{(-k)}-h_2^*\|_{2,P^*} = o_p(1)$ by assumption,
to show the rate conditions it suffices to show that $\|\hat{h}_1^{(-k)}-h_1^*\|_{2,P^*} = o_p(N^{-1/4})$.
To that end,
we write
\[
\hat{h}_1^{(-k)}(r)-h_1^*(r) = f\left(r;\hat{\eta}^{(-k)}\right)\zeta\left(x;\hat{\eta}^{(-k)}\right) - f(r;\eta^*)\zeta(x;\eta^*)
\]
and note that since $|f(r;\eta^*)| \leq C$ as shown above
and we have $\max\left(|\zeta(x;\eta^*)|,\Big|\zeta\left(x;\hat{\eta}^{(-k)}\right)\Big|\right) \leq C$ and $\|\zeta(x;\hat{\eta}^{(-k)})-\zeta(x;\eta^*)\|_{2,P^*} = o_p(N^{-1/4})$ by assumption,
it suffices by Lemma~\ref{lemma:multiplication} to show that
\begin{equation}
\label{eq:f_rate}
\int \left(f\left(r;\hat{\eta}^{(-k)}\right)-f(r;\eta^*)\right)^2 dP^*(w) = o_p(N^{-1/2}).
\end{equation}
We compute
\begin{align*}
& \left(f\left(s,z,x;\hat{\eta}^{(-k)}\right)-f(s,z,x;\eta^*)\right)^2 \\
& = sz\left(\ell\left(\hat{m}_{11}^{(-k)}(x)\right)-\ell(m_{11}^*(x))\right)^2 \\
& \quad + s(1-z)\left(\frac{\hat{e}^{(-k)}(x)\ell'\left(\hat{m}_{10}^{(-k)}(x)\right)}{1-\hat{e}^{(-k)}(x)}-\frac{e^*(x)\ell'(m_{10}^*(x))}{1-e^*(x)}\right)^2  \nonumber \\
& \quad + (1-s)z\left(\frac{\hat{p}^{(-k)}(x)\hat{e}^{(-k)}(x)\ell'\left(\hat{m}_{01}^{(-k)}(x)\right)}{(1-\hat{p}^{(-k)}(x))\hat{q}^{(-k)}(x)}-\frac{p^*(x)e^*(x)\ell'(m_{01}^*(x))}{(1-p^*(x))q^*(x)}\right)^2 \nonumber \\
& \quad + (1-s)(1-z)\left(\frac{\hat{p}^{(-k)}(x)\hat{e}^{(-k)}(x)\ell'\left(\hat{m}_{00}^{(-k)}(x)\right)}{(1-\hat{p}^{(-k)}(x))(1-\hat{q}^{(-k)}(x))}-\frac{p^*(x)e^*(x)\ell'(m_{00}^*(x))}{(1-p^*(x))(1-q^*(x))}\right)^2. \nonumber
\end{align*}

By~\eqref{eq:nuisance_rates} and~\eqref{eq:bounded_derivatives}
\begin{equation}
\label{eq:l_prime_rates}
\|\ell' \circ \hat{m}_{sz}^{(-k)} - \ell' \circ m^*\|_{2,P^*} = o_p(N^{-1/4}), \quad \forall (s,z) \in \{0,1\}^2
\end{equation}
Taking $(s,z)=(1,1)$ above shows
\[
\int  sz\left(\ell\left(\hat{m}_{11}^{(-k)}(x)\right)-\ell(m_{11}^*(x))\right)^2 dP^*(w) = o_p(N^{-1/2}).
\]
For the remaining three terms,
by Assumption~\ref{assump:overlap},
overlap in infinite-dimensional nuisance estimates,
and~\eqref{eq:nuisance_rates},
we know by Lemma~\ref{lemma:multiplication} that
\[
\int \left(\frac{\hat{e}^{(-k)}(x)}{1-\hat{e}^{(-k)}(x)}-\frac{e^*(x)}{1-e^*(x)}\right)^2 dP^*(w) = \int \left(\frac{\hat{e}^{(-k)}(x)-e^*(x)}{(1-e^*(x))(1-\hat{e}^{(-k)}(x))}\right)^2 = o_p(N^{-1/2}).
\]
Similarly
\[
\int \left(\frac{\hat{p}^{(-k)}(x)}{1-\hat{p}^{(-k)}(x)}-\frac{p^*(x)}{1-p^*(x)}\right)^2 dP^*(w) = \int \left(\frac{\hat{p}^{(-k)}(x)-p^*(x)}{(1-p^*(x))(1-\hat{p}^{(-k)}(x))}\right)^2 dP^*(w) = o_p(N^{-1/2}).
\]
and 
\[
\int \left(\frac{1}{1-\hat{r}^{(-k)}(x)}-\frac{1}{1-r^*(x)}\right)^2 dP^*(w) + \int \left(\frac{1}{\hat{r}^{(-k)}(x)}-\frac{1}{r^*(x)}\right)^2 dP^*(w)= o_p(N^{-1/2}).
\]
Then applying Lemma~\ref{lemma:multiplication} recursively we conclude
\begin{align*}
\int s(1-z)\left(\frac{\hat{e}^{(-k)}(x)\ell'\left(\hat{m}_{10}^{(-k)}(x)\right)}{1-\hat{e}^{(-k)}(x)}-\frac{e^*(x)\ell'(m_{10}^*(x))}{1-e^*(x)}\right)^2 dP^*(w) & = o_p(N^{-1/2}) \\
\int (1-s)z\left(\frac{\hat{p}^{(-k)}(x)\hat{e}^{(-k)}(x)\ell'\left(\hat{m}_{01}^{(-k)}(x)\right)}{(1-\hat{p}^{(-k)}(x))\hat{r}^{(-k)}(x)}-\frac{p^*(x)e^*(x)\ell'(m_{01}^*(x))}{(1-p^*(x))r^*(x)}\right)^2 dP^*(w) & = o_p(N^{-1/2}), \quad \text{ and} \\
\int (1-s)(1-z)\left(\frac{\hat{p}^{(-k)}(x)\hat{e}^{(-k)}(x)\ell'\left(\hat{m}_{00}^{(-k)}(x)\right)}{(1-\hat{p}^{(-k)}(x))(1-\hat{r}^{(-k)}(x))}-\frac{p^*(x)e^*(x)\ell'(m_{00}^*(x))}{(1-p^*(x))(1-r^*(x))}\right)^2 dP^*(w) & = o_p(N^{-1/2})
\end{align*}
as desired.

\section{Details of numerical simulations in Section~\ref{sec:simulations_data_fusion}}
\label{app:sim_details}
Here we specify the full data generating processes for the numerical simulations in Section~\ref{sec:simulations_data_fusion}.


For the simulations in Section~\ref{sec:simulations_data_fusion},
in the discrete scenario,
we emulate the data generating process in Section 5 of~\citet{guo2022multi}:
\begin{align*}
X=(X_1,X_2) & \sim \bern(0.5) \times \bern(0.5) \\
Z \mid X & \sim \bern(X_1-X_2) \\
Y \mid X,Z & \sim \bern\left(\expit\left(-0.5+Z+(1-2Z)(X_1-X_2)\right)\right).
\end{align*}
For the continuous scenario, we have
\begin{align*}
X=(X_1,X_2) & \sim \mathcal{\dunif}(-1,1) \times \mathcal{\dunif}(-1,1) \\
Z \mid X & \sim \bern(X_1-X_2) \\
Y \mid X,Z & \sim \bern\left(\expit\left(-0.5+Z+(1-2Z)(X_1-X_2)+(1.5Z-1)X_1X_2\right)\right).
\end{align*}
To generate the simulated observational dataset,
we draw observations from these DGP's with the selection process described in Section~\ref{sec:simulations_data_fusion} until 3000 selected observations are obtained.
Then additional observations from the above DGP's are generated without selection to form the RCT.

\section{Implementation detail for control variate estimator}
\label{app:cv_detail}
As indicated in the main text,
our implementation of the control variate estimator $\hat{\tau}_{\cv}$
follows the guidance of~\citet{guo2022multi}.
In the discrete scenario simulation,
the differences $\{OR_1(x)-OR_0(x) \mid x \in \{0,1\}^2\}$ are used as four control variates.
In the continuous scenario simulation and the Spambase data example,
we use the average of $\log(OR_1(x))-\log(OR_0(x))$ over 50 randomly chosen distinct values of the covariates $x$ in the combined RCT and observational datasets as a single control variate.
As the odds ratios are determinstic transformations of the outcome mean function $m$,
we can estimate all of these control variates
via deterministic transformations of an outcome mean function estimator $\hat{m}$.
We use the same model for $\hat{m}$ as the cross-fit outcome mean function estimates $\hat{m}^{(-k)}$ used for computing the baseline estimator $\hat{\tau}_{\ba}$
(i.e. the +1/+2 estimator in the discrete simulation scenario,
MARS in the continuous simulation scenario
)
except that we refit the model to the entire dataset
(i.e. avoiding cross-fitting)
as cross-fitting is not known to improve estimation of mean functions at a single point.

We can write
\[
\hat{\tau}_{\cv} = \hat{\tau}_{\ba} - \hat{\Gamma}\hat{\lambda}
\]
where $\hat{\tau}_{\ba}$ is the initial baseline estimator
(AIPW or AIPSW in all simulations and the data example),
$\hat{\lambda}$ is the estimate of the control variate(s) viewed as a vector,
and $\hat{\Gamma}$ is an estimate of the optimal adjustment factor $\cov(\hat{\lambda},\hat{\tau}_{\ba}) \cov(\hat{\lambda})^{-1} $.
The same adjustment $\hat{\Gamma}$ is used for all Monte Carlo simulations in each simulation scenario.
This adjustment is computed by calculating $\hat{\tau}_{\ba}$ and $\hat{\lambda}$ in each of $B=1000$ Monte Carlo replicates of the data-generating process of the scenario,
independent of the Monte Carlo replicates from which we report results,
and then extracting the appropriate components of the sample covariance matrix.
Similarly, for the Spambase data example,
we compute $\hat{\Gamma}$ using the sample covariance matrix of estimates $\hat{\tau}_{\ba}$ and $\hat{\lambda}$ computed on  $B=1000$ bootstrapped replications of the dataset,
independent of the bootstrap replications on which we report results.
This is exactly what is done by~\citet{guo2022multi}.
We remark that in practice,
an investigator will not have access to additional, 
independent data on which to compute a regression adjustment independently of the data,
and even if they did,
they would want to use it for causal estimation to gain full efficiency.
Thus,
we believe a more accurate estimate of the performance of $\hat{\tau}_{\cv}$ would be reflected by computing a different adjustment factor $\hat{\Gamma}$ each time $\hat{\tau}_{\cv}$ is computed.
We presume this is not done in~\citet{guo2022multi}
due to the computation time
required
(each computation of $\hat{\tau}_{\ba}$ requires estimating nuisance functions,
so estimating separate adjustments $\hat{\Gamma}$ would require fitting $B$ times as many nuisance function estimates).

\bibliography{combine_obs_rct}
\bibliographystyle{abbrvnat}

\end{document}